\documentclass[11pt,a4paper]{article}
\usepackage[margin=25mm]{geometry}
\usepackage[T1]{fontenc}\usepackage{lmodern,microtype,bm,amsmath,amssymb,amsthm,mathtools,booktabs,array,longtable}
\usepackage[round,authoryear]{natbib}\usepackage[hidelinks]{hyperref}\usepackage{placeins}
\newtheorem{theorem}{Theorem}\newtheorem{proposition}{Proposition}\newtheorem{corollary}{Corollary}
\theoremstyle{remark}\newtheorem*{tuningremark}{Remark}
\newcommand{\E}{\mathbb E}\newcommand{\Var}{\operatorname{Var}}\newcommand{\Cov}{\operatorname{Cov}}\newcommand{\diag}{\operatorname{diag}}\newcommand{\tr}{\operatorname{tr}}\newcommand{\op}{\mathrm{op}}\newcommand{\norm}[1]{\lVert #1\rVert}
\newcommand{\Pp}{\mathbb P_p}

\hypersetup{pdftitle={Conditional validity of prediction-powered inference for spatially dependent finite populations},pdfauthor={Shinichiro Shirota},pdfsubject={Finite-population regression PPI}}
\title{Conditional validity of prediction-powered inference for spatially dependent finite populations}
\author{Shinichiro Shirota\\\small Hitotsubashi University}
\date{8 October 2026}
\begin{document}\maketitle
\begin{abstract}
Prediction-powered inference (PPI) combines full-frame predictions with a small labelled sample to estimate finite-population regression coefficients. We establish conditions for the conditional validity of existing calibrated residual-regression PPI estimators under nonuniform label acquisition and spatially dependent prediction errors. Approximate calibration can leave a population-specific displacement: the projection of prediction errors onto coefficient-specific calibration defects. We characterize its role in conditional inference through the noncentrality of a fixed-candidate Wald region. The relevant target-weight function is the coefficient loading divided by acquisition intensity. When mean-approximation error does not dominate, directional covariance control yields a wider sufficient recovery regime than a total residual-energy bound. A bounded spectral-basis theorem gives conditional coverage under directional covariance, tail and nondegeneracy conditions, with Gaussian and bounded non-Gaussian spatial examples. The analysis accounts for estimation of the one-step augmentation and permits same-label selection from a fixed finite library with a common first-order influence vector. A periodic Fourier construction attains the stochastic directional-gap order within a restricted dictionary class. Constructed and real-record-set experiments illustrate why mean calibration need not protect a slope. The results concern descriptive census projection coefficients under the stated sampling and field conditions.
\end{abstract}
\noindent\textit{Keywords:} prediction-powered inference; spatial association; finite-population regression; target calibration; conditional coverage.

\section{Introduction}
Remote-sensing maps provide predictions for locations at which high-quality reference measurements may be unavailable or expensive to obtain. Using these maps to study associations with environmental or geographical covariates raises a measurement problem: regression coefficients based on predictions need not coincide with those based on reference outcomes. Regression inference using remote-sensing maps has been studied in prediction-powered settings \citep{lu}. We consider the least-squares coefficients that would be obtained from reference outcomes over a specified finite population of locations. These describe census associations rather than causal or superpopulation parameters.

Combining full-population auxiliary information with outcomes observed on a smaller sample has a long history in finite-population inference, including generalized-difference estimation \citep{cassel}, calibration weighting \citep{deville} and model calibration \citep{wu}. Model-assisted extensions use flexible predictors \citep{breidt}. Prediction-powered inference (PPI) similarly combines predictions with limited reference measurements \citep{ppi}; its connections to survey difference estimation have been discussed by \citet{mozer} and \citet{denis}. With known inclusion probabilities, design-weighted residual correction provides a design-unbiased benchmark for finite-population linear functionals, irrespective of spatial dependence.

Nonuniform acquisition of reference measurements poses a further challenge. PPI with imputed measurements and nonuniform sampling is studied by \citet{kluger}. Penalized and soft calibration incorporate auxiliary information under imbalance or selection bias \citep{guggemos,gao}. Coefficient-specific balancing and augmented regression representations also build on established work \citep{aml,bruns}, while \citet{ji} study prediction recalibration for efficiency. These estimation and regression--balancing representations do not by themselves establish when displacement is negligible for repeated label sampling conditional on an entire realised, spatially dependent population. We address that coverage question for existing residual-regression PPI under nonuniform acquisition.

The inferential guarantee also depends on what is held fixed. Design-based inference conditions on population outcomes and treats sample selection as random, whereas \citet{jin} study inference conditional on population attributes. Spatial dependence and incomplete reference labels are addressed by \citet{salerno}, who investigate spatially robust inference with missing-at-random labels and cross-fitted nuisance estimates. Related model-assisted work considers machine-learning-assisted estimation for survey data \citep{an}. Prediction-powered smoothing and estimator validation for finite-population domain means provide a related but distinct perspective \citep{kawano}.

These studies motivate a specific conditional-inference question. Suppose that full-frame predictions and auxiliary spatial information are available, but reference labels are acquired nonuniformly and the feasible PPI correction does not use the true inclusion probabilities. For an existing coefficient-targeted calibrated residual-regression estimator, when does repeated label acquisition yield nominal coverage conditional on the complete realised finite population? A consistently estimated sampling covariance need not suffice: imperfect calibration can leave a population-specific displacement from the census regression target. We characterize the effect of that displacement on conditional coverage and give covariance, tail and fitting-complexity conditions under which it vanishes at the sampling-error scale. The residual-field law describes typical populations; it does not change the conditional sampling target.

The central quantity is the population displacement
\[
\bm{G}=N^{-1}\mathbf{D}^T\bm{\Delta},
\]
where \(\bm{\Delta}\) is the prediction-error field and the columns of \(\mathbf D\) are coefficient-specific calibration defects. For a prespecified correction, the displacement relative to its sampling covariance determines the noncentrality of the limiting Wald statistic. For a prespecified contrast, conditional validity therefore requires its projected displacement to be negligible on the sampling-error scale, in addition to sampling normality and covariance consistency; full-vector validity requires the entire displacement to be negligible at that scale. The relevant target-weight function for coefficient \(j\) is \(\bm g_j=\bm\ell_{\cdot j}/\bm\lambda\), where \(\bm\ell_{\cdot j}\) is its vector of full-population regression loadings and \(\bm\lambda\) is the vector of relative acquisition intensities; division is componentwise. This formulation distinguishes calibration for a regression coefficient from calibration for a population mean and identifies the directions in which residual dependence matters.

Our analysis makes three contributions. First, for a prespecified correction, Proposition~\ref{pr:gap-criterion} characterizes conditional Wald coverage through the realised displacement and its noncentrality, using standard noncentral chi-square asymptotics. Second, Theorem~\ref{th:main} gives sufficient recovery conditions based on directional covariance, tail regularity and nondegeneracy. When mean error does not dominate, they certify a wider rate regime than a total residual-energy bound for the same estimator. A periodic construction attains the stochastic slope-gap order within a restricted dictionary class. Third, an exact expansion of the implemented one-step correction, retaining fitted-augmentation variation, establishes consistency of its full inclusion-derivative covariance estimator. Same-label selection is asymptotically harmless for a fixed predeclared library with coalescing influence functions; this asserts neither finite-sample improvement nor unrestricted tuning guarantees.

Gaussian spectral models on periodic Fourier and nonperiodic cosine grids remain explicit special cases; the latter accommodate ordinary centred coordinate slopes. Additional periodic examples allow non-diagonal Gaussian covariance and binary responses. The covariance, tail and nondegeneracy assumptions remain essential; these examples do not certify arbitrary irregular spatial frames.

Numerical experiments illustrate the distinction between reproducing a population mean and protecting a regression slope. A constructed spatial mechanism isolates a coefficient-specific calibration failure, while a separate fixed-dictionary experiment examines how increasing the label budget can expose a persistent gap. A paired LUCAS--WorldCover analysis evaluates descriptive coordinate regressions on a fixed observed-site frame. The comparisons show that coefficient-specific calibration can substantially change slope coverage at a fixed auxiliary dimension, although simpler methods remain competitive and one-step augmentation is not uniformly beneficial. These experiments illustrate the theoretical mechanism; they do not verify the spectral assumptions of the recovery theorem on the irregular LUCAS frame.

\section{The regression target and prediction-powered correction}
We define the finite-population target and label-acquisition law, then formulate existing calibrated residual-regression estimators as coefficient-specific PPI \citep{ppi,lu}. These estimators are the objects of the conditional-validity analysis in Section 3; their definition does not require a spatial model.

\subsection{Finite-population target and label acquisition}\label{sec:target-acquisition}
Roman matrices use bold upright capitals and vectors are bold italic; Greek matrices use bold capitals. Scalars and individual entries are unbolded, retaining conventional capitals such as $N,K,Y_i$ and $R_i$. Sets, spaces and operators retain their stated notation. Products and quotients of frame vectors are componentwise unless a matrix product is indicated.

Let $U=\{1,\ldots,N\}$ be a fixed finite population. Unit $i$ has scientific regressors ${\bm{x}}_i\in\mathbb R^p$, reference response $Y_i\in\mathbb R$ and prediction $f_i\in\mathbb R$, with $p$ fixed. Regressors and predictions are observed throughout the frame. Define the stacked objects
\[
 {\mathbf{X}}=({\bm{x}}_1,\ldots,{\bm{x}}_N)^T\in\mathbb R^{N\times p},\qquad
 {\bm{Y}}=(Y_1,\ldots,Y_N)^T,\qquad {\bm{f}}=(f_1,\ldots,f_N)^T.
\]
Assume that ${\mathbf{J}}=N^{-1}{\mathbf{X}}^T{\mathbf{X}}$ is nonsingular. The finite-population least-squares target and its known loadings are
\begin{equation}\label{eq:target}
 {\mathbf{J}}=N^{-1}{\mathbf{X}}^T{\mathbf{X}},\qquad {\bm{\beta}}_U={\mathbf{J}}^{-1}N^{-1}{\mathbf{X}}^T{\bm{Y}},\qquad {\mathbf{L}}={\mathbf{X}}{\mathbf{J}}^{-1},\quad {\bm{\ell}}_i={\mathbf{J}}^{-1}{\bm{x}}_i.
\end{equation}
The $i$th row of ${\mathbf{L}}$ is ${\bm{\ell}}_i^T$; write $\ell_{ij}$ for its $j$th entry. The scientific design ${\mathbf{X}}$ defines the estimand. Changing ${\mathbf{X}}$, the population frame, or the measurement protocol generally changes ${\bm{\beta}}_U$.

For ${\bm{\Delta}}={\bm{Y}}-{\bm{f}}$ the exact PPI decomposition is
\begin{equation}\label{eq:ppi}
 {\bm{\beta}}_U=\underbrace{N^{-1}{\mathbf{L}}^T{\bm{f}}}_{{\bm{\beta}}_f\text{, known}}
             +N^{-1}{\mathbf{L}}^T{\bm{\Delta}}.
\end{equation}
Thus ${\bm{\ell}}_i$ maps the response contribution at unit $i$ into the coefficient vector: ${\bm{\beta}}_U=N^{-1}\sum_i{\bm{\ell}}_iY_i$. The target ${\bm{\beta}}_U$ describes the least-squares association over the complete finite population; its definition requires no data-generating model. When ${\mathbf{X}}$ contains centered regressors and an intercept, the intercept is the frame mean. A causal interpretation would require additional assumptions.

The coefficient loadings also make the adjustment implicit in a slope explicit. For a scientific slope $j$, let ${\bm{x}}_{\cdot j}$ be column $j$ of ${\mathbf{X}}$ and ${\mathbf{X}}_{-j}$ the matrix of its remaining columns. Residualize ${\bm{x}}_{\cdot j}$ over the full frame: ${\bm{v}}_j={\bm{x}}_{\cdot j}-{\mathbf{X}}_{-j}({\mathbf{X}}_{-j}^T{\mathbf{X}}_{-j})^{-1}{\mathbf{X}}_{-j}^T{\bm{x}}_{\cdot j}$, with entries $v_{ij}$. Then
\begin{equation}\label{eq:fwl}
 \ell_{ij}=\frac{v_{ij}}{N^{-1}\sum_hv_{hj}^2}.
\end{equation}
Thus the elevation coefficient, for example, is loaded by elevation after adjustment for the other scientific regressors, not elevation alone. This standard partial-regression identity identifies the direction that a coefficient-specific correction must protect. A nearly zero denominator makes the coefficient ill-conditioned. PPI does not repair ill-conditioning already present in the scientific design.

Predictions and scientific covariates are fixed before the evaluation labels are acquired. Fitting the residual correction on those labels is covered below; refitting ${\bm{f}}$ or choosing ${\mathbf{X}}$ using them is not silently included. Under the model law used for gap calculations, the residual mean and covariance are conditional on the entire fixed observed frame, including any prediction values used to build the correction.

We next specify the randomness underlying inference for this fixed-population target. Let $\mathcal F_N$ denote the fixed observed frame, including ${\mathbf{X}}$, ${\bm{f}}$, and the candidate dictionaries and penalties. Conditional on $\mathcal F_N$, the law $P_\xi$ generates ${\bm{\Delta}}$. Write $\mathcal U_N=\mathcal F_N\vee\sigma({\bm{\Delta}})$ for the complete finite population. Conditional on $\mathcal U_N$, let $R_i$ be independent Bernoulli with
\begin{equation}\label{eq:acquisition}
 \pi_i=(n/N)\lambda_i,\quad N^{-1}\sum_i\lambda_i=1,
 \quad 0<\lambda_-\le\lambda_i\le\lambda_+<\infty.
\end{equation}
Define the acquired sample by $S=\{i\in U:R_i=1\}$; the response $Y_i$ is observed only when $R_i=1$. Here $n=\sum_i\pi_i$ is the expected high-quality label count, not the prediction-training size; $|S|$ is random. We study $n\to\infty$, $n/N\to0$. The true acquisition intensities are fixed parameters of the sampling law, not inputs required by the feasible estimators. They do not depend on the latent error realization. Known $\pi_i$ give the exactly design-unbiased Horvitz--Thompson (HT) benchmark
\[
 \widehat{\bm{\beta}}_{\rm HT}={\bm{\beta}}_f+N^{-1}\sum_iR_i{\bm{\ell}}_i\Delta_i/\pi_i,
\]
whose conditional covariance is $N^{-2}\sum_i(1-\pi_i){\bm{\ell}}_i{\bm{\ell}}_i^T\Delta_i^2/\pi_i$. Its interval additionally requires sampling normality and variance consistency. The following feasible corrections do not use the true intensities.

\subsection{Calibrated residual regression}\label{sec:calibrated-regression}
Let ${\mathbf{B}}=({\bm{b}}_1,\ldots,{\bm{b}}_N)^T\in\mathbb R^{N\times K}$ be an auxiliary dictionary of dimension $K$, with ${\bm{b}}_i\in\mathbb R^K$ and constants in its column span. For example, the columns of ${\mathbf{B}}$ may be a constant and smooth functions of spatial coordinates, evaluated at all frame locations, to capture geographical patterns in the prediction residuals. The scientific design ${\mathbf{X}}$ defines the estimand, whereas ${\mathbf{B}}$ defines only the residual correction: changing ${\mathbf{B}}$ does not change ${\bm{\beta}}_U$. The two sets of columns may overlap, but need not coincide. For any full-frame vector or matrix, a subscript $S$ denotes restriction to the sampled units, so ${\mathbf{B}}_S$ has $|S|$ rows and ${\bm{\Delta}}_S$ has $|S|$ entries.

For theory normalize ${\mathbf{B}}^T{\mathbf{B}}/N={\mathbf{I}}_K$. Let ${\bm{\Lambda}}\in\mathbb R^{K\times K}$ be a symmetric positive semidefinite normalized penalty, fixed before label acquisition, with the absolute penalty defined below. The unpenalized dictionary span is $\{{\mathbf{B}}{\bm{\gamma}}:{\bm{\gamma}}\in\ker({\bm{\Lambda}})\}$. We call these directions protected; the theory below includes constants in this span. For nonsingular ${\mathbf{T}}$, set
\begin{equation}\label{eq:C}
 {\mathbf{C}}_B=N^{-1}{\mathbf{B}}^T{\mathbf{L}},\quad {\mathbf{T}}={\mathbf{B}}_S^T{\mathbf{B}}_S+{\bm{\Gamma}},\quad
 \widehat{\bm{\alpha}}={\mathbf{T}}^{-1}{\mathbf{B}}_S^T{\bm{\Delta}}_S,\quad{\bm{\Gamma}}=n{\bm{\Lambda}}\succeq0.
\end{equation}
The deterministic absolute penalty ${\bm{\Gamma}}$ is held fixed in inclusion perturbations and deletion calculations. Using $|S|{\bm{\Lambda}}$ instead changes the derivative. Soft calibration and best linear unbiased predictor (BLUP) augmentation motivate the following form \citep{gao}. The uncorrected and one-step-corrected estimates are
\begin{equation}\label{eq:est}
 \widehat{\bm{\beta}}_0={\bm{\beta}}_f+{\mathbf{C}}_B^T\widehat{\bm{\alpha}},\qquad
 \widehat{\bm{\beta}}_+={\bm{\beta}}_f+{\mathbf{C}}_B^T\widehat{\bm{\alpha}}+{\mathbf{C}}_B^T{\mathbf{T}}^{-1}{\bm{\Gamma}}\widehat{\bm{\alpha}}.
\end{equation}
With ${\bm{\Gamma}}=0$, the coefficient weights ${\mathbf{W}}={\mathbf{B}}_S({\mathbf{B}}_S^T{\mathbf{B}}_S)^{-1}{\mathbf{C}}_B$ satisfy ${\mathbf{B}}_S^T{\mathbf{W}}={\mathbf{C}}_B$. If ${\bm{\Delta}}={\mathbf{B}}{\bm{\gamma}}$ and ${\bm{\Gamma}}{\bm{\gamma}}=0$, both estimates reproduce ${\bm{\beta}}_U$ exactly on every nonsingular sample. This is coefficient-specific balancing for a linear functional, as in the Riesz formulation of \citet{aml}; the reproduction identity is proved in Supplement S1.

In contrast, define $\bar {\bm{b}}_U=N^{-1}\sum_{i\in U}{\bm{b}}_i$. Common mean weights $w_i$ satisfying $\sum_{i\in S}w_i{\bm{b}}_i=\bar {\bm{b}}_U$ have $\sum_{i\in S}w_i=1$, since ${\mathbf{B}}$ contains constants. They yield $\widehat{\bm{\beta}}_{\rm com}={\bm{\beta}}_f+\sum_Sw_i{\bm{\ell}}_i\Delta_i$, which need not balance the coefficient-weighted dictionary products $\ell_{ij}b_{ih}$. This common-weight estimator is not inherently invalid and can work under accurate acquisition weights or adequate score approximation. The experiments retain both simple common weights and targeted residual regression.

The same correction can be expressed directly as a regression on the full frame. Define $\widehat{\bm{\Delta}}_+={\mathbf{B}}\{\widehat{\bm{\alpha}}+{\mathbf{T}}^{-1}{\bm{\Gamma}}\widehat{\bm{\alpha}}\}$. Then
\begin{equation}\label{eq:prediction-equivalence}
 \widehat{\bm{\beta}}_+={\mathbf{J}}^{-1}N^{-1}{\mathbf{X}}^T({\bm{f}}+\widehat{\bm{\Delta}}_+).
\end{equation}
This is exactly regression of a residual-adjusted prediction on the complete scientific frame. This weighting--regression equivalence is an instance of the established connection analysed by \citet{bruns}. Taking ${\mathbf{B}}={\mathbf{X}}$ in its original coordinates and zero penalty, it reduces to sample ordinary least squares (OLS) when ${\bm{f}}=0$, and to ${\bm{\beta}}_f$ plus sample OLS on ${\bm{\Delta}}$ for fixed ${\bm{f}}$. The theoretical normalization of ${\mathbf{B}}$ is not needed for these identities.

Exact reproduction on the protected residual space does not by itself guarantee valid inference, because the prediction residuals need not lie in that space. Under nonuniform label acquisition, the residual fit gives greater representation to units that are more likely to be labelled. Calibration corrects this imbalance for protected components, but the remaining components can contribute differently to the fitted correction and to the full-frame coefficient target. This mismatch can leave a population-specific displacement. A sampling covariance estimate describes fluctuations across acquired label sets, not this displacement; even a consistently estimated covariance therefore need not give nominal coverage. The unresolved question is whether the displacement is negligible on the scale of the reported standard errors.

\section{Conditional validity and a regression-specific recovery theorem}
Conditional validity is recovered when this displacement is negligible relative to sampling uncertainty, provided the sampling normal approximation and variance estimator are valid. Here the complete realised population is held fixed, and coverage refers to repeated label acquisition. This conditions on responses as well as attributes, whereas \citet{jin} study validity conditional on population attributes. The spatial field law plays a separate role: it allows us to show that the displacement is sufficiently small for populations occurring with probability tending to one, rather than for every possible population.

The key is the overlap between prediction residuals and the coefficient-specific directions left unbalanced by calibration. We first define the displacement and covariance estimator, then relate the standardized displacement to coverage. Under the spectral-basis conditions below, smooth target weights and decaying residual coefficient variances control this overlap while bounds on fitting error justify feasible inference. The theorem also permits selection from a fixed finite candidate library under a common-influence condition. The proof outline explains these steps, and the final subsection compares directional and total residual-energy bounds when mean-approximation error does not dominate.

\subsection{Population quantities and implemented covariance}
We identify the population displacement and construct the sampling covariance estimate used in the confidence region. The former describes the remaining systematic error; the latter measures uncertainty from label acquisition.

We use $\|\cdot\|$ for the Euclidean vector norm, $\|\cdot\|_\op$ for the induced operator norm, and $\|{\bm{u}}\|_N^2=N^{-1}{\bm{u}}^T{\bm{u}}$ for a frame vector; $\|\bm u\|_\infty=\max_i|u_i|$. Write ${\mathbf{I}}_d$ for the $d\times d$ identity, omitting the subscript when its size is clear. For symmetric matrices, ${\mathbf{A}}\preceq {\mathbf{B}}$ means that ${\mathbf{B}}-{\mathbf{A}}$ is positive semidefinite. Throughout the bounds, $C$ denotes a generic positive constant that may change from line to line. Field expectations and covariances suppress conditioning on $\mathcal F_N$; $P_p$ and $\E_p$ denote probability and expectation, respectively, under the label-sampling law conditional on the complete finite population $\mathcal U_N$.
Supplement S2 derives the population normal equations and the exact expansion below from the estimator in Section~\ref{sec:calibrated-regression}. For analysis only, put ${\mathbf{D}}_\lambda=\diag(\lambda_1,\ldots,\lambda_N)$ and ${\bm{\Sigma}}_\Delta=\Cov_\xi({\bm{\Delta}}\mid\mathcal F_N)$, and define
\begin{align*}
 {\mathbf{A}}&=N^{-1}{\mathbf{B}}^T{\mathbf{D}}_\lambda {\mathbf{B}},&{\mathbf{M}}&={\mathbf{A}}+{\bm{\Lambda}},&
 {\bm{\alpha}}&={\mathbf{M}}^{-1}N^{-1}{\mathbf{B}}^T{\mathbf{D}}_\lambda{\bm{\Delta}},\\
 {\mathbf{U}}&={\mathbf{M}}^{-1}{\mathbf{C}}_B,&{\mathbf{U}}_+&={\mathbf{U}}+{\mathbf{M}}^{-1}{\bm{\Lambda}} {\mathbf{U}},&
 {\bm{q}}_i&={\mathbf{U}}^T{\bm{b}}_i,\quad {\bm{q}}_{+,i}={\mathbf{U}}_+^T{\bm{b}}_i,\\
 e_i&=\Delta_i-{\bm{b}}_i^T{\bm{\alpha}},&
 r_i&={\bm{b}}_i^T{\mathbf{M}}^{-1}{\bm{\Lambda}}{\bm{\alpha}},&
 {\bm{\phi}}_i&={\bm{q}}_{+,i}e_i-{\bm{q}}_ir_i.
\end{align*}
The matrices ${\mathbf{A}},{\mathbf{M}}$ are $K\times K$, ${\mathbf{U}},{\mathbf{U}}_+$ are $K\times p$, and ${\bm{\alpha}}$ is a $K$-vector. The vectors ${\bm{q}}_i,{\bm{q}}_{+,i}$ and ${\bm{\phi}}_i$ are $p$-vectors, whereas $e_i,r_i$ are scalars. We write ${\bm{\ell}}_{\cdot j}=(\ell_{ij})_{i=1}^N$ for column $j$ of ${\mathbf{L}}$, and similarly ${\bm{q}}_{\cdot j},{\bm{q}}_{+,\cdot j},{\bm{\phi}}_{\cdot j}$ for the corresponding $N$-vectors. Also ${\bm{\lambda}}=(\lambda_i)_{i=1}^N$. Let ${\mathbf{D}}$ have rows ${\bm{d}}_i^T=(\lambda_i{\bm{q}}_{+,i}-{\bm{\ell}}_i)^T$. The population-normal-equation displacement is
\begin{equation}\label{eq:gap}
 {\bm{G}}=N^{-1}{\mathbf{D}}^T{\bm{\Delta}},\qquad
 \Cov_\xi({\bm{G}})=N^{-2}{\mathbf{D}}^T{\bm{\Sigma}}_\Delta {\mathbf{D}}.
\end{equation}
The displacement identity and its covariance follow in Supplement S2. This is not the finite-sample bias, and these oracle quantities are not supplied to a feasible estimator.

Linearization-based variance estimation is also used in soft calibration \citep{gao}. For the present estimator, Supplement S2 differentiates the entire fitted correction with respect to inclusion weights, giving
\begin{equation}\label{eq:psi}
 \widehat{\bm{\psi}}_i=\widehat{\bm{w}}_{+,i}\widehat e_i
          -\widehat{\bm{w}}_{0,i}{\bm{b}}_i^T{\mathbf{T}}^{-1}{\bm{\Gamma}}\widehat{\bm{\alpha}},
\end{equation}
where $\widehat{\bm{w}}_{0,i}={\mathbf{C}}_B^T{\mathbf{T}}^{-1}{\bm{b}}_i$ and $\widehat{\bm{w}}_{+,i}={\mathbf{C}}_B^T({\mathbf{T}}^{-1}+{\mathbf{T}}^{-1}{\bm{\Gamma}} {\mathbf{T}}^{-1}){\bm{b}}_i$, with $\widehat e_i=\Delta_i-{\bm{b}}_i^T\widehat{\bm{\alpha}}$. Use
\begin{equation}\label{eq:vhat}
 \widehat {\mathbf{V}}=\sum_{i\in S}\widehat{\bm{\psi}}_i\widehat{\bm{\psi}}_i^T.
\end{equation}
The second derivative term is necessary; the outcome fit changes with label inclusion. The formula omits the finite-population correction (FPC); this is justified asymptotically by $\max_i\pi_i\to0$, not by fitting an arbitrary propensity into $1-\widehat\pi_i$. Singular fits or singular required covariance matrices give empty confidence sets; all such events remain in the reporting denominator.

Fix a nominal error probability $0<\alpha_0<1$, distinct from the residual-regression coefficient ${\bm{\alpha}}$. For a fixed full-row-rank contrast matrix ${\mathbf{F}}\in\mathbb R^{r\times p}$, $1\le r\le p$, the reported region for ${\bm{b}}\in\mathbb R^p$ is
\begin{equation}\label{eq:wald}
 \mathcal C_F=\left\{{\bm{b}}:({\mathbf{F}}\widehat{\bm{\beta}}_+-{\mathbf{F}}{\bm{b}})^T({\mathbf{F}}\widehat {\mathbf{V}}{\mathbf{F}}^T)^{-1}
 ({\mathbf{F}}\widehat{\bm{\beta}}_+-{\mathbf{F}}{\bm{b}})\le\chi^2_{r,1-\alpha_0}\right\}.
\end{equation}
Here $\chi^2_{r,1-\alpha_0}$ is the $(1-\alpha_0)$ quantile of the central chi-square distribution with $r$ degrees of freedom. For a single coefficient this is the usual two-sided normal interval. Marginal intervals separately are not a joint $1-\alpha_0$ guarantee. Supplement S3 establishes covariance consistency and the conditional normal approximation underlying this region.

\subsection{Fixed-candidate coverage and the population gap}\label{sec:gap-criterion}
For a prespecified correction, the displacement measured relative to its sampling covariance determines the loss of coverage. This gives a validity criterion without specifying a spatial model.

\begin{proposition}[Conditional coverage and population displacement]\label{pr:gap-criterion}
Consider a calibration candidate fixed before acquisition of the evaluation labels. Let \(\mathbf V=n^{-2}\sum_i\pi_i(1-\pi_i)\bm\phi_i\bm\phi_i^T\) be its population linearization covariance and let \(\widehat{\mathbf V}\) denote its feasible covariance estimator. Suppose the following sampling properties hold, conditionally on the realised finite population, in \(P_\xi\)-probability: (i) the centred and standardized estimator \(\mathbf V^{-1/2}(\widehat{\bm\beta}_+-\bm\beta_U-\bm G)\) converges in distribution to \(N_p(\bm0,\mathbf I_p)\); (ii) \(n\|\widehat{\mathbf V}-\mathbf V\|_{\mathrm{op}}=o_p(1)\); and (iii) the eigenvalues of \(n\mathbf V\) are bounded above and away from zero on events of arbitrarily high \(P_\xi\)-probability. Let \(\mathbf F\in\mathbb R^{r\times p}\), \(1\le r\le p\), be a predeclared fixed full-row-rank contrast matrix. Then
\[
P_p\{\bm\beta_U\in\mathcal C_F\mid\mathcal U_N\}
= \Pr\{\chi_r^2(\delta_{F,N}^2)\le\chi^2_{r,1-\alpha_0}\}+o_{P_\xi}(1),
\quad
\delta_{F,N}^2=(\mathbf F\bm G)^\top
(\mathbf F\mathbf V\mathbf F^\top)^{-1}
(\mathbf F\bm G).
\]
Consequently, the conditional coverage probability converges to \(1-\alpha_0\) in \(P_\xi\)-probability if and only if
\[
\sqrt n\,\mathbf F\bm G\longrightarrow\bm0
\qquad\text{in }P_\xi\text{-probability}.
\]
If the residual field is Gaussian and the calibration-defect matrix is deterministic conditional on the observed frame, this is further equivalent to
\[
n\|\mathbf F \E_\xi\bm G\|^2
+n\operatorname{tr}\{\mathbf F\operatorname{Cov}_\xi(\bm G)\mathbf F^\top\}
\longrightarrow0.
\]
\end{proposition}

Here \(o_p(1)\) means that, for every \(\epsilon>0\), the conditional sampling probability of an error exceeding \(\epsilon\) tends to zero in \(P_\xi\)-probability. The notation \(\chi_r^2(u)\) denotes the noncentral chi-square law with noncentrality \(u\); its distribution function above is evaluated at the realised \(\delta_{F,N}^2\). The event uses the region in \eqref{eq:wald} and is equivalent to coverage of \(\mathbf F\bm\beta_U\). Gaussianity is required only for the moment equivalence. Supplement S3.1 proves the proposition, including fixed-rank contrasts.

The characterization is for a prespecified candidate, not for an outcome-selected correction. For \(r<p\), only the projected displacement \(\mathbf F\bm G\) matters; valid slope coverage need not imply validity for the intercept or the entire coefficient vector. The covariance and sampling-approximation assumptions are separate from the gap condition. Theorem~\ref{th:main} below supplies primitive sufficient conditions for these sampling properties and for a negligible gap, and its additional common-influence argument justifies selection from a fixed finite candidate library using the same labels.

\subsection{Primitive conditions for a bounded spectral basis}\label{sec:spectral}
We give conditions on a basis and residual spectrum that imply conditional validity. Periodic Fourier and nonperiodic cosine grids give two primary spectral-frame examples; orthogonality alone is insufficient.

Let $\mathcal H_N$ index a complete real basis $\{\bm\varphi_h\}$ of the frame, including the constant vector, with
\[
 N^{-1}\bm\varphi_h^T\bm\varphi_{h'}=\mathbf1\{h=h'\},\qquad
 \sup_{N,h,i}|\varphi_h(i)|\le C_\varphi.
\]
Give each mode an integer degree $|h|\ge0$, with the constant at degree zero and $\#\{h:|h|=r\}\le C(1+r)^{d_0-1}$ for fixed $d_0\ge1$. The finite mode sets exhaust a fixed graded index set; all modes of degree at most $k=k(N)$ are present. Write $H_k=\operatorname{span}\{\bm\varphi_h:|h|\le k\}$. The scientific dimension $p$ is fixed and $c_J\mathbf I_p\preceq\mathbf J\preceq C_J\mathbf I_p$ for fixed positive constants.

Index candidates by $\nu\in\mathcal I_N$, with $|\mathcal I_N|\le J_0<\infty$. Dictionaries and penalties are fixed conditional on $\mathcal F_N$ before evaluation-label acquisition. Each dictionary is a subset of the basis containing $H_k$, which is unpenalized. Included modes above $k$ have normalized diagonal penalty $\rho(|h|/k)^{2b}$, with uniform positive lower and upper bounds on $\rho$ and $b\ge b_->d_0/2$. Alternatively, an unpenalized dictionary containing $H_k$ may have maximum degree $O(k)$. A subscript $\nu$ denotes a candidate and $\widehat\nu$ a label-dependent choice.

For every coefficient $j$, put
\begin{equation}\label{eq:g}
 \bm g_j=\bm\ell_{\cdot j}/\bm\lambda.
\end{equation}
Assume the spectral Sobolev bound
\[
 \max_j\sum_{h\in\mathcal H_N}(1+|h|^2)^{s_*}
       |N^{-1}\bm\varphi_h^T\bm g_j|^2\le C,\qquad s_*>d_0/2.
\]
This condition concerns the target weights, not inverse acquisition intensity alone. Suppose there exists a decomposition on a possibly enlarged field probability space,
\[
 \bm\Delta=\bm\mu+\bm v+\bm\zeta,\qquad
 \bm\mu=\E_\xi\bm\Delta,\qquad \E_\xi\bm v=\E_\xi\bm\zeta=\bm0.
\]
The same decomposition must satisfy \eqref{eq:field-cov}, \eqref{eq:field-nd} and one of the stated tail alternatives jointly. It need not be unique. The label-sampling kernel depends only on $\mathcal U_N$, not the auxiliary decomposition; neither latent component is used by the estimator, and $\bm v,\bm\zeta$ need not be independent. For fixed $t>d_0/2$, assume some $\bm h_k\in H_k$ satisfies
\[
 \|\bm\mu-\bm h_k\|_N\le Ck^{-t},\qquad
 \|\bm\mu-\bm h_k\|_\infty\le Ck^{d_0/2-t},\qquad
 \|\bm u\|_N^2=N^{-1}\bm u^T\bm u.
\]
No amplitude bound is imposed on $\bm h_k$. A uniform spectral Sobolev-$t$ bound suffices (Supplement S5). Let $\mathbf Q_k$ be the Euclidean orthogonal projector onto $H_k^\perp$. For $\bm\Sigma_v=\Cov_\xi(\bm v)$ and $\bm\Sigma_\zeta=\Cov_\xi(\bm\zeta)$, assume
\begingroup\renewcommand{\theHequation}{assumptionG}
\begin{equation}\tag{G}\label{eq:field-cov}
 \|\mathbf Q_k\bm\Sigma_v\mathbf Q_k\|_\op\le CNk^{-a},\quad
 N^{-1}\tr(\mathbf Q_k\bm\Sigma_v\mathbf Q_k)\le Ck^{d_0-a},\quad
 \|\bm\Sigma_\zeta\|_\op\le C,\quad a>d_0.
\end{equation}
\endgroup
Define the proof-only quantities $\bm\phi_{0i}=\bm\ell_i\zeta_i/\lambda_i$ and
$\mathbf V_0=n^{-2}\sum_i\pi_i(1-\pi_i)\bm\phi_{0i}\bm\phi_{0i}^T$; neither enters the feasible estimator. Require
\begingroup\renewcommand{\theHequation}{assumptionND}
\begin{equation}\tag{ND}\label{eq:field-nd}
 \forall\eta>0\ \exists c_\eta>0:\quad
 \liminf_{N\to\infty}P_\xi\{\lambda_{\min}(n\mathbf V_0)\ge c_\eta\}\ge1-\eta.
\end{equation}
\endgroup
Condition \eqref{eq:field-nd} requires nondegenerate variation in the common-influence coefficient directions, not independent or Gaussian noise. For a decomposition with $\bm\zeta=\sigma_N\bm\varepsilon$, $\sigma_N\to0$ and uniformly bounded second moments of $\varepsilon_i$, bounded loadings imply $n\mathbf V_0\to0$ in probability. The theorem does not cover this change of scale through that decomposition; any alternative must satisfy all conditions jointly.

Finally, assume either $\sup_{N,i}|\Delta_i|\le C_\Delta$ almost surely, or, for a fixed $\delta>0$,
\begingroup\renewcommand{\theHequation}{assumptionT}
\begin{equation}\tag{T}\label{eq:field-tail}
 \forall\eta>0\ \exists C_\eta<\infty:\quad
 \liminf_{N\to\infty}P_\xi\!\left\{
 \max_\nu N^{-1}\sum_i(|e_{\nu i}|^{2+\delta}+|r_{\nu i}|^{2+\delta})\le C_\eta
 \right\}\ge1-\eta.
\end{equation}
\endgroup
The event in \eqref{eq:field-tail} is $\mathcal U_N$ measurable. Bounded $Y$ alone does not suffice without control of $f$. Neither \eqref{eq:field-nd} nor \eqref{eq:field-tail} follows from \eqref{eq:field-cov}; upper eigenvalue localization follows from $\E_\xi\tr(n\mathbf V_0)\le C$. The operator and trace bounds in \eqref{eq:field-cov} control displacement and common-influence approximation, respectively; \eqref{eq:field-nd} fixes a nondegenerate sampling scale, and bounded $\bm\Delta$ or \eqref{eq:field-tail} supports conditional sampling and feasible covariance estimation. Supplement Lemma S2 supplies a sufficient Gaussian check for \eqref{eq:field-tail}, with a separate pointwise tail-variance condition.

\begin{samepage}
\begin{theorem}[Regression and contrast recovery]\label{th:main}
Under the preceding conditions, including \eqref{eq:field-nd} and either bounded $\bm\Delta$ or \eqref{eq:field-tail}, if
\begin{equation}\label{eq:rates}
 k\to\infty,\quad k^{2d_0}/n\to0,\quad
 nk^{-(a+2s_*)}+nk^{-2(s_*+t)}\to0,\quad n/N\to0,
\end{equation}
then, for any measurable same-label selector $\widehat\nu$ from the fixed finite candidate library, every predeclared fixed full-row-rank ${\mathbf{F}}$ satisfies
\[
 \Pp\{{\bm{\beta}}_U\in\mathcal C_{F,\widehat\nu}\mid\mathcal U_N\}
   \longrightarrow1-\alpha_0\quad\text{in }P_\xi\text{-probability}.
\]
Uniformly over candidates, $\max_\nu\E_\xi\norm{\bm G_\nu}^2\le Ck^{-2s_*}(k^{-2t}+k^{-a}+N^{-1})$ and $\max_\nu\norm{\bm{\mathcal R}_\nu}=O_p(k^{d_0}/n)$. The same observed labels may fit the correction and select the candidate; continuous unrestricted nuisance tuning or choosing ${\mathbf{X}}$ is not included.
\end{theorem}
\end{samepage}

\begin{corollary}[Periodic Fourier frame]\label{cor:fourier}
On the odd periodic grid $\{0,\ldots,m-1\}^{d_0}/m$, with $N=m^{d_0}$, use the canonical real Fourier basis and degree $|h|=\|\bm h\|_\infty$. Take $\bm v=\sum_hc_h\bm\varphi_h$ with independent centered Gaussian $c_h$, $\Var(c_h)\le C(1+|h|)^{-a}$, $a>d_0$, and $\bm\zeta=\sigma\bm\varepsilon$, $\sigma>0$, where the $\varepsilon_i$ are iid standard Gaussian and independent of the coefficients. Under the stated target, mean and acquisition conditions, Theorem~\ref{th:main} applies with the same rates and coverage conclusion. In particular, the original full-cube dictionaries and penalties are included. Euclidean-frequency Sobolev weights and variance decay give the same class up to constants.
\end{corollary}

\begin{corollary}[Nonperiodic cosine frame]\label{cor:dct}
Let $N=m^2$, $m\ge3$, and $(x_i,y_j)=((i+1/2)/m,(j+1/2)/m)$, $0\le i,j<m$. Keep the scientific design $\mathbf X=(\bm1,\bm x-\bm1/2,\bm y-\bm1/2)$ and take $\lambda(x,y)=1+\eta\cos(\pi x)$ for fixed $0<\eta<1$. Use the tensor type-II discrete cosine transform (DCT-II) basis
\[
 \begin{gathered}
 a_0(x)=1,\qquad a_u(x)=\sqrt2\cos(\pi ux)\quad(1\le u<m),\\
 \varphi_{uv}(x,y)=a_u(x)a_v(y),\qquad |(u,v)|=\max(u,v).
 \end{gathered}
\]
For the dictionary and penalty rules above, a Gaussian structural field diagonal in this basis with coefficient variances bounded by $C(1+\max(u,v))^{-a}$, $a>2$, and independent noise $\bm\zeta=\sigma\bm\varepsilon$, $\sigma>0$, with iid standard Gaussian $\varepsilon_i$ satisfy the theorem's field conditions. If $\mu(x,y)=c_0+c_1(x-1/2)+c_2(y-1/2)$ with uniformly bounded coefficients, the target and mean conditions hold for any fixed $1<s_*,t<3/2$. Hence \eqref{eq:rates} implies the same conditional coverage and finite-library selection conclusion.
\end{corollary}
Here $\mathbf J=\diag(1,V_m,V_m)$, $V_m=(m^2-1)/(12m^2)$, and the three target weights are $1/\lambda$, $(x-1/2)/(V_m\lambda)$ and $(y-1/2)/(V_m\lambda)$. Supplement S5.3--S5.4 proves both corollaries, including the uniform discrete coefficient bounds after division by $\lambda$. For example, $a=3$, $s_*=5/4$, $t=7/5$ and $k\asymp n^{1/5}$ give the three rate powers $-1/5,-1/10,-3/50$ in \eqref{eq:rates}; also assume $n/N\to0$. This cosine-grid result is for the stated spectral covariance class, not an arbitrary nonperiodic covariance or irregular frame.

A smooth amplitude modulation of a periodic Gaussian field provides a nonstationary, non-diagonal example (Supplement S5.6). A periodic logistic--Mat\'ern binary population with specified smooth predictors and intensities satisfies the bounded-response route (Supplement S5.7). These restricted examples require no claim about arbitrary Euclidean Mat\'ern covariance or irregular frames.

\subsection{Proof outline}\label{sec:proof-outline}
The algebra and sampling arguments are basis independent. Supplement S2--S5 proves the following steps, uniformly over the fixed finite library; fitting orders below refer to the joint field and sampling law.

\paragraph{1. Control displacement in the covariance directions left unbalanced.}
The smoother $\mathcal S=\mathbf B\mathbf M^{-1}\mathbf B^T\mathbf D_\lambda/N$ fixes $H_k$ and is self-adjoint in the $\lambda$-weighted inner product. Thus $\bm d_{\cdot j}=-\mathbf D_\lambda(I-\mathcal S)^2\bm g_j$ is Euclidean orthogonal to $H_k$. Sobolev approximation gives $\|\bm d_{\cdot j}\|_N=O(k^{-s_*})$, so mean approximation bounds $|\E_\xi G_{\nu j}|$ by $Ck^{-(s_*+t)}$. For the possibly correlated $\bm v,\bm\zeta$, use $\E(A+B)^2\le2\E A^2+2\E B^2$ and \eqref{eq:field-cov}, obtaining
\begin{equation}\label{eq:outline-gap}
 \max_\nu\E_\xi\|\bm G_\nu\|^2
 \le Ck^{-2s_*}(k^{-2t}+k^{-a}+N^{-1}).
\end{equation}
Hence $\max_\nu\sqrt n\|\bm G_\nu\|=o_{P_\xi}(1)$ under \eqref{eq:rates}. The operator bound, rather than total tail energy, controls this step.

\paragraph{2. Account for fitting on the evaluation labels.}
On nonsingular samples, the exact expansion in Supplement S2 is
\begin{equation}\label{eq:outline-expansion}
 \widehat{\bm\beta}_{+,\nu}-\bm\beta_U
 =\bm G_\nu+n^{-1}\sum_i(R_i-\pi_i)\bm\phi_{\nu i}+\bm{\mathcal R}_\nu,
 \qquad \bm\phi_{\nu i}=\bm q_{+,\nu i}e_{\nu i}-\bm q_{\nu i}r_{\nu i}.
\end{equation}
The full inclusion derivative retains both terms. Shell bounds give
$d_{\rm eff}=\tr(\mathbf A\mathbf M^{-1}),\ell_{\max}=\max_i\bm b_i^T\mathbf M^{-1}\bm b_i=O(k^{d_0})$.
The three exact remainder products therefore satisfy
\begin{equation}\label{eq:outline-remainder}
 \bm{\mathcal R}_\nu=O_p\{\sqrt{d_{\rm eff}\ell_{\max}}/n\}
 =O_p(k^{d_0}/n)=o_p(n^{-1/2}).
\end{equation}
Supplement S3 handles bounded residuals and localized moments separately. Both routes yield $n\|\widehat{\mathbf V}_\nu-\mathbf V_\nu\|_\op=o_p(1)$; $\pi_{\max}:=\max_{1\le i\le N}\pi_i=O(n/N)\to0$ justifies omission of the FPC in the feasible covariance.

\paragraph{3. Reduce the candidates to a common first-order vector.}
The trace bound in \eqref{eq:field-cov}, mean approximation and smoothing of $\bm\zeta$ give
\begin{equation}\label{eq:outline-influence}
 \max_\nu\E_\xi\left\{N^{-1}\sum_i\lambda_i
 \|\bm\phi_{\nu i}-\bm\phi_{0i}\|^2\right\}
 \le C\{k^{-2s_*}+k^{-2t}+k^{d_0-a}+k^{d_0}/N\}\to0.
\end{equation}
Cauchy--Schwarz implies $\max_\nu n\|\mathbf V_\nu-\mathbf V_0\|_\op=o_{P_\xi}(1)$. Condition \eqref{eq:field-nd} and $\E_\xi\tr(n\mathbf V_0)\le C$ transfer positive lower and finite upper eigenvalue bounds to the $\mathcal U_N$-measurable $\mathbf V_\nu$, before the conditional sampling limits are applied. No deterministic limit for $n\mathbf V_0$ is required. Conditional Chebyshev and a finite union bound make all candidates equivalent at scales $n^{-1/2}$ for estimates and $n^{-1}$ for covariances. Markov's inequality transfers joint approximation errors to conditional ones, giving same-label selected Wald validity without conditioning on the selected candidate itself.

The protected-space dimension is $\dim(H_k)=O(k^{d_0})$; for Fourier cubes, $\dim(H_k)=(2k+1)^{d_0}$. The sufficient sampling rate is $k^{2d_0}/n\to0$, not $k^2/n$ in all dimensions. In two dimensions the conditions include $k^4/n\to0$. Here $a_n\ll b_n$ means $a_n/b_n\to0$, while $a_n\asymp b_n$ means their ratio stays bounded above and away from zero. A nonempty window is $n^{1/r_*}\ll k\ll n^{1/(2d_0)}$, where $r_*:=\min\{a+2s_*,2(s_*+t)\}>2d_0$. If $t\ge a/2$, the mean condition adds no rate restriction. These are sufficient rates, not a proof that this dimension restriction is necessary. Variance estimation in a proportional-dimensional generalized regression (GREG) regime is studied by \citet{bouhadra}; the effective-complexity bounds used here impose a different regime.

\begin{tuningremark}[Choosing from predeclared candidates]
Fix finitely many dictionaries and penalties with a common protected radius in the stated window and the required uniform bounds. This permits same-label selection, not an optimal finite-sample tuning rule; fixed small radii do not satisfy $k\to\infty$. For example, $k(n)=\max\{1,\lfloor n^{1/(2d_0)}/\log n\rfloor\}$, $n>1$, needs no numerical knowledge of $a,s_*,t$: for fixed $r_*>2d_0$, $k^{2d_0}/n\asymp(\log n)^{-2d_0}\to0$ and $nk^{-r_*}\asymp n^{1-r_*/(2d_0)}(\log n)^{r_*}\to0$. On the two grids, $n/N\to0$ also gives $k/m\to0$. The field regularity assumptions still require separate justification.
\end{tuningremark}

\subsection{Why total residual MSE is not the deciding quantity}\label{sec:energy-comparison}
For the same estimator, evaluating covariance in the remaining calibration-error directions certifies validity under weaker sufficient conditions than bounding total residual energy. We compare these two bounds to isolate the gain from retaining directional information.

The structural residual energy beyond $k$ is $O(k^{d_0-a})$. Combining it with target-weight error $O(k^{-2s_*})$ by a norm-product bound gives $k^{-(a+2s_*-d_0)}$. The directional covariance calculation instead uses its operator norm divided by $N$, $O(k^{-a})$, yielding the stronger bound in the theorem. The $\bm\zeta$ contribution remains $O(k^{-2s_*}/N)$; it cannot be discarded as structured energy.

For this comparison suppose $t\ge a/2$, so the mean term imposes no additional restriction, and put $r_v=a+2s_*$. Using the residual-energy bound to make the gap negligible requires $nk^{-(r_v-d_0)}\to0$. Together with $k^{2d_0}/n\to0$, this admits a radius sequence only if $r_v>3d_0$. The directional bound instead requires $nk^{-r_v}\to0$ and admits such a sequence whenever $r_v>2d_0$, as already implied by the theorem's regularity assumptions. For example, take $d_0=2$, $a=3$, $s_*=5/4$, $t\ge3/2$ and $k\asymp n^{1/5}$. Then
\[
 k^{2d_0}/n\asymp n^{-1/5},\qquad
 nk^{-r_v}\asymp n^{-1/10},\qquad
 nk^{-(r_v-d_0)}\asymp n^{3/10}.
\]
The $\bm\zeta$ contribution vanishes under $n/N\to0$. Thus the directional argument certifies recovery in a regime that this residual-energy argument does not certify (Supplement S4). This is a comparison of sufficient bounds for the same estimator, not a failure result for all procedures, nor a claim that all previous calibration theory uses the weaker bound.

The benefit is a wider certified recovery range for the same estimator. Known true acquisition probabilities make the design gap zero without these field assumptions. The recovery conditions concern the specified scientific regressors and their measurement definition.

\section{Coefficient-specific gaps and limits of correction}
For a protected subspace $H\subseteq\mathbb R^N$, let ${\mathbf{B}}_H$ have columns spanning $H$ and let ${\mathbf{P}}_H$ be its Euclidean orthogonal projector, ${\mathbf{Q}}_H={\mathbf{I}}-{\mathbf{P}}_H$, and $E_{\mu,H}=\inf_{\bm{h}\in H}\norm{{\bm{\mu}}-{\bm{h}}}_N$, with $\norm{\bm{u}}_N^2={\bm{u}}^T{\bm{u}}/N$. These unweighted projections differ from projection in the ${\bm{\lambda}}$-weighted norm. On an arbitrary fixed frame, for residual mean ${\bm{\mu}}=\E_\xi{\bm{\Delta}}$ and covariance ${\bm{\Sigma}}=\Cov_\xi({\bm{\Delta}})$, put $\epsilon_{H,j}=\inf_{\bm{h}\in H}\norm{{\bm{g}}_j-{\bm{h}}}_\lambda$. The corrected smoother is $2\mathcal{S}-\mathcal{S}^2$, with $\mathcal{S}={\mathbf{B}}{\mathbf{M}}^{-1}{\mathbf{B}}^T{\mathbf{D}}_\lambda/N$. It fixes $H$ and is a weighted self-adjoint contraction. Here ${\bm{q}}_{+,\cdot j}$ is the frame vector for coefficient $j$, and ${\bm{d}}_{\cdot j}$ is the corresponding column of the defect matrix. Consequently
\begin{equation}\label{eq:bound}
 \norm{{\bm{q}}_{+,\cdot j}-{\bm{g}}_j}_\lambda\le\epsilon_{H,j},\quad {\mathbf{B}}_H^T{\bm{d}}_{\cdot j}=0,
\end{equation}
\[
 \E_\xi G_j^2\le\lambda_+\epsilon_{H,j}^2
 \left\{E_{\mu,H}^2+\frac{\norm{{\mathbf{Q}}_H{\bm{\Sigma}} {\mathbf{Q}}_H}_\op}{N}\right\}.
\]
This keeps the stochastic residual field ${\bm{\Delta}}$ unchanged and places the scientific covariates in the deterministic target weights. It does not assume, incorrectly, that multiplying a stationary field by an arbitrary covariate preserves stationarity.

The bound also gives a limited interpretation of map-augmented dictionaries. If fixed prediction $\bm f$ lies in the unpenalized linear space $H$, chosen before label acquisition, then $\bm\mu=\E_\xi\bm Y-\bm f$ and
\[
 E_{\mu,H}=\inf_{\bm h\in H}\|\E_\xi\bm Y-\bm f-\bm h\|_N
           =\inf_{\bm h\in H}\|\E_\xi\bm Y-\bm h\|_N.
\]
Thus a discontinuity of $f$ need not itself burden residual-mean approximation once $f$ is protected. This ensures neither small target-weight error nor small unprotected covariance and does not verify Theorem~\ref{th:main} for irregular LUCAS sites.

The fixed-candidate criterion in Section~\ref{sec:gap-criterion} converts these displacement bounds into coverage statements once the sampling approximation and covariance conditions hold. The displacement is an oracle theoretical quantity, not an observable diagnostic when reference responses are missing.

A simple construction shows why average balance need not protect a slope. On the periodic grid, let $\bm{x}=(\sqrt2\sin(2\pi s_{i1}))_{i=1}^N$, $\bm{w}=(\sqrt2\cos(2\pi s_{i1}))_{i=1}^N$, ${\mathbf{X}}=(\bm{1},{\bm{x}})$, ${\mathbf{B}}=(\bm{1},{\bm{x}},{\bm{w}})$ and ${\bm{\lambda}}=\bm{1}+\eta {\bm{x}}{\bm{w}}$, $|\eta|<1$. Here $\bm{1}$ is the all-ones vector, products and powers are componentwise, and ${\bm{w}}$ is a cosine feature, separate from the target-weight vectors ${\bm{g}}_j$. With centred Gaussian amplitudes $Z_x,Z_w$, independent of each other and of ${\bm{\varepsilon}}$, let ${\bm{\Delta}}=Z_x{\bm{x}}+Z_w{\bm{w}}+\sigma{\bm{\varepsilon}}$. Then common mean calibration has structural gap $(0,(\eta/2)Z_w)^T$, while target calibration with the same three functions removes both structural components. Common calibration can also remove them by adding ${\bm{x}}^2,{\bm{x}}{\bm{w}}$. Independent nugget errors contribute the remaining gaps.

Proposition~\ref{pr:gap-criterion} identifies the residual displacement that must vanish for conditional coverage, whereas Theorem~\ref{th:main} provides a spectral sufficient condition for its disappearance. The following periodic construction shows that the stochastic directional-gap rate in that sufficient condition can actually be attained for a slope, even when the mean is structurally protected. Its purpose is to distinguish failure of gap recovery from failure of the centred sampling approximation or variance estimation; it is not a general minimax lower bound.

\begin{proposition}[An attainable slope gap with a protected mean]\label{pr:slope-boundary}
Restrict every candidate dictionary to radius at most $C_0k$, for a fixed $C_0\ge1$. Take ${\mathbf{X}}=(\bm{1},{\bm{x}})$ with $\bm{x}=(\sqrt2\cos(2\pi s_{i1}))_{i=1}^N$, choose an integer $L_k>2C_0k$ of order $k$, and suppose $m>4(L_k+1)$. Let
\[
 \bm{z}=(\sqrt2\cos(2\pi L_ks_{i1}))_{i=1}^N,\quad {\bm{\lambda}}=\bm{1}+\delta_k\bm{z},\quad
 \delta_k=c k^{-s_*},
\]
with fixed $0<c<1/\sqrt2$. Let ${\bm{\Delta}}=Z_c{\bm{h}}_c+Z_s{\bm{h}}_s+\sigma{\bm{\varepsilon}}$, where ${\bm{h}}_c,{\bm{h}}_s$ are the normalized cosine/sine pair at frequency $L_k+1$ and $Z_c,Z_s$ are independent $N(0,\tau_k)$, $\tau_k\asymp k^{-a}$. Then the periodic Fourier weight/field classes of Corollary~\ref{cor:fourier} hold, and every candidate, with or without one-step correction, has
\begin{equation}\label{eq:slope-boundary}
 \Cov_\xi({\bm{G}})=\delta_k^2
 \begin{pmatrix}\sigma^2/N&0\\0&\tau_k/2+\sigma^2/N\end{pmatrix}.
\end{equation}
The mean has no structured gap, but the slope attains order $k^{-(a+2s_*)}+k^{-2s_*}/N$. When the sampling complexity conditions hold, the mean remains valid whereas slope validity requires the structured gap to vanish at the $n^{-1/2}$ scale. This sequence attains the stochastic directional-gap order within the specified dictionary class, without establishing sharpness of the mean-approximation term.
\end{proposition}

In particular, writing \(r_v=a+2s_*\), the admissible regularity conditions give \(r_v>2d_0\). Along sequences with \(n\asymp k^{r_v}\), \(k^{2d_0}/n\to0\), while \(n\operatorname{Var}_\xi(G_{\mathrm{slope}})\) stays bounded away from zero in the stated construction (with \(n/N\to0\)). Thus the slope fails the gap criterion of Proposition~\ref{pr:gap-criterion} even though the sampling-complexity requirement is met. This conclusion concerns the restricted dictionary and acquisition sequences of Proposition~\ref{pr:slope-boundary}.

The proof uses ${\mathbf{A}}={\mathbf{I}}_K$ and ${\bm{q}}_{+,i}=(1,x_i)^T$: acquisition is orthogonal to every product of dictionary functions, but $\bm{z}{\bm{x}}$ has a component in the omitted residual mode. Supplement S5 provides the exact calculation and the uniform Sobolev bound for ${\bm{\ell}}_{\cdot j}/{\bm{\lambda}}$. The construction identifies an error direction that the correction must span or sufficiently approximate. Full-grid dictionaries are permitted by the recovery theorem but excluded from this omitted-direction construction. Smoothing penalties can reduce instability while leaving or reintroducing imbalance. Exact reproduction, a small gap, and small total RMSE are distinct properties.

\section{Constructed regression example}
Monte Carlo standard error is abbreviated MCSE; RMS denotes root mean square, and RMSE root mean squared estimation error. Table coverage percentages refer to nominal 95\% intervals. The constructed experiment uses a $129\times129$ grid, 32 Gaussian residual populations, 500 common label sets per mechanism and budget, $\eta=0.6$, structural standard deviations 0.8 and nugget standard deviation 0.4. Results for all budgets $100,400,1600$ and both uniform and cross-loading acquisition are retained in the computational records. Table~\ref{tab:mechanism} uses the middle budget. Slope gap RMS is the root mean square of the method-specific population slope displacements across the 32 generated residual populations. The standard errors account for the crossed population/label-set design, not $32\times500$ independent Bernoulli trials.
\begin{table}[!htbp]
\centering\small
\caption{Constructed stationary periodic mechanism, $n=400$. Crossed-population/label Monte Carlo SEs; 32 populations, 500 shared label sets.}\label{tab:mechanism}
\begin{tabular}{lrrrrr}\toprule
Method ($K$) & Mean (\%) & MCSE & Slope (\%) & MCSE & Slope gap RMS\\
\midrule
Unweighted & 94.53 & 0.45 & 26.41 & 6.59 & 0.2508\\
Common (3) & 95.35 & 0.17 & 26.44 & 6.58 & 0.2508\\
Product (5) & 95.08 & 0.15 & 94.94 & 0.16 & 0.0004\\
Target (3) & 95.35 & 0.17 & 95.09 & 0.15 & 0.0011\\
HT oracle & 94.86 & 0.43 & 95.46 & 0.35 & 0.0000\\
\bottomrule\end{tabular}
\end{table}

The experiment separates exact reproduction of structured errors from remaining nugget and sampling errors. Its prediction-error basis and acquisition pattern were deliberately constructed to expose the distinction. It is not an external-data ranking and does not constitute a numerical proof of Theorem~\ref{th:main} along a growing sequence. A separate fixed-budget experiment in Supplement S6.3 retains an omitted residual direction: with the dictionary fixed, increasing the label count narrows intervals around a displaced target and can reduce slope coverage. The observed coverage is compared with the Gaussian approximation using the realised gap and sampling covariance; this is a failure-mechanism check, not a test of the theorem along an admissible growing sequence. An exploratory growing-radius pilot is reported in Supplement S6.6.

\paragraph{Directional covariance at fixed residual energy.}
To isolate the role of residual covariance directions, we conducted an additional periodic-grid experiment in which the scientific regression target, label-acquisition intensities and residual-correction dictionaries were fixed. Two Gaussian residual fields were constructed with the same covariance eigenvalues and expected squared residual energy, but different orientations of their structural variation relative to the slope-specific calibration defect. The aligned field contains a structural Fourier component in the defect direction; the orthogonal field does not. Both fields retain the same independent nugget variance (\href{supplement.pdf\#directional-covariance}{Supplement S6.4}).

For the targeted soft-calibration estimator with its fitted one-step correction, at an expected label count of 1000 and one structural Fourier pair, the aligned and orthogonal fields gave 75.39\% and 94.85\% marginal slope coverage, respectively, for nominal 95\% intervals. Their paired difference was 19.46 percentage points (field-cluster Monte Carlo standard error 1.83 points). Under uniform label acquisition, the corresponding coverages were 94.96\% and 94.90\%, with no material contrast. These results demonstrate the finite-frame importance of covariance variation in the calibration-defect direction. They neither establish the asymptotic rate conditions of Theorem~\ref{th:main} nor imply superiority of the targeted estimator in every setting.

\section{Real-record-set regression: LUCAS--WorldCover}
The 2,665 joined Estonia LUCAS sites form the finite frame \citep{lucas}. The response is the 2018 woodland indicator, the prediction the 2021 WorldCover tree-cover indicator \citep{wc}, and the scientific ${\mathbf{X}}$ consists of an intercept and centered projected east/north coordinates in 100 km units. The target is the census linear-projection trend in the binary response on this observed-site frame. The complete reference coefficient is
\[
 {\bm{\beta}}_U=(0.47279550,\ 0.02883751,\ -0.01154455)^T.
\]
The reference protocol contains field and image-interpretation records, and the two products differ in date and class definition. Their discrepancy is not an unqualified 2021 map-accuracy measure.

Three predeclared artificial acquisition mechanisms (uniform, map, map/east) and expected label counts $100,250,500$ give 18,000 saved label sets, 2,000 per cell. Each of the 12 methods is evaluated on every label set, and each evaluation returns both prediction variants. There are 216,000 method--label-set evaluations and 432,000 prediction-specific three-coefficient estimates. This is a conditional real-residual validation with labels artificially hidden, not repetition of the actual LUCAS survey design. All scientific covariates were fixed before the regression runs; the data and earlier mean results were already known.

Twelve methods were retained. The map/quadratic basis has seven coordinates; its relevant product closure has thirteen. Sixteen-dimensional soft methods use a full-frame QR-orthonormalized map/polynomial basis, protect its first seven coordinates, and apply ${\bm{\Gamma}}=n\diag(0_7,{\mathbf{I}}_9)$ in those coordinates, where $0_7$ is the $7\times7$ zero block and ${\mathbf{I}}_9$ the $9\times9$ identity, with matched uncorrected and corrected versions. BC in the tables denotes the one-step correction. Supplement S6 defines the ordered bases and fixed-penalty jackknife. No penalty tuning was done. The target-${\mathbf{X}}$ method is the direct residual-regression comparator. Setting the prediction offset to zero while retaining the original map in ${\mathbf{B}}$ does not remove prediction information; these variants are not prediction-free comparators.
\begin{table}[!htbp]
\centering\small
\caption{LUCAS east slope, map/east acquisition, expected $n=250$. All 12 predeclared methods; full inclusion-derivative covariance. MCSE is in percentage points. Pop. gap is the method-specific population-normal-equation displacement: $G_j^{\rm com}$ for common-weight methods, $G_j$ for targeted methods, and zero for the known-probability HT oracle (Supplement S6.1). It uses the complete response and true acquisition intensities.}\label{tab:east}
\begin{tabular}{lrrrrr}\toprule
Method ($K$) & Coverage (\%) & MCSE & RMSE & Width & Pop. gap\\
\midrule
Unweighted score PPI & 90.30 & 0.66 & 0.03199 & 0.10788 & -0.01982\\
Common map (2) & 96.20 & 0.43 & 0.02992 & 0.12123 & +0.00040\\
Common spatial (6) & 95.70 & 0.45 & 0.03090 & 0.12079 & +0.00256\\
Common map/spatial (7) & 89.90 & 0.67 & 0.04166 & 0.14157 & +0.02005\\
Common product space (13) & 92.90 & 0.57 & 0.03435 & 0.12370 & +0.00015\\
Target ${\mathbf{X}}$ (3) & 96.00 & 0.44 & 0.02912 & 0.11718 & +0.00447\\
Target map/spatial (7) & 94.90 & 0.49 & 0.02936 & 0.11533 & +0.00102\\
Common soft (16) & 88.15 & 0.72 & 0.04163 & 0.13660 & +0.01986\\
Common soft + BC (16) & 87.75 & 0.73 & 0.04183 & 0.13531 & +0.01976\\
Target soft (16) & 93.95 & 0.53 & 0.02941 & 0.11116 & +0.00067\\
Target soft + BC (16) & 93.50 & 0.55 & 0.02969 & 0.11032 & +0.00049\\
Known-$\pi$ HT (oracle) & 94.05 & 0.53 & 0.03534 & 0.13542 & +0.00000\\
\bottomrule\end{tabular}
\end{table}
\FloatBarrier

The matched seven-dimensional comparison changes the calibration right-hand side, not the target coefficient, label set or residual space. For the east slope under map/east acquisition, common/targeted coverage is 93.10\%/91.80\% at $n=100$, 89.90\%/94.90\% at $n=250$ (Table~\ref{tab:east}), and 86.35\%/96.10\% at $n=500$ (\href{supplement.pdf\#lucas-crosscondition}{Supplement Table S7}). The corresponding population gaps, $0.020054$ and $0.001020$, do not change with $n$; the $n=250$ paired coverage difference is 5.00 percentage points (MCSE 0.78). This pattern is consistent with a persistent displacement becoming more consequential as sampling uncertainty decreases, although finite-sample covariance estimation also matters. Common map calibration and target-${\mathbf{X}}$ residual regression remain competitive for slopes. Target-${\mathbf{X}}$ has poor intercept coverage here (21.80\% at $n=250$), so slope performance does not establish full-vector validity.

For the north slope, the gap changes from $-0.00036$ (common map/spatial) to $+0.00499$ (targeted; Supplement Table S1). Across nine acquisition--budget settings and both slopes, targeting decreases RMSE and mean interval width in all 18 matched comparisons, but increases coverage in only 10 and improves absolute proximity to 95\% in only six (\href{supplement.pdf\#lucas-crosscondition}{Supplement Table S7}; archived summaries). Changes in gap size alone do not determine finite-sample coverage or RMSE. In the matched targeted soft pair, adding BLUP correction does not improve east-slope RMSE. The supplement also reports additional north-slope, paired and small-budget results; machine-readable results retain both variance estimates and ${\bm{f}}=0$ variants.
\begin{table}[!htbp]
\centering\small
\caption{Joint 95\% ellipsoid for both coordinate slopes, map/east acquisition, $n=250$.}\label{tab:joint}
\begin{tabular}{lrr}\toprule
Method & Coverage (\%) & MCSE (pp)\\
\midrule
Unweighted score PPI & 93.35 & 0.56\\
Common map (2) & 96.65 & 0.40\\
Common spatial (6) & 95.60 & 0.46\\
Common map/spatial (7) & 93.45 & 0.55\\
Common product space (13) & 93.50 & 0.55\\
Target ${\mathbf{X}}$ (3) & 96.00 & 0.44\\
Target map/spatial (7) & 95.25 & 0.48\\
Common soft (16) & 91.10 & 0.64\\
Common soft + BC (16) & 90.20 & 0.66\\
Target soft (16) & 94.50 & 0.51\\
Target soft + BC (16) & 94.05 & 0.53\\
Known-$\pi$ HT (oracle) & 95.75 & 0.45\\
\bottomrule\end{tabular}
\end{table}

\FloatBarrier
The two-slope ellipsoid uses off-diagonal covariance and is not formed by intersecting nominal marginal intervals. At small $n$, both linearization and fixed-penalty jackknife are reported: the latter can improve coverage by widening intervals. Fifty-one joint-region evaluations, all at $n=100$, fail the fixed covariance-rank check across method/prediction/contrast/variance combinations, arising from two label configurations. They count as noncoverage in the all-attempt denominator. No pseudoinverse replaces them. At expected $n=500$, the expected frame sampling fraction is 18.76\%; under map/east acquisition, the maximum inclusion probability is 0.43894. Thus the vanishing-maximum-probability approximation need not be accurate at this budget. Omitting the FPC can contribute to conservative intervals, but does not remove the calibration gap; known-probability HT retains its FPC.

\section{Discussion and evidential scope}
Conditional validity of calibrated residual-regression PPI depends on how the realised prediction-error field aligns with coefficient-specific calibration defects. A consistent sampling covariance estimate cannot correct a persistent population displacement. For a prespecified correction, we characterize its effect on Wald coverage through noncentrality; for the implemented correction family, we give sufficient conditions for the displacement to vanish on the sampling-error scale. By controlling covariance in the remaining defect directions instead of total residual energy, these conditions certify recovery over a wider regime under the stated assumptions. This geometry is coefficient dependent: a population mean may be protected while a regression slope is not.

The corresponding target weights are partial-regression loadings divided by acquisition intensity. Formulating the correction through residual regression avoids treating products of scientific covariates and a stationary residual field as stationary. The ordinary least-squares reduction, calibration algorithms, Riesz representations and regression dualities are established ingredients, not new estimation methods. Same-label selection is justified only for a fixed finite, predeclared library with asymptotically common influence functions; the result does not allow unrestricted adaptive discovery. The protected space need only approximate the residual mean. Gaussian Fourier and DCT models, including centred coordinate slopes in the latter, are special cases. The periodic non-diagonal and binary examples require their stated smooth designs; no result verifies the conditions for irregular LUCAS sites or discontinuous map boundaries.

The empirical conclusions remain limited. Coefficient-specific targeting improves on some common-weight spatial corrections at the same basis size, but simple methods remain competitive and additional augmentation is not uniformly helpful. The descriptive coordinate regression provides a mechanism check on a fixed observed-site frame. Generalization to a different scientific response or population requires its own measurement definition, acquisition design and validation.

\section*{Computational reproducibility and data}
The LUCAS and constructed experiments have been rerun using their recorded populations or seeds and shared label sets. Stored estimates, covariance matrices and summaries agree within prespecified cross-platform tolerances. Separate singular value decomposition (SVD) calculations, inclusion derivatives, deletion identities and deterministic Fourier checks provide additional numerical verification. Failed or undefined confidence regions remain in the reported denominators. These checks are distinct from the mathematical proofs.

The analysis code, run specifications, software versions and reproducibility instructions are publicly available at \url{https://github.com/ShinShirota/spatial-ppi} under the BSD-3-Clause licence. The LUCAS and WorldCover source datasets are identified in the references; their reuse conditions remain separate from the software licence. The processed analysis inputs, saved label draws and reference outputs are distributed with checksums in the data bundle attached to \href{https://github.com/ShinShirota/spatial-ppi/releases/tag/v0.1.0-candidate}{release v0.1.0-candidate}. The repository documents both reconstruction of the analysis frame from public sources and use of this bundle for replay and reference comparisons. The growing-radius pilot code and complete synthetic outputs for Supplement S6.6 are provided separately in \href{https://github.com/ShinShirota/spatial-ppi/releases/tag/v0.2.0-candidate}{release v0.2.0-candidate}. These releases have no DOI.

\section*{Declaration of the use of generative AI and AI-assisted technologies}
During the preparation of this work the author used ChatGPT and Codex to assist with mathematical derivations and their review, code development, numerical checks, and manuscript editing. After using these tools, the author reviewed and edited the content as necessary and takes full responsibility for the content of the publication.

\section*{Funding}
This work was supported by JSPS KAKENHI Grant Numbers JP24K00244 and JP26K14728.

\section*{Competing interests}
The author declares no competing interests.

\section*{Supplementary material}
The Supplementary Material contains proofs of the conditional-coverage and spectral-basis results, the two restricted periodic spatial-model examples, additional numerical results, and computational verification details.

\begingroup
\raggedright

\endgroup

\end{document}


\maketitle
\noindent This supplement provides the proofs for fixed-dimensional regression targets, additional empirical results, and numerical verification details. The residual field is kept separate from the deterministic coefficient loadings; covariate multiplication is not assumed to preserve stationarity.

We use the main paper's convention of bold upright Roman matrix capitals, bold italic vectors, bold Greek matrix capitals and unbolded scalars and entries. We retain its definitions: ${\bm{\Delta}}={\bm{Y}}-{\bm{f}}$, ${\mathbf{X}}\in\mathbb R^{N\times p}$, ${\mathbf{B}}\in\mathbb R^{N\times K}$, $S=\{i:R_i=1\}$, $\pi_i=(n/N)\lambda_i$ and ${\bm{\Gamma}}=n{\bm{\Lambda}}$. A subscript $S$ restricts rows to sampled units. The vector ${\bm{\ell}}_i$ belongs to unit $i$, while ${\bm{\ell}}_{\cdot j}=(\ell_{ij})_{i=1}^N$ is the frame loading for coefficient $j$; the same dotted-column convention applies to ${\bm{q}},{\bm{q}}_+$ and ${\bm{\phi}}$. Frame-vector products and quotients are componentwise unless a matrix product is displayed.

Throughout, $\mathcal F_N$ denotes the fixed observed frame of the main paper and $\mathcal U_N=\mathcal F_N\vee\sigma({\bm{\Delta}})$ the complete finite population. The field law is $P_\xi$ and $P_p(\cdot\mid\mathcal U_N)$ is the Bernoulli label-sampling law. The joint field-and-label-sampling law is $P_{p,\xi}(A)=\E_\xi\{P_p(A\mid\mathcal U_N)\}$ for any event $A$. Acquisition intensities are fixed parameters of this law and need not be known by the feasible estimators. Candidate dictionaries and penalties are fixed before evaluation-label acquisition, and the library cardinality is bounded independently of $N$. We use $j$ for coefficient coordinates and $\nu$ for candidate indices. The nominal error probability is $\alpha_0$; ${\bm{\alpha}}$ denotes a residual-regression coefficient. The defect and loading matrices are ${\mathbf{D}}$ and ${\mathbf{L}}$, whereas effective complexity and maximum leverage are $d_{\rm eff}$ and $\ell_{\max}$. The sampled set is $S$, while $\mathcal{S}$ denotes the smoothing operator. The matrix ${\mathbf{C}}_B={\mathbf{B}}^T{\mathbf{L}}/N$ is a calibration target; in bounds, $C$ denotes a generic positive constant that may change from line to line. Unsubscripted vector norms are Euclidean, $\|\cdot\|_\op$ is the induced operator norm, and ${\mathbf{I}}$ is the identity of the required dimension. For a frame vector, $\norm{\bm{u}}_N^2=N^{-1}{\bm{u}}^T{\bm{u}}$; $\norm{\cdot}_F$ denotes the Frobenius norm, whereas ${\mathbf{F}}$ without a norm subscript denotes a contrast matrix.

\section{Target geometry and reproduction}
Let ${\mathbf{J}}={\mathbf{X}}^T{\mathbf{X}}/N$ be positive definite, ${\mathbf{L}}={\mathbf{X}}{\mathbf{J}}^{-1}$, ${\bm{\ell}}_i={\mathbf{J}}^{-1}{\bm{x}}_i$, and ${\mathbf{C}}_B={\mathbf{B}}^T{\mathbf{L}}/N$. The target is ${\bm{\beta}}_U={\mathbf{L}}^T{\bm{Y}}/N$ and ${\bm{\beta}}_f={\mathbf{L}}^T{\bm{f}}/N$. Any PPI estimate $\widehat{\bm\beta}=\bm\beta_f+\widehat{\bm c}$, where $\widehat{\bm c}\in\mathbb R^p$ is a residual correction, has error $\widehat{\bm c}-{\mathbf L}^T\bm\Delta/N$. The inverse Gram is known because ${\mathbf{X}}$ is observed everywhere; nevertheless rate statements require its smallest eigenvalue to be bounded away from zero.

If ${\mathbf{T}}={\mathbf{B}}_S^T{\mathbf{B}}_S+{\bm{\Gamma}}$ is nonsingular and ${\bm{\Delta}}={\mathbf{B}}{\bm{\gamma}}$ with ${\bm{\Gamma}}{\bm{\gamma}}=0$, then ${\mathbf{T}}{\bm{\gamma}}={\mathbf{B}}_S^T{\bm{\Delta}}_S$. Thus $\widehat{\bm{\alpha}}={\bm{\gamma}}$, the one-step term vanishes, and ${\mathbf{C}}_B^T\widehat{\bm{\alpha}}={\mathbf{L}}^T{\bm{\Delta}}/N$. This proves hard-space reproduction, including after correction. It does not show unbiasedness for arbitrary residuals. For fixed ${\mathbf{B}}$ and ${\bm{\Gamma}}$, replacing ${\bm{f}}$ by ${\bm{f}}+{\mathbf{B}}{\bm{\gamma}}$ with ${\bm{\Gamma}}{\bm{\gamma}}=0$ changes ${\bm{\beta}}_f$ by ${\mathbf{C}}_B^T{\bm{\gamma}}$ and $\widehat{\bm{\alpha}}$ by $-{\bm{\gamma}}$. These changes cancel in both estimates, and the one-step term is unchanged. This identity holds under continuous perturbations of the inclusion weights wherever ${\mathbf{T}}$ is nonsingular, so the complete inclusion derivative is unchanged as well.

For a slope $j$, let ${\bm{x}}_{\cdot j}$ denote column $j$ of ${\mathbf{X}}$ and ${\mathbf{X}}_{-j}$ its remaining columns. Let ${\bm{v}}_j$ be the residual of ${\bm{x}}_{\cdot j}$ on ${\mathbf{X}}_{-j}$ in the full population. The vector ${\bm{v}}_j/({\bm{v}}_j^T{\bm{v}}_j/N)$ is in the span of ${\mathbf{X}}$, has full-frame inner product one with ${\bm{x}}_{\cdot j}$ and zero with the other columns, and therefore equals column $j$ of ${\mathbf{L}}$. This is the usual partial-regression identity. A small residual norm is a scientific design problem, not something calibration necessarily fixes.

Under ${\mathbf{X}}\mapsto {\mathbf{X}}{\mathbf{R}}$ with nonsingular ${\mathbf{R}}$, ${\mathbf{J}}\mapsto {\mathbf{R}}^T{\mathbf{J}}{\mathbf{R}}$, ${\mathbf{L}}\mapsto {\mathbf{L}}{\mathbf{R}}^{-T}$ and ${\mathbf{C}}_B\mapsto {\mathbf{C}}_B{\mathbf{R}}^{-T}$. Consequently every estimate in the targeted family transforms to ${\mathbf{R}}^{-1}\widehat{\bm{\beta}}$, its full covariance to ${\mathbf{R}}^{-1}\widehat {\mathbf{V}}{\mathbf{R}}^{-T}$, and an equivalent full-vector Wald statistic is unchanged. A contrast matrix must be transformed along with its meaning. Under ${\mathbf{B}}\mapsto {\mathbf{B}}{\mathbf{H}}$ the equivalent penalty is ${\mathbf{H}}^T{\bm{\Gamma}} {\mathbf{H}}$, not the original numerical diagonal. These properties provide finite-sample implementation checks, not new statistical optimality results.

\section{Exact vector expansion}
Assume ${\mathbf{B}}^T{\mathbf{B}}/N={\mathbf{I}}_K$, ${\mathbf{B}}$ contains constants, and the penalty is zero on them. Put ${\mathbf{D}}_\lambda=\diag({\bm{\lambda}})$ and
\[
 {\mathbf{A}}={\mathbf{B}}^T{\mathbf{D}}_\lambda {\mathbf{B}}/N,\quad {\mathbf{M}}={\mathbf{A}}+{\bm{\Lambda}},\quad
 {\bm{\alpha}}={\mathbf{M}}^{-1}{\mathbf{B}}^T{\mathbf{D}}_\lambda{\bm{\Delta}}/N,\quad {\mathbf{U}}={\mathbf{M}}^{-1}{\mathbf{C}}_B.
\]
These are the population counterparts of the sample normal equations. Since $\pi_i=(n/N)\lambda_i$, conditional on $\mathcal U_N$,
\[
 \E_p\left(\frac{\mathbf B_S^T\mathbf B_S}{n}\right)=\mathbf A,
 \qquad
 \E_p\left(\frac{\mathbf B_S^T\bm\Delta_S}{n}\right)
 =\frac{\mathbf B^T\mathbf D_\lambda\bm\Delta}{N}.
\]
Replacing the sample Gram and score by these expectations, with $\bm\Gamma/n=\bm\Lambda$, gives $\mathbf M\bm\alpha=\mathbf B^T\mathbf D_\lambda\bm\Delta/N$. Thus $\bm\alpha$ solves the averaged normal equations; it is generally not $\E_p\widehat{\bm\alpha}$. The matrices $\mathbf U,\mathbf U_+$ express the resulting coefficient correction as weights on residuals.

The matrix ${\mathbf{U}}$ is $K\times p$, whereas ${\bm{\alpha}}$ is $K\times1$. Define
\[
 {\mathbf{U}}_+={\mathbf{U}}+{\mathbf{M}}^{-1}{\bm{\Lambda}} {\mathbf{U}},\quad {\bm{q}}_i={\mathbf{U}}^T{\bm{b}}_i,\quad {\bm{q}}_{+,i}={\mathbf{U}}_+^T{\bm{b}}_i,
 \quad e_i=\Delta_i-{\bm{b}}_i^T{\bm{\alpha}},\quad r_i={\bm{b}}_i^T{\mathbf{M}}^{-1}{\bm{\Lambda}}{\bm{\alpha}}.
\]
The population-normal-equation target is ${\bm{\beta}}_f+{\mathbf{C}}_B^T{\bm{\alpha}}+{\mathbf{U}}^T{\bm{\Lambda}}{\bm{\alpha}}={\bm{\beta}}_U+{\bm{G}}$ with
\[
 {\bm{G}}={\mathbf{D}}^T{\bm{\Delta}}/N,\quad {\bm{d}}_i^T=(\lambda_i{\bm{q}}_{+,i}-{\bm{\ell}}_i)^T.
\]
The equality follows by writing ${\mathbf{C}}_B^T{\bm{\alpha}}+{\mathbf{U}}^T{\bm{\Lambda}}{\bm{\alpha}}={\mathbf{U}}_+^T{\mathbf{B}}^T{\mathbf{D}}_\lambda{\bm{\Delta}}/N$. The ${\bm{q}}$ functions are fixed under the field law; the fitted population coefficient ${\bm{\alpha}}$ is not. Since $\mathbf D$ is fixed under that law, the covariance of the linear map $\mathbf D^T\bm\Delta/N$ is $\mathbf D^T\bm\Sigma_\Delta\mathbf D/N^2$.

Let $\widehat {\mathbf{A}}={\mathbf{B}}_S^T{\mathbf{B}}_S/n$, ${\mathbf{E}}=\widehat {\mathbf{A}}-{\mathbf{A}}$, $\widehat {\mathbf{M}}={\mathbf{M}}+{\mathbf{E}}$ and
\[
 {\bm{g}}_e=n^{-1}\sum_i(R_i-\pi_i){\bm{b}}_ie_i,\quad
 {\bm{g}}_r=n^{-1}\sum_i(R_i-\pi_i){\bm{b}}_ir_i,\quad {\mathbf{G}}_q={\mathbf{E}}{\mathbf{U}},\quad {\mathbf{G}}_{q+}={\mathbf{E}}{\mathbf{U}}_+.
\]
Write $\widehat {\mathbf{U}}=\widehat {\mathbf{M}}^{-1}{\mathbf{C}}_B$ and $\widehat {\mathbf{U}}_+=\widehat {\mathbf{U}}+\widehat {\mathbf{M}}^{-1}{\bm{\Lambda}}\widehat {\mathbf{U}}$. The following is an algebraic identity, conditional on nonsingular $\widehat {\mathbf{M}}$:
\begin{align}
 \widehat{\bm{\beta}}_+-{\bm{\beta}}_U
 &={\bm{G}}+n^{-1}\sum_i(R_i-\pi_i){\bm{\phi}}_i+{\bm{\mathcal R}},\qquad
 {\bm{\phi}}_i={\bm{q}}_{+,i}e_i-{\bm{q}}_ir_i,\label{eq:expansion}\\
 {\bm{\mathcal R}}
 &=-{\mathbf{G}}_{q+}^T\widehat {\mathbf{M}}^{-1}{\bm{g}}_e
   +{\mathbf{G}}_q^T\widehat {\mathbf{M}}^{-1}{\bm{g}}_r
   -{\mathbf{G}}_q^T\widehat {\mathbf{M}}^{-1}{\bm{\Lambda}}\widehat {\mathbf{M}}^{-1}{\bm{g}}_e.\label{eq:exactR}
\end{align}
To prove it, ${\bm{d}}=\widehat{\bm{\alpha}}-{\bm{\alpha}}=\widehat {\mathbf{M}}^{-1}{\bm{g}}_e$ and ${\mathbf{E}}_U=\widehat {\mathbf{U}}-{\mathbf{U}}=-\widehat {\mathbf{M}}^{-1}{\mathbf{E}}{\mathbf{U}}$. Expanding ${\mathbf{C}}_B^T\widehat{\bm{\alpha}}+\widehat {\mathbf{U}}^T{\bm{\Lambda}}\widehat{\bm{\alpha}}$ about $({\bm{\alpha}},{\mathbf{U}})$ gives ${\mathbf{C}}_B^T{\bm{d}}+{\mathbf{E}}_U^T{\bm{\Lambda}}{\bm{\alpha}}+{\mathbf{U}}^T{\bm{\Lambda}} {\bm{d}}+{\mathbf{E}}_U^T{\bm{\Lambda}} {\bm{d}}$. Substitute ${\mathbf{M}}{\bm{d}}={\bm{g}}_e-{\mathbf{E}}{\bm{d}}$ and ${\mathbf{M}}{\mathbf{E}}_U=-{\mathbf{E}}{\mathbf{U}}-{\mathbf{E}}{\mathbf{E}}_U$. The linear part is ${\mathbf{U}}_+^T{\bm{g}}_e-{\mathbf{U}}^T{\mathbf{E}}{\mathbf{M}}^{-1}{\bm{\Lambda}}{\bm{\alpha}}$, and the remaining three products are exactly \eqref{eq:exactR}. Neither independence of the products' factors nor commutation of ${\mathbf{A}},{\bm{\Lambda}}$ is used.

It is sometimes tempting to replace the corrected estimator by its effective-ridge representation and then treat the effective penalty as fixed. That would omit the second term of ${\bm{\phi}}_i$, since the effective penalty depends on the sample Gram. Formula \eqref{eq:expansion} differentiates the entire implemented procedure instead.

\paragraph{Derivative with respect to inclusion weights.}
To obtain the implemented derivative, replace $R_i$ by continuous weights $\omega_i$ while holding $\bm\Gamma$ fixed. Define
\[
 \mathbf T(\bm\omega)=\sum_i\omega_i\bm b_i\bm b_i^T+\bm\Gamma,
 \qquad
 \widehat{\bm\alpha}(\bm\omega)=\mathbf T(\bm\omega)^{-1}
                  \sum_i\omega_i\bm b_i\Delta_i.
\]
Write $\partial_i=\partial/\partial\omega_i$ and suppress the argument $\bm\omega$. Wherever $\mathbf T$ is nonsingular, differentiation gives
\[
 \partial_i\mathbf T^{-1}=-\mathbf T^{-1}\bm b_i\bm b_i^T\mathbf T^{-1},
 \qquad
 \partial_i\widehat{\bm\alpha}=\mathbf T^{-1}\bm b_i
                    (\Delta_i-\bm b_i^T\widehat{\bm\alpha}).
\]
Applying the product rule to the entire correction,
\begin{align*}
 \partial_i\{\bm\beta_f+\mathbf C_B^T\widehat{\bm\alpha}
                +\mathbf C_B^T\mathbf T^{-1}\bm\Gamma\widehat{\bm\alpha}\}
 &=\mathbf C_B^T(\mathbf T^{-1}+\mathbf T^{-1}\bm\Gamma\mathbf T^{-1})
                  \bm b_i\widehat e_i\\
 &\quad-\mathbf C_B^T\mathbf T^{-1}\bm b_i\bm b_i^T
                  \mathbf T^{-1}\bm\Gamma\widehat{\bm\alpha},
\end{align*}
where $\widehat e_i=\Delta_i-\bm b_i^T\widehat{\bm\alpha}$. At $\omega_i=R_i$, this is $\widehat{\bm\psi}_i$ in the main paper. Both the fitted residual coefficient and the inverse Gram vary; the second term cannot be omitted. Section S3 verifies the covariance estimator formed from these derivatives.

\section{Sampling limits and feasible matrix covariance}
Let $d_{\rm eff}=\tr({\mathbf{A}}{\mathbf{M}}^{-1})$ and $\ell_{\max}=\max_i {\bm{b}}_i^T{\mathbf{M}}^{-1}{\bm{b}}_i$. Constant protection implies $d_{\rm eff}\ge1$. For the following conditional argument assume fixed $p$, bounded ${\bm{\lambda}}$ above and away from zero, bounded coordinates of ${\bm{q}},{\bm{q}}_+$, and, for some fixed $\delta>0$,
\[
 N^{-1}\sum_i(|e_i|^{2+\delta}+|r_i|^{2+\delta})\le C,\quad
 c_V {\mathbf{I}}_p\preceq n{\mathbf{V}}\preceq C_V {\mathbf{I}}_p,
\]
where $0<c_V\le C_V<\infty$ and
\[
 {\mathbf{V}}=n^{-2}\sum_i\pi_i(1-\pi_i){\bm{\phi}}_i{\bm{\phi}}_i^T.
\]
For a field-random sequence, these constants are localized on $\mathcal U_N$-measurable events: for every $\eta>0$ there are common constants with limiting event probability at least $1-\eta$, uniformly over candidates. Condition (T) in the main paper quantifies the moment event; Section~\ref{sec:general-proof} transfers (ND) to the scale event before this sampling argument is invoked. Take the sequence limit on each event and then let $\eta\downarrow0$. This is stronger than finiteness at one realized population. The bounded-response alternative is proved separately in Lemma~\ref{lem:bounded-sampling}.

Write ${\mathbf{Z}}={\mathbf{M}}^{-1/2}{\mathbf{E}}{\mathbf{M}}^{-1/2}$ and ${\mathbf{P}}={\mathbf{M}}^{-1/2}{\bm{\Lambda}} {\mathbf{M}}^{-1/2}$, so $0\preceq {\mathbf{P}}\preceq {\mathbf{I}}$. Independence of labels and $\pi_i=(n/N)\lambda_i$ give
\[
 \E_p\norm{\mathbf{Z}}_F^2
 =n^{-2}\sum_i\pi_i(1-\pi_i)({\bm{b}}_i^T{\mathbf{M}}^{-1}{\bm{b}}_i)^2
 \le d_{\rm eff}\ell_{\max}/n.
\]
If $d_{\rm eff}\ell_{\max}/n\to0$, then $({\mathbf{I}}+{\mathbf{Z}})^{-1}$ has bounded norm with probability tending to one. Further,
\begin{align*}
 \norm{{\mathbf{M}}^{-1/2}{\mathbf{G}}_q}_F+\norm{{\mathbf{M}}^{-1/2}{\mathbf{G}}_{q+}}_F
 &=O_p(\sqrt{d_{\rm eff}/n}),\\
 \norm{{\mathbf{M}}^{-1/2}{\bm{g}}_e}+\norm{{\mathbf{M}}^{-1/2}{\bm{g}}_r}
 &=O_p(\sqrt{\ell_{\max}/n}).
\end{align*}
The first rate uses bounded target weights and fixed $p$; the second uses mean-square residual bounds. Each product in \eqref{eq:exactR} is therefore $O_p(\sqrt{d_{\rm eff}\ell_{\max}}/n)$.

For plug-in covariance let ${\bm{\kappa}}={\mathbf{M}}^{-1}{\bm{\Lambda}}{\bm{\alpha}}$ and set $\widehat{\bm{\kappa}}=\widehat {\mathbf{M}}^{-1}{\bm{\Lambda}}\widehat{\bm{\alpha}}$, $\widehat {\bm{q}}_i=\widehat {\mathbf{U}}^T{\bm{b}}_i$, $\widehat {\bm{q}}_{+,i}=\widehat {\mathbf{U}}_+^T{\bm{b}}_i$, $\widehat e_i=\Delta_i-{\bm{b}}_i^T\widehat{\bm{\alpha}}$, $\widehat r_i={\bm{b}}_i^T\widehat{\bm{\kappa}}$, and $\widehat{\bm{\phi}}_i=\widehat {\bm{q}}_{+,i}\widehat e_i-\widehat {\bm{q}}_i\widehat r_i$. Exact coefficient differences satisfy
\begin{align*}
 \widehat{\bm{\alpha}}-{\bm{\alpha}}&=\widehat {\mathbf{M}}^{-1}{\bm{g}}_e,\\
 \widehat{\bm{\kappa}}-{\bm{\kappa}}&=-\widehat {\mathbf{M}}^{-1}{\bm{g}}_r
                 +\widehat {\mathbf{M}}^{-1}{\bm{\Lambda}}\widehat {\mathbf{M}}^{-1}{\bm{g}}_e,\\
 \widehat {\mathbf{U}}_+-{\mathbf{U}}_+&=-\widehat {\mathbf{M}}^{-1}{\mathbf{G}}_{q+}
                 -\widehat {\mathbf{M}}^{-1}{\bm{\Lambda}}\widehat {\mathbf{M}}^{-1}{\mathbf{G}}_q.
\end{align*}
They imply whitened regression errors of size $O_p(\sqrt{\ell_{\max}/n})$ and maximum coordinate errors in $\widehat {\bm{q}},\widehat {\bm{q}}_+$ of size $O_p(\sqrt{d_{\rm eff}\ell_{\max}/n})$. Sample fitted-square errors are $O_p(\ell_{\max}/n)=o_p(1)$, since $d_{\rm eff}\ge1$. Expanding
\[
 \widehat{\bm{\phi}}_i-{\bm{\phi}}_i=(\widehat {\bm{q}}_{+,i}-{\bm{q}}_{+,i})e_i+\widehat {\bm{q}}_{+,i}(\widehat e_i-e_i)
                   -(\widehat {\bm{q}}_i-{\bm{q}}_i)r_i-\widehat {\bm{q}}_i(\widehat r_i-r_i)
\]
coordinatewise yields $n^{-1}\sum_iR_i\norm{\widehat{\bm{\phi}}_i-{\bm{\phi}}_i}^2=o_p(1)$. The $e_i$ and $r_i$ sample mean squares are bounded in probability; cross-products use Cauchy--Schwarz. No residual term is suppressed.

Put $\pi_{\max}=\max_i\pi_i$. For the oracle non-FPC matrix $\widetilde {\mathbf{V}}=n^{-2}\sum_iR_i{\bm{\phi}}_i{\bm{\phi}}_i^T$, every quadratic form obeys
\[
 0\preceq \E_p\widetilde {\mathbf{V}}-{\mathbf{V}}
 =n^{-2}\sum_i\pi_i^2{\bm{\phi}}_i{\bm{\phi}}_i^T
 \preceq\frac{\pi_{\max}}{1-\pi_{\max}}{\mathbf{V}}.
\]
Its relative bias vanishes because $\pi_{\max}=O(n/N)\to0$. In the Gaussian special case verified in Lemma~\ref{lem:containment}, localized fourth moments suffice: for every matrix entry,
\[
 \operatorname{Var}_p\{n(\widetilde V_{ab}-\E_p\widetilde V_{ab})\}
 \le \frac{C}{nN}\sum_i\phi_{ia}^2\phi_{ib}^2=O(n^{-1}).
\]
 More generally, put $\tau=\min(\delta,2)$. The bounded weights and frame $(2+\delta)$-moment assumption give, by the independent-array inequality with exponent $1+\tau/2$, $\E_p|n(\widetilde V_{ab}-\E_p\widetilde V_{ab})|^{1+\tau/2}=O(n^{-\tau/2})$. Thus each unscaled entry fluctuates by $o_p(n^{-1})$. To justify the independent-sum bound, for centered independent $X_i$, independent copies $X_i'$ and $1\le u\le2$, symmetrization and conditional Rademacher signs give
\[
 \E\left|\sum_iX_i\right|^u
 \le\E\left|\sum_i(X_i-X_i')\right|^u
 \le\E\left\{\sum_i(X_i-X_i')^2\right\}^{u/2}
 \le2^u\sum_i\E|X_i|^u.
\]
Apply this with $u=1+\tau/2$ and $X_i=n^{-1}(R_i-\pi_i)\phi_{ia}\phi_{ib}$. This route asserts convergence, not the sharper bounded-response fluctuation rate below. For a $p\times p$ matrix, the operator norm is bounded by the Frobenius norm, which is at most $p$ times the largest absolute entry. Since $p$ is fixed, entrywise convergence therefore implies operator-norm convergence. For the fitted part, the empirical Cauchy--Schwarz bound is
\[
 n\norm{\widehat {\mathbf{V}}-\widetilde {\mathbf{V}}}_\op
 \le \left\{\frac1n\sum_iR_i\norm{\widehat{\bm{\phi}}_i-{\bm{\phi}}_i}^2\right\}^{1/2}
      \left\{\frac1n\sum_iR_i(\norm{\widehat{\bm{\phi}}_i}+\norm{{\bm{\phi}}_i})^2\right\}^{1/2}
 =o_p(1).
\]
The preceding plug-in estimate gives $n\norm{\widehat {\mathbf{V}}-{\mathbf{V}}}_\op=o_p(1)$ for fixed $p$. Notice that $\widehat {\mathbf{V}}=n^{-2}\sum_S\widehat{\bm{\phi}}_i\widehat{\bm{\phi}}_i^T$ equals the unnormalized inclusion-derivative formula in the main paper.

For a unit vector ${\bm{u}}$, define the whitened summands
\[
 Z_i=\frac{R_i-\pi_i}{\sqrt n}\,
      {\bm{u}}^T(n{\mathbf{V}})^{-1/2}{\bm{\phi}}_i.
\]
Their conditional variances sum to one. With $\tau=\min(\delta,2)>0$, the Bernoulli bound $\E_p|R_i-\pi_i|^{2+\tau}\le2\pi_i$ gives, on the localized moment and eigenvalue events,
\[
 \sum_i\E_p|Z_i|^{2+\tau}
 \le C n^{-1-\tau/2}\sum_i\pi_i\norm{{\bm{\phi}}_i}^{2+\tau}
 \le Cn^{-\tau/2}\frac1N\sum_i\norm{{\bm{\phi}}_i}^{2+\tau}
 =O(n^{-\tau/2}).
\]
Moreover $\sqrt n{\bm{\mathcal R}}=O_p(\sqrt{d_{\rm eff}\ell_{\max}/n})=o_p(1)$. Fixed-dimensional Cram\'er--Wold and the nonsingular covariance bound prove
\[
 {\mathbf{V}}^{-1/2}(\widehat{\bm{\beta}}_+-{\bm{\beta}}_U-{\bm{G}})\Rightarrow N_p(0,{\mathbf{I}}_p).
\]
The remainder is negligible after whitening and $\widehat {\mathbf{V}}$ may replace ${\mathbf{V}}$. This argument can be localized uniformly in the stated moments and covariance bounds, yielding the conditional normal approximation in $P_\xi$ probability.

\begin{lemma}[Bounded-response sampling route]\label{lem:bounded-sampling}
Retain the fixed dimension, bounded intensities and target weights, and localized covariance eigenvalue bounds used above. Replace the residual-moment hypothesis by $\sup_i|\Delta_i|\le C_\Delta$. If $d_{\rm eff}\ell_{\max}/n\to0$ and $n/N\to0$, the same centered conditional CLT and remainder order hold, and
\[
 n\|\widehat{\mathbf V}-\mathbf V\|_\op
 =O_p\left\{\sqrt{d_{\rm eff}\ell_{\max}/n}
       +\sqrt{(1+\ell_{\max})/n}+n/N\right\}=o_p(1).
\]
The constants are uniform on each scale-localization event.
\end{lemma}
\begin{proof}
Weighted contraction of $\mathcal S$ gives $\|e\|_\lambda+\|r\|_\lambda\le C$.
The row-energy inequality $|\mathcal Sz(i)|\le\sqrt{\ell_{\max}}\|z\|_\lambda$ yields
$\max_i(|e_i|+|r_i|)\le C(1+\sqrt{\ell_{\max}})$, not a fixed maximum bound. Bounded $q,q_+$ therefore imply
\[
 N^{-1}\sum_i\lambda_i\|\bm\phi_i\|^2\le C,\qquad
 \max_i\|\bm\phi_i\|^2\le C(1+\ell_{\max}).
\]
The Gram and score second moments above use only these mean-square bounds; hence the exact three products in \eqref{eq:exactR} have the stated order. On $\|\mathbf Z\|_\op\le1/2$, whose complement has probability at most $4d_{\rm eff}\ell_{\max}/n$, the inverse normalized sample Gram has norm at most two. The coefficient identities above give sampled fitted-square errors $O_p(\ell_{\max}/n)$ and maximum weight errors $O_p(\sqrt{d_{\rm eff}\ell_{\max}/n})$.
Since $\E_p[n^{-1}\sum R_i(e_i^2+r_i^2)]\le C$, Markov and those identities imply
\[
 n^{-1}\sum_iR_i\|\widehat{\bm\phi}_i-\bm\phi_i\|^2
 =O_p(d_{\rm eff}\ell_{\max}/n).
\]
For $\widetilde{\mathbf V}=n^{-2}\sum R_i\bm\phi_i\bm\phi_i^T$,
$\Var_p(n\widetilde V_{ab})\le C(1+\ell_{\max})/n$; the FPC bias is $O(n/N)$ after scaling by $n$. Empirical Cauchy--Schwarz bounds the fitted-covariance difference by $O_p(\sqrt{d_{\rm eff}\ell_{\max}/n})$. This proves the displayed rate without assuming bounded fourth moments of $e,r$.
On the localized eigenvalue event, each whitened scalar Bernoulli summand has magnitude at most $C_\eta\sqrt{(1+\ell_{\max})/n}\to0$, and the variances sum to one. Lindeberg and fixed-dimensional Cram\'er--Wold give the conditional CLT. The $\sqrt n$-scaled remainder vanishes by the same complexity condition.
\end{proof}

\subsection{Full-vector and fixed-contrast gap criteria}
As in the main paper, $\chi_p^2(t)$ denotes the noncentral chi-square law with $p$ degrees of freedom and noncentrality $t$, and $\chi^2_{p,1-\alpha_0}$ its central $(1-\alpha_0)$ quantile. For a fixed candidate let ${\bm{b}}_G=\E_\xi {\bm{G}}$ and ${\bm{\Omega}}_G=\Cov_\xi({\bm{G}})$. The full-vector Wald coverage equals the noncentral chi-square probability
\[
 \Pr\{\chi^2_p(\delta_N^2)\le\chi^2_{p,1-\alpha_0}\}+o_{P_\xi}(1),
 \qquad \delta_N^2={\bm{G}}^T{\mathbf{V}}^{-1}{\bm{G}}.
\]
Here ${\mathbf{V}}$ may depend on the realized residuals, so $\delta_N$ need not be Gaussian. We give the conditional argument explicitly. Put
\[
 {\bm{W}}_N={\mathbf{V}}^{-1/2}(\widehat{\bm{\beta}}_+-{\bm{\beta}}_U-{\bm{G}}),\qquad
 {\bm{a}}_N={\mathbf{V}}^{-1/2}{\bm{G}},\qquad
 {\mathbf{H}}_N={\mathbf{V}}^{-1/2}\widehat {\mathbf{V}} {\mathbf{V}}^{-1/2}.
\]
Conditional on $\mathcal U_N$, both ${\bm{a}}_N$ and ${\mathbf{V}}$ are fixed; no independence between ${\bm{G}}$ and ${\mathbf{V}}$ is needed. On localized eigenvalue events,
\[
 \norm{{\mathbf{H}}_N-{\mathbf{I}}}_\op
 \le \frac{n\norm{\widehat {\mathbf{V}}-{\mathbf{V}}}_\op}{\lambda_{\min}(n{\mathbf{V}})}=o_p(1).
\]
This also implies that $\widehat {\mathbf{V}}$ is positive definite with conditional probability tending to one in $P_\xi$ probability.

The centred conditional CLT above, including the negligible remainder, implies for every fixed $R<\infty$
\[
 \sup_{\bm{a}\in\mathbb R^p,\,0\le r\le R}
 \left|P_p\{\norm{{\bm{W}}_N+{\bm{a}}}\le r\mid\mathcal U_N\}
       -\Pr\{\norm{{\bm{Z}}+{\bm{a}}}\le r\}\right|=o_{P_\xi}(1),
 \qquad {\bm{Z}}\sim N_p(0,{\mathbf{I}}_p).
\]
To see the required uniformity, first restrict centres to a compact set. A finite grid of centres and radii sandwiches each ball between nearby balls; Gaussian probabilities of the intervening shells vanish uniformly as the grid spacing tends to zero. The latter follows from continuity on the compact parameter set and zero Gaussian mass on every sphere. Outside a large compact set, a ball of radius at most $R$ lies in the tail, and conditional tightness bounds both probabilities. Applying this deterministic weak-convergence argument along almost-surely convergent subsubsequences gives the stated convergence in field probability.

Let $c=\chi^2_{p,1-\alpha_0}$ and $0<\epsilon<1$. On $\norm{{\mathbf{H}}_N-{\mathbf{I}}}_\op\le\epsilon$ the Wald event is squeezed between
\[
 \{\norm{{\bm{W}}_N+{\bm{a}}_N}^2\le c(1-\epsilon)\}
 \ \subseteq\
 \{({\bm{W}}_N+{\bm{a}}_N)^T{\mathbf{H}}_N^{-1}({\bm{W}}_N+{\bm{a}}_N)\le c\}
 \ \subseteq\
 \{\norm{{\bm{W}}_N+{\bm{a}}_N}^2\le c(1+\epsilon)\}.
\]
The complement has vanishing conditional probability. The uniform ball approximation applies even to unbounded ${\bm{a}}_N$; the same compact-set and tail argument makes the Gaussian shell probability between these two radii vanish uniformly in ${\bm{a}}_N$ as $\epsilon\downarrow0$. This proves the displayed noncentral approximation without multiplying covariance estimation error by an unbounded noncentrality.

For $0<\alpha_0<1$, write $h(t)=\Pr\{\chi_p^2(t)\le c\}$. This is continuous and strictly decreasing for $t\ge0$, with $h(0)=1-\alpha_0$. Thus $h(\delta_N^2)\to h(0)$ in field probability if and only if $\delta_N^2\to0$: for any $\eta>0$, $h(0)-h(t)\ge h(0)-h(\eta)>0$ whenever $t\ge\eta$. On an event where $c_V{\mathbf{I}}\preceq n{\mathbf{V}}\preceq C_V{\mathbf{I}}$,
\[
 \frac{n\norm {\bm{G}}^2}{C_V}\le\delta_N^2\le\frac{n\norm {\bm{G}}^2}{c_V}.
\]
Localizing on events of arbitrarily high field probability proves that, under the stated fixed-candidate sampling conditions, nominal full-vector conditional coverage is equivalent to $\sqrt n{\bm{G}}\to0$ in $P_\xi$ probability. A sufficient condition is
\[
 n\norm{{\bm{b}}_G}^2+n\tr{\bm{\Omega}}_G\to0.
\]
For Gaussian ${\bm{\Delta}}$ and deterministic candidate weights, $\sqrt n{\bm{G}}$ is Gaussian with deterministic mean and covariance. Its convergence to zero is equivalent to both the displayed terms vanishing (apply the univariate fact to each coordinate; $p$ is fixed). This gives necessity for a fixed full-vector procedure. A single contrast requires only the corresponding projected mean and covariance; validity of selected slopes does not imply validity of the intercept or all coefficients.

\paragraph{Fixed-rank contrasts.}
Let \(\mathbf F\) be fixed and of full row rank, and define
\[
\mathbf A_N=(\mathbf F\mathbf V\mathbf F^\top)^{-1/2}
\mathbf F\mathbf V^{1/2},\qquad
\bm a_{F,N}=(\mathbf F\mathbf V\mathbf F^\top)^{-1/2}
\mathbf F\bm G.
\]
Since \(\mathbf A_N\mathbf A_N^\top=\mathbf I_r\), \(\|\mathbf A_N\|_{\mathrm{op}}=1\). The bounded-Lipschitz distance cannot increase under a linear map of operator norm one, and every such row-orthonormal map sends a standard Gaussian vector to a standard Gaussian vector in the contrast dimension. This uniform bound transfers the centred full-vector conditional Gaussian approximation to the contrast, even though \(\mathbf A_N\) may depend on the realised population. The same uniform-ball argument used above applies to the possibly unbounded shift \(\bm a_{F,N}\). Moreover,
\[
\mathbf H_{F,N}=(\mathbf F\mathbf V\mathbf F^\top)^{-1/2}
\mathbf F\widehat{\mathbf V}\mathbf F^\top
(\mathbf F\mathbf V\mathbf F^\top)^{-1/2}
\longrightarrow\mathbf I_r
\]
in the required conditional sense, because \(\mathbf F\) is fixed, \(n\|\widehat{\mathbf V}-\mathbf V\|_{\mathrm{op}}=o_p(1)\), and the eigenvalues of \(n\mathbf V\) are locally bounded above and away from zero. Consequently the Wald coverage has the asserted noncentral \(\chi_r^2\) approximation with noncentrality \(\delta_{F,N}^2=\|\bm a_{F,N}\|^2\). The noncentral \(\chi_r^2\) distribution function at the positive central quantile is continuous and strictly decreasing in its noncentrality. Finally, on the same eigenvalue-localization events,
\[
\frac{n\|\mathbf F\bm G\|^2}
{C_V\lambda_{\max}(\mathbf F\mathbf F^\top)}
\le \delta_{F,N}^2\le
\frac{n\|\mathbf F\bm G\|^2}
{c_V\lambda_{\min}(\mathbf F\mathbf F^\top)}.
\]
This proves the first equivalence of Proposition 1 of the main paper. If \(\bm\Delta\) is Gaussian and the defect matrix is fixed under its law, then \(\sqrt n\mathbf F\bm G\) is Gaussian with deterministic mean and covariance. In fixed dimension, its convergence in probability to zero is equivalent to convergence of both its mean and covariance to zero, yielding the stated moment equivalence. Neither argument treats a data-selected candidate as fixed.

\subsection{Selection from a fixed finite library}
Fix an arbitrary reference candidate $\nu_0\in\mathcal I_N$. Suppose candidate gaps are uniformly $o_{P_\xi}(n^{-1/2})$, their sampling remainders satisfy the preceding bounds, and
\[
 \max_\nu N^{-1}\sum_i\lambda_i\norm{{\bm{\phi}}_{\nu i}-{\bm{\phi}}_{\nu_0 i}}^2\to0
 \quad\text{in }P_\xi\text{-probability}.
\]
For clarity, uniformity refers to a library of cardinality at most a fixed $J_0$. The applicable moment or bounded-response route, weight bounds and covariance bounds are required for every candidate. For any $\eta>0$, choose each candidate's localization event with failure probability at most $\eta/J_0$, then intersect the events, take maxima of their upper constants and minima of their positive lower constants. The Gram, remainder and plug-in bounds above therefore hold with common constants on that intersection. The required complexity rates must also hold for the maximum over candidates; in the spectral-basis theorem they follow from the common bounds $d_{\rm eff},\ell_{\max}=O(k^{d_0})$. Finiteness alone does not replace these candidate-wise hypotheses.

Write ${\bm{h}}_{\nu i}={\bm{\phi}}_{\nu i}-{\bm{\phi}}_{\nu_0 i}$ and $a_\nu=N^{-1}\sum_i\lambda_i\norm{{\bm{h}}_{\nu i}}^2$. Conditional independence gives
\[
 \E_p\left\|\frac1{\sqrt n}\sum_i(R_i-\pi_i){\bm{h}}_{\nu i}\right\|^2
 =\frac1n\sum_i\pi_i(1-\pi_i)\norm{{\bm{h}}_{\nu i}}^2\le a_\nu.
\]
Consequently the conditional failure probability is bounded by $\min(1,a_\nu/\epsilon^2)$. Since $a_\nu\to0$ in $P_\xi$ probability, this bounded expression has expectation tending to zero; no unproved convergence of $\E_\xi a_\nu$ is required. If $b_\nu=N^{-1}\sum_i\lambda_i\norm{{\bm{\phi}}_{\nu i}}^2=O_{P_\xi}(1)$, Cauchy--Schwarz also gives
\[
 n\norm{{\mathbf{V}}_\nu-{\mathbf{V}}_{\nu_0}}_\op\le a_\nu^{1/2}(b_\nu^{1/2}+b_{\nu_0}^{1/2})=o_{P_\xi}(1).
\]
 There are finitely many candidates. Gaps and remainders are uniformly negligible, so every same-label selected estimate differs from the reference candidate by $o_p(n^{-1/2})$ in joint probability. Cauchy--Schwarz and covariance consistency give $n\norm{\widehat {\mathbf{V}}_\nu-\widehat {\mathbf{V}}_{\nu_0}}_\op\to0$. To transfer these joint bounds to the conditional statement, let $A_N$ be an event on which a displayed approximation fails by more than a fixed tolerance. Then, for every $\epsilon>0$,
\[
 P_\xi\{P_p(A_N\mid\mathcal U_N)>\epsilon\}
 \le P_{p,\xi}(A_N)/\epsilon\longrightarrow0.
\]
This is Markov's inequality applied to the conditional failure probability. In particular, for every $\epsilon>0$,
\[
 P_p\!\left\{\max_\nu n\|\widehat{\mathbf V}_\nu-\mathbf V_\nu\|_\op>\epsilon
 \mid\mathcal U_N\right\}\to0\quad\text{in }P_\xi\text{-probability}.
\]
The selected error is bounded by this maximum, regardless of its dependence on the labels. Consequently the selected ordinary marginal and fixed-rank Wald procedures inherit validity; no conditioning on the selected index is taken. This does not assert exact post-selection validity at finite $n$, or allow a growing library or unrestricted data-driven dictionary discovery.

\section{Target weights and covariance directions}
Use the weighted frame norm $\norm {\bm{h}}_\lambda^2=N^{-1}\sum_i\lambda_i h_i^2$. Let $H$ be a common protected subspace, and let ${\mathbf{B}}_H$ have columns spanning it. As in the main paper, ${\mathbf{P}}_H$ and ${\mathbf{Q}}_H={\mathbf{I}}-{\mathbf{P}}_H$ are Euclidean orthogonal projectors, and $E_{\mu,H}=\inf_{\bm{h}\in H}\norm{{\bm{\mu}}-{\bm{h}}}_N$. They are not the weighted projectors used for target approximation. All penalties vanish on $H$. For coefficient $j$, ${\bm{g}}_j={\bm{\ell}}_{\cdot j}/{\bm{\lambda}}$ and ${\mathbf{B}}^T{\mathbf{D}}_\lambda {\bm{g}}_j/N={\bm{c}}_{B,\cdot j}$, where ${\bm{c}}_{B,\cdot j}$ is column $j$ of ${\mathbf{C}}_B$. Also ${\bm{e}}=(e_i)_{i=1}^N$ and ${\bm{r}}=(r_i)_{i=1}^N$. Define $\mathcal{S}={\mathbf{B}}{\mathbf{M}}^{-1}{\mathbf{B}}^T{\mathbf{D}}_\lambda/N$. It is self-adjoint in this norm, has spectrum in $[0,1]$ and fixes $H$. Direct multiplication gives
\[
 {\bm{q}}_{\cdot j}=\mathcal{S}{\bm{g}}_j,\qquad {\bm{q}}_{+,\cdot j}=(2\mathcal{S}-\mathcal{S}^2){\bm{g}}_j,\qquad
 {\bm{e}}=({\mathbf{I}}-\mathcal{S}){\bm{\Delta}},\quad {\bm{r}}=\mathcal{S}({\mathbf{I}}-\mathcal{S}){\bm{\Delta}}.
\]
For any ${\bm{h}}\in H$, ${\bm{q}}_{+,\cdot j}-{\bm{g}}_j=-({\mathbf{I}}-\mathcal{S})^2({\bm{g}}_j-{\bm{h}})$. Hence
\[
 \norm{{\bm{q}}_{+,\cdot j}-{\bm{g}}_j}_\lambda\le\epsilon_{H,j}
 :=\inf_{\bm{h}\in H}\norm{{\bm{g}}_j-{\bm{h}}}_\lambda.
\]
For $h\in H$, weighted self-adjointness gives
\[
 N^{-1}h^T\bm d_{\cdot j}
 =-\langle h,(I-\mathcal S)^2\bm g_j\rangle_\lambda
 =-\langle(I-\mathcal S)h,(I-\mathcal S)\bm g_j\rangle_\lambda=0,
\]
where $\langle u,z\rangle_\lambda=N^{-1}u^T\mathbf D_\lambda z$. Thus the Euclidean projector satisfies $\mathbf P_H\bm d_{\cdot j}=0$, equivalently ${\mathbf B}_H^T\bm d_{\cdot j}=0$, and
$\norm{{\bm{d}}_{\cdot j}}_N^2\le\lambda_+\epsilon_{H,j}^2$. With ${\bm{\Delta}}={\bm{\mu}}+{\bm{u}}$, $\E_\xi {\bm{u}}=0$ and $\Cov_\xi({\bm{u}})={\bm{\Sigma}}$,
\[
 \E_\xi G_j^2
 =\{N^{-1}{\bm{d}}_{\cdot j}^T{\bm{\mu}}\}^2+N^{-2}{\bm{d}}_{\cdot j}^T{\bm{\Sigma}} {\bm{d}}_{\cdot j}
 \le\lambda_+\epsilon_{H,j}^2
       \{E_{\mu,H}^2+\norm{{\mathbf{Q}}_H{\bm{\Sigma}} {\mathbf{Q}}_H}_\op/N\}.
\]
Summing over fixed $p$ gives the trace/mean-square vector bound. Off-diagonal covariance entries are retained in ${\mathbf{D}}^T{\bm{\Sigma}} {\mathbf{D}}/N^2$; bounding each marginal does not make that matrix diagonal.

\subsection{Directional and residual-energy recovery conditions}
This comparison concerns two bounds on the same population displacement. Under the spectral-basis conditions write ${\bm{\Sigma}}_v=\Cov_\xi({\bm{v}}\mid\mathcal F_N)$ and ${\mathbf{Q}}={\mathbf{Q}}_{H_k}$. Since ${\bm{d}}_{\cdot j}\in H_k^\perp$, the structural term satisfies
\[
 \E_\xi (N^{-1}{\bm{d}}_{\cdot j}^T{\bm{v}})^2
 \le \norm{{\bm{d}}_{\cdot j}}_N^2\frac{\norm{{\mathbf{Q}}{\bm{\Sigma}}_v{\mathbf{Q}}}_\op}{N}
 \le Ck^{-(a+2s_*)}.
\]
Alternatively, Cauchy--Schwarz before taking expectation gives
\[
 \E_\xi (N^{-1}{\bm{d}}_{\cdot j}^T{\bm{v}})^2
 \le \norm{{\bm{d}}_{\cdot j}}_N^2\E_\xi\norm{{\mathbf{Q}}{\bm{v}}}_N^2
 =\norm{{\bm{d}}_{\cdot j}}_N^2\frac{\tr({\mathbf{Q}}{\bm{\Sigma}}_v{\mathbf{Q}})}N
 \le Ck^{-(a+2s_*-d_0)}.
\]
The last bound is the trace condition in (G); it is not substituted for the stronger operator bound in the first calculation. Both arguments retain the $\bm\zeta$ contribution bounded by $C\norm{\bm d_{\cdot j}}_N^2/N\le Ck^{-2s_*}/N$. For correlated $v,\zeta$, the variance of their sum is at most twice the sum of these two bounds. Independent Gaussian noise recovers the exact term $\sigma^2\norm{\bm d_{\cdot j}}_N^2/N$. No factor $k^{d_0}$ is lost on this term. The mean contributes at most $Ck^{-2(s_*+t)}$ to the squared gap. For the following comparison assume $t\ge a/2$, so this term is no larger in order than the directional structural bound.

Let $r_v=a+2s_*$. The residual-energy sufficient conditions require
\[
 n^{1/(r_v-d_0)}\ll k\ll n^{1/(2d_0)},
\]
which is possible if and only if $r_v>3d_0$. The directional sufficient conditions instead require
\[
 n^{1/r_v}\ll k\ll n^{1/(2d_0)},
\]
which is possible if and only if $r_v>2d_0$. Here $\ll$ means a ratio tending to zero, and $r_v>2d_0$ follows from $a>d_0$ and $s_*>d_0/2$. For fixed exponents, any power of $n$ strictly between the displayed endpoints gives a radius, with integer rounding immaterial. On the Fourier and DCT grids, $n/N\to0$ and the complexity condition also imply $k/m\to0$, so the finite grid introduces no additional obstruction.

For $d_0=2$, $a=3$ and $s_*=5/4$ (with $t\ge3/2$), one has $r_v=11/2$. Choosing $k\asymp n^{1/5}$ gives the three powers $-1/5$, $-1/10$ and $3/10$ reported in the main paper. The positive last power shows that the residual-energy upper bound cannot certify this sequence; it is not a lower bound on the actual gap. In particular, this comparison neither proves invalidity outside the sufficient window nor establishes a minimax lower bound. It isolates the benefit of retaining covariance directions in this proof.

\subsection{Uniform weight bounds need more than mean-square control}
For each ${\bm{h}}\in H$, energy of the normal equations gives
\[
 \norm{{\bm{q}}_{\cdot j}-{\bm{h}}}_\infty\le\sqrt{\ell_{\max}}\norm{{\bm{g}}_j-{\bm{h}}}_\lambda.
\]
For the corrected weight use
\[
 {\mathbf{M}}^{-1}+{\mathbf{M}}^{-1}{\bm{\Lambda}} {\mathbf{M}}^{-1}=({\mathbf{A}}+{\bm{\Lambda}}_{\rm eff})^{-1},\qquad
 {\bm{\Lambda}}_{\rm eff}={\bm{\Lambda}}({\mathbf{A}}+2{\bm{\Lambda}})^{-1}{\bm{\Lambda}},
\]
which has the same null space as ${\bm{\Lambda}}$ and inverse at most $2{\mathbf{M}}^{-1}$ in Loewner order. Thus
\[
 \norm{{\bm{q}}_{+,\cdot j}-{\bm{g}}_j}_\infty
 \le\norm{{\bm{h}}-{\bm{g}}_j}_\infty+\sqrt{2\ell_{\max}}\norm{{\bm{h}}-{\bm{g}}_j}_\lambda.
\]
The same inequality for ${\bm{q}}_{\cdot j}$ has $\sqrt{\ell_{\max}}$. It supplies the required maximum bound only when the approximating function is uniformly accurate relative to leverage. Slope weights are signed, so no positive lower bound for ${\bm{q}}_{\cdot j}$ is sought.

\subsection{Common influence vector}\label{sec:common-influence}
Use the auxiliary decomposition from the main paper and put $\bm\phi_{0i}=\bm\ell_i\zeta_i/\lambda_i$. Write $\bm w=\bm\mu-\bm h_k+\mathbf Q_k\bm v$. Since $\mathcal S$ fixes $H_k$,
\[
 \bm e-\bm\zeta=(I-\mathcal S)\bm w-\mathcal S\bm\zeta,
 \qquad \bm r=\mathcal S(I-\mathcal S)\bm w+\mathcal S(I-\mathcal S)\bm\zeta.
\]
The trace bound in (G) gives $\E_\xi\|\bm w\|_\lambda^2\le C(k^{-2t}+k^{d_0-a})$.
Let $\widetilde{\mathcal S}=\mathbf D_\lambda^{1/2}\mathcal S\mathbf D_\lambda^{-1/2}$. This matrix is symmetric with spectrum in $[0,1]$, and $\|\bm\Sigma_\zeta\|_\op\le C$ gives
\[
 \E_\xi\|\mathcal S\bm\zeta\|_\lambda^2
 =N^{-1}\tr\{\widetilde{\mathcal S}^2\mathbf D_\lambda^{1/2}
             \bm\Sigma_\zeta\mathbf D_\lambda^{1/2}\}
 \le C d_{\rm eff}/N.
\]
The same bound holds for $\mathcal S(I-\mathcal S)\bm\zeta$ by its eigenvalues. Squared-triangle inequalities, without independence, yield
$\E_\xi(\|e-\zeta\|_\lambda^2+\|r\|_\lambda^2)\le C(k^{-2t}+k^{d_0-a}+d_{\rm eff}/N)$.
For each coordinate,
\[
 \phi_{ij}-g_{ij}\zeta_i=(q_{+,ij}-g_{ij})\zeta_i
       +q_{+,ij}(e_i-\zeta_i)-q_{ij}r_i.
\]
The first term has expected weighted square at most $C\epsilon_{H_k,j}^2$, since $q_+-g$ is deterministic and $\E_\xi\zeta_i^2\le C$. Bounded $q,q_+$ control the remaining terms. Thus, for influence matrices $\bm\Psi_\nu,\bm\Psi_0$ with rows $\bm\phi_{\nu i}^T,\bm\phi_{0i}^T$,
\[
 \max_\nu\E_\xi\|\bm\Psi_\nu-\bm\Psi_0\|_{\lambda,F}^2
 \le C\{\max_j\epsilon_{H_k,j}^2+k^{-2t}+k^{d_0-a}+d_{\rm eff}/N\},
\]
where $\|\mathbf H\|_{\lambda,F}^2=N^{-1}\sum_i\lambda_i\|\bm H_i\|^2$. The stronger gap condition at scale $n^{-1/2}$ is separate.

\section{Proof of the bounded spectral-basis theorem}\label{sec:general-proof}
This section verifies the primitive conditions used in the three-step proof outline of the main paper. The displacement bound comes from Section S4, the exact fitting expansion and its remainder from Sections S2--S3, and the selection argument from Section S3. The calculations below establish their spectral rates and the required matrix nondegeneracy.
All conditions here are those of the main theorem, with $p,d_0$ fixed. The complete real basis is normalized by $N^{-1}\bm\varphi_h^T\bm\varphi_{h'}=\mathbf1\{h=h'\}$, has uniformly bounded entries and at most $C(1+r)^{d_0-1}$ modes of degree $r$. Its subsets inherit $\mathbf B^T\mathbf B/N=\mathbf I$. The general field assumptions are (G), (ND) and either bounded residuals or (T) in the main paper; Gaussian coefficient models are verified separately below. All constants below are uniform over the fixed finite candidate library.

\paragraph{Complexity.}
Since $\lambda_-{\mathbf{I}}\preceq {\mathbf{A}}\preceq\lambda_+{\mathbf{I}}$, ${\mathbf{M}}^{-1}\preceq(\lambda_-{\mathbf{I}}+{\bm{\Lambda}})^{-1}$. There are at most $C(1+r)^{d_0-1}$ basis modes on shell $r$, each uniformly bounded. For the specified penalties,
\[
 d_{\rm eff},\ell_{\max}
 \le C\left[k^{d_0}+\sum_{r>k}
 \frac{r^{d_0-1}}{1+\rho_-(r/k)^{2b_-}}\right]
 \le Ck^{d_0},
\]
since $2b_->d_0$. A wholly unpenalized subset of maximum degree $O(k)$ has the same bound. This permits a raw dictionary as large as $N$ when its high-frequency part is sufficiently penalized.

\paragraph{Target approximation.}
Let ${\bm{h}}_{j,k}$ truncate the basis expansion of ${\bm{g}}_j={\bm{\ell}}_{\cdot j}/{\bm{\lambda}}$ at $k$. The uniform Sobolev-$s_*$ bound implies
\[
 \norm{{\bm{g}}_j-{\bm{h}}_{j,k}}_N^2\le Ck^{-2s_*},\qquad
 \norm{{\bm{g}}_j-{\bm{h}}_{j,k}}_\infty\le Ck^{d_0/2-s_*}.
\]
The second follows from Cauchy--Schwarz over coefficients and the convergent shell sum $\sum_{r>k}r^{d_0-1-2s_*}$. Weighted norms are equivalent, so $\epsilon_{H,j}=O(k^{-s_*})$. Section S4 gives uniform boundedness of both uncorrected and corrected target weights, because $s_*>d_0/2$. It also implies uniform boundedness of ${\bm{\ell}}_{\cdot j}={\bm{\lambda}} {\bm{g}}_j$. No smoothness of $1/{\bm{\lambda}}$ alone is substituted for these target conditions.

\paragraph{Approximate mean and protected components.}
Truncate the complete basis expansion at degree $k$. For the discrete coefficients $\mu_h=N^{-1}{\bm{\varphi}}_h^T{\bm{\mu}}$, a sufficient condition is
\[
 \sup_N\sum_h(1+|h|^2)^t|\mu_h|^2<\infty,
 \qquad t>d_0/2.
\]
Parseval gives a squared truncation error at most $Ck^{-2t}$. Bounded basis entries and Cauchy--Schwarz give a uniform error at most
$C\{\sum_{r>k}r^{d_0-1-2t}\}^{1/2}\le Ck^{d_0/2-t}$.
The finite sum is bounded by the infinite shell sum, uniformly in $N$. These are discrete coefficients; no claim that sampling an arbitrary continuous function preserves the bound without controlling aliasing is needed. The direct approximation assumption also permits arbitrarily large protected components and is weaker than this sufficient Sobolev condition.

Put ${\bm{u}}_k={\bm{\mu}}-{\bm{h}}_k$. For any frame vector ${\bm{z}}$, the normal-equation energy bound gives
$\|\mathcal S_\nu{\bm{z}}\|_\infty\le\sqrt{\ell_{\max,\nu}}\|{\bm{z}}\|_\lambda$.
Indeed, with ${\bm{a}}={\mathbf M}^{-1}{\mathbf B}^T{\mathbf D}_\lambda{\bm{z}}/N$, one has ${\bm{a}}^T{\mathbf M}{\bm{a}}\le\|{\bm{z}}\|_\lambda^2$, and Cauchy--Schwarz bounds each ${\bm{b}}_i^T{\bm{a}}$. Since $\mathcal S_\nu$ fixes $H_k$ and is a weighted contraction,
\begin{align*}
 \|(I-\mathcal S_\nu){\bm{\mu}}\|_\infty
 &\le\|{\bm{u}}_k\|_\infty+\sqrt{\ell_{\max,\nu}}\|{\bm{u}}_k\|_\lambda,\\
 \|\mathcal S_\nu(I-\mathcal S_\nu){\bm{\mu}}\|_\infty
 &\le2\sqrt{\ell_{\max,\nu}}\|{\bm{u}}_k\|_\lambda.
\end{align*}
Both are $O(k^{d_0/2-t})=o(1)$ uniformly over candidates. Large components in $H_k$ cancel also in the fitted residuals: if ${\bm{h}}_k={\mathbf B}{\bm{\gamma}}$ with ${\bm\Lambda}{\bm{\gamma}}=0$, adding ${\bm{h}}_k$ to the responses adds ${\bm{\gamma}}$ to both population and sample regression coefficients and leaves $e,r,\widehat e,\widehat r$ unchanged on nonsingular samples. The corrected coefficient estimate and its target shift by the same amount.

\paragraph{Influence agreement and localization.}
Section~\ref{sec:common-influence} and the deterministic bounds above imply
\[
 \max_\nu\E_\xi\|\bm\Psi_\nu-\bm\Psi_0\|_{\lambda,F}^2
 \le C\{k^{-2s_*}+k^{-2t}+k^{d_0-a}+k^{d_0}/N\}\to0.
\]
Since $\ell_i$ is bounded, (G) implies
\[
 \E_\xi\tr(n\mathbf V_0)
 =N^{-1}\sum_i(1-\pi_i)\|\bm\ell_i\|^2\E_\xi\zeta_i^2/\lambda_i\le C.
\]
Thus Markov provides upper eigenvalue localization, while (ND) supplies the lower one. Write $a_\nu=\|\bm\Psi_\nu-\bm\Psi_0\|_{\lambda,F}^2=o_{P_\xi}(1)$ and $b_\nu=\|\bm\Psi_\nu\|_{\lambda,F}^2=O_{P_\xi}(1)$, also for $b_0$. Cauchy--Schwarz gives
\[
 \max_\nu n\|\mathbf V_\nu-\mathbf V_0\|_\op
 \le\max_\nu\sqrt{a_\nu}(\sqrt{b_\nu}+\sqrt{b_0})=o_{P_\xi}(1).
\]
For any $\eta>0$, allocate failure probability $\eta/3$ each to (ND) and the Markov upper bound. Weyl's inequality and finite-library convergence transfer their intersection to all $n\mathbf V_\nu$, with lower bound $c_{\eta/3}/2$, a finite upper bound, and limiting failure probability at most $\eta$. Although $\mathbf V_0$ is auxiliary, each $\mathbf V_\nu$ is $\mathcal U_N$ measurable. Its eigenvalue event has the same probability under the original population law and the extension. This establishes the scale conditions before applying S3; no deterministic limit of $n\mathbf V_0$ is needed.

\paragraph{Gap and conclusion.}
Orthogonality gives $|\E_\xi G_{\nu j}|\le Ck^{-(s_*+t)}$. Apply (G) to $N^{-1}\bm d_{\nu,\cdot j}^T\bm v$ and $N^{-1}\bm d_{\nu,\cdot j}^T\bm\zeta$ separately, and use $\E(A+B)^2\le2\E A^2+2\E B^2$, obtaining
\[
 \max_\nu\E_\xi\|\bm G_\nu\|^2\le Ck^{-2s_*}(k^{-2t}+k^{-a}+N^{-1}).
\]
All three $n$-scaled terms vanish under the main-paper rates. In S3 apply Lemma~\ref{lem:bounded-sampling} for bounded $\Delta$, or the preceding localized-moment proof under (T). Both yield $\bm{\mathcal R}_\nu=O_p(k^{d_0}/n)=o_p(n^{-1/2})$, covariance consistency and conditional normality. Here joint $O_p$ means tightness under $P_{p,\xi}$; its constants are obtained by first fixing a high-probability population event. Joint $o_p$ errors transfer to conditional errors by S3.2.

For a deterministic reference $\nu_0$, $\|\bm\Psi_\nu-\bm\Psi_{\nu_0}\|_{\lambda,F}^2\le2a_\nu+2a_{\nu_0}=o_{P_\xi}(1)$. The left side is $\mathcal U_N$ measurable, so this statement descends from the extension. All requirements of S3.2 now hold uniformly over the finite library. They give same-label selection validity and, for every fixed full-row-rank contrast, conditional Wald coverage. No latent quantity is required for the feasible estimator.

\begin{lemma}[Pointwise Gaussian residual moments]\label{lem:gaussian-moments}
Retain the deterministic basis, dictionary and mean assumptions, with $k\to\infty$ and $k^{d_0}/N\to0$. Suppose $(\bm v,\bm\zeta)$ is jointly centered Gaussian, the operator bounds in (G) hold, and additionally
\[
 \max_i\Var_\xi\{(\mathbf Q_k\bm v)_i\}\le Ck^{d_0-a}.
\]
Then $\sup_{\nu,i}\E_\xi(e_{\nu i}^4+r_{\nu i}^4)\le C$ uniformly for all sufficiently large $N$. No gap-recovery rate is needed for this moment bound. Thus (T) holds with $\delta=2$. The additional pointwise tail bound is verified separately in the models below; it is not inferred from a trace bound alone.
\end{lemma}
\begin{proof}
Suppress $\nu$. Since $0\preceq\widetilde{\mathcal S}\preceq I$, its powers satisfy
$\widetilde{\mathcal S}^2,\widetilde{\mathcal S}^4\preceq\widetilde{\mathcal S}$.
Its diagonal is $\widetilde{\mathcal S}_{ii}=\lambda_i\bm b_i^T\mathbf M^{-1}\bm b_i/N$; equivalence of the weighted Euclidean norms yields
\[
 \max_i\sum_j(\mathcal S_{ij})^2+
 \max_i\sum_j(\mathcal S^2_{ij})^2\le C\ell_{\max}/N.
\]
Right multiplication by the Euclidean projector $\mathbf Q_k$ cannot increase these row norms. By (G), for $j=1,2$,
\[
 \max_i\Var\{(\mathcal S^j\mathbf Q_k\bm v)_i\}
 \le C\ell_{\max}k^{-a}\le Ck^{d_0-a}.
\]
Since $(I-\mathcal S)\bm v=(I-\mathcal S)\mathbf Q_k\bm v$ and
$\mathcal S(I-\mathcal S)\bm v=(\mathcal S-\mathcal S^2)\mathbf Q_k\bm v$, the structural parts of $e,r$ have bounded pointwise variance. Condition (G) and the same row bounds control the nugget parts, since $\ell_{\max}/N=O(k^{d_0}/N)=o(1)$. Squared-triangle bounds suffice even if $v,\zeta$ are correlated.
Writing $u_k=\mu-h_k$, the mean contributions have maximum norms bounded by
\[
 \norm{u_k}_\infty+\sqrt{\ell_{\max}}\norm{u_k}_\lambda,
 \qquad 2\sqrt{\ell_{\max}}\norm{u_k}_\lambda,
\]
respectively; both are $O(k^{d_0/2-t})=o(1)$. For Gaussian $W$ with mean $m_W$ and variance $s_W^2$,
$\E W^4=m_W^4+6m_W^2s_W^2+3s_W^4$. This proves the pointwise claim. Summing over the at most $J_0$ candidates and using Markov gives a constant $C_\eta=CJ_0/\eta$ for (T).
\end{proof}

\begin{lemma}[Gaussian spectral special case]\label{lem:containment}
Retain the deterministic basis, dictionary, target, mean and acquisition assumptions of the main paper. Let
$\bm v=\sum_h c_h\bm\varphi_h$, with independent centered Gaussian coefficients of variances $\tau_h^2\le C(1+|h|)^{-a}$, $a>d_0$, and let $\bm\zeta=\sigma\bm\varepsilon$, $\sigma>0$, with iid standard Gaussian $\varepsilon_i$ independent of the coefficients. Then all field, moment and (ND) assumptions of Theorem 1 of the main paper hold. With the same target, mean, sampling and rate assumptions, the recovery conclusion applies, including the Fourier and DCT corollaries.
\end{lemma}
\begin{proof}
Discrete orthogonality gives covariance eigenvalues $N\tau_h^2$, so
\begin{align*}
 \norm{\mathbf Q_k\bm\Sigma_v\mathbf Q_k}_\op
 &=N\max_{|h|>k}\tau_h^2\le CNk^{-a},\\
 N^{-1}\tr(\mathbf Q_k\bm\Sigma_v\mathbf Q_k)
 &=\sum_{|h|>k}\tau_h^2
 \le C\sum_{r>k}r^{d_0-1-a}\le Ck^{d_0-a}.
\end{align*}
An empty tail contributes zero. Bounded basis entries also give
$\max_i\Var\{(\mathbf Q_kv)_i\}\le C\sum_{|h|>k}\tau_h^2\le Ck^{d_0-a}$.
Since $\bm\Sigma_\zeta=\sigma^2I$, Lemma~\ref{lem:gaussian-moments} supplies the fourth-moment condition, including nonzero approximately protected means.

For (ND), put $\mathbf V_{0,N}^{\mathrm{det}}=\sigma^2N^{-1}\sum_i\bm\ell_i\bm\ell_i^T/\lambda_i$.
Independence of the Gaussian squares, $\Var(\varepsilon_i^2)=2$, bounded loadings and fixed $p$ imply
\[
 \norm{n\mathbf V_0-\mathbf V_{0,N}^{\mathrm{det}}}_\op
 =O_{P_\xi}(N^{-1/2}+n/N),\qquad
 \frac{\sigma^2}{\lambda_+}\mathbf J^{-1}
 \preceq\mathbf V_{0,N}^{\mathrm{det}}\preceq
 \frac{\sigma^2}{\lambda_-}\mathbf J^{-1}.
\]
The identity $N^{-1}\mathbf L^T\mathbf L=\mathbf J^{-1}$ supplies the last comparison. Thus its smallest eigenvalue has a fixed positive lower bound; Weyl's inequality proves (ND). This uses actual independent nugget squares, not independence of candidate fits.

No periodic geometry entered this calculation. The Fourier and DCT-II systems have the bounded entries and shell counts, and their target and mean approximation bounds are verified in Sections S5.3--S5.4. Proposition 2 of the main paper uses a construction in this Gaussian class; it does not provide a lower bound for every model covered by the theorem. The moment equivalence in Proposition 1 is asserted only under its Gaussian assumption.
\end{proof}

\subsection{Roles of the assumptions}
S2 is algebraic. The operator bound in (G) controls the gap, whereas its trace bound controls common-influence convergence; they are not interchangeable at the same rate. Weighted self-adjointness proves Euclidean defect orthogonality by the inner-product calculation in S4. The conditional sampling argument uses either bounded $\Delta$ or the explicit population-measurable localized moments (T), together with eigenvalue bounds transferred from (ND). Neither (T) nor (ND) follows from covariance upper bounds alone.

Gaussianity and independent nugget noise are sufficient primitive conditions for the special cases in Lemma~\ref{lem:containment}, not general assumptions of Theorem 1. The Gaussian assumption is retained for the mean/covariance equivalence in Proposition 1, as proved in S3.1; this equivalence is not asserted for other fields. The following Gaussian construction retains its original scope. Uniform localization across a fixed finite library permits S3.2 without independence between selector and candidate fits.

\subsection{Proof of the slope-only boundary construction}
Set $\bm{x}=(\sqrt2\cos(2\pi s_{i1}))_{i=1}^N$, $\bm{z}=(\sqrt2\cos(2\pi L_ks_{i1}))_{i=1}^N$ and ${\bm{h}}_c,{\bm{h}}_s$ to the normalized real pair at $L_k+1$. The no-aliasing condition gives ${\mathbf{J}}={\mathbf{I}}_2$ and
\[
 N^{-1}{\mathbf{B}}^T\diag(\bm{z}){\mathbf{B}}=0,
\]
since every product of dictionary functions has first-coordinate frequency at most $2C_0k<L_k$. Thus ${\mathbf{A}}={\mathbf{I}}_K$. The two columns of ${\mathbf{C}}_B$ select the protected constant and frequency-one cosine; normalized penalties are zero there. Therefore ${\bm{q}}_i={\bm{q}}_{+,i}=(1,x_i)^T$ for every candidate, regardless of the penalties on other coordinates. The two defect columns are $\delta_k\bm{z}$ and $\delta_k\bm{z}{\bm{x}}$.

The product identity
\[
 \bm{z}{\bm{x}}=\frac{{\bm{h}}_{L_k+1,c}+{\bm{h}}_{L_k-1,c}}{\sqrt2}
\]
implies $\langle\bm{z},{\bm{h}}_c\rangle_N=\langle\bm{z},{\bm{h}}_s\rangle_N=0$, $\langle \bm{z}{\bm{x}},{\bm{h}}_c\rangle_N=1/\sqrt2$ and $\langle \bm{z}{\bm{x}},{\bm{h}}_s\rangle_N=0$. Also $\norm{\bm{z}}_N^2=\norm{\bm{z}{\bm{x}}}_N^2=1$ and $\langle\bm{z},\bm{z}{\bm{x}}\rangle_N=0$. Hence the structured gap is exactly $(0,\delta_k Z_c/\sqrt2)^T$, while the independent nugget contributes covariance $(\sigma^2\delta_k^2/N){\mathbf{I}}_2$. This proves the displayed covariance in the main proposition. Positivity holds because $\|\delta_k\bm{z}\|_\infty\le\sqrt2c<1$, and ${\bm{\lambda}}$ has frame mean one.

For completeness the target functions also satisfy the asserted uniform regularity. Expand $(\bm{1}+\delta_k\bm{z})^{-1}=\sum_{r\ge0}(-\delta_k\bm{z})^r$. The absolute Fourier coefficient sum of $\bm{z}^r$ is at most $(\sqrt2)^r$ and its largest frequency is $rL_k$. Sampling on the periodic grid can alias these frequencies but cannot increase their canonical magnitude beyond $rL_k$. The weighted coefficient $\ell_1$ norm, which bounds the Sobolev $\ell_2$ norm, is therefore at most
\[
 C+ C\sum_{r\ge1}(1+rL_k)^{s_*}(\sqrt2|\delta_k|)^r
 \le C\{1+L_k^{s_*}|\delta_k|\}\le C.
\]
Multiplying by the fixed frequency-one function ${\bm{x}}$ changes the bound by a constant, so it holds for both ${\bm{\ell}}_{\cdot j}/{\bm{\lambda}}$. The structural field has only one nonzero Fourier pair with variance of order $(L_k+1)^{-a}$ and is a member of the stated upper-bound Gaussian class; equal pair variances make it stationary. Its mean is zero, so the mean approximation condition holds with ${\bm{h}}_k=0$ for any $t>d_0/2$; choose $t\ge a/2$ to avoid an additional rate restriction. The sequence is intentionally chosen within the class and does not claim a typical realization of every smooth acquisition law.

Let $k\to\infty$, $k^{2d_0}/n\to0$ and $n/N\to0$. The centered sampling approximation and covariance consistency of Section S3 apply to this construction independently of gap recovery. Indeed, the target weights are bounded, the residual and correction fourth moments satisfy the localization bounds of S5, and $d_{\rm eff}\ell_{\max}/n=O(k^{2d_0}/n)\to0$. Independent Bernoulli acquisition gives $\pi_{\max}=O(n/N)\to0$. The common-influence argument and the matrix nondegeneracy calculation in S5 bound the eigenvalues of $n\mathbf V$ above and away from zero on events of arbitrarily high field probability. None of these steps requires $nk^{-(a+2s_*)}\to0$; the recovery conclusion of Theorem 1 is not invoked.

Write $G_{\rm slope}$ for the second coordinate of $\bm G$. Its structural and independent nugget components have variances $\delta_k^2\tau_k/2$ and $\sigma^2\delta_k^2/N$, respectively. Thus
\[
 n\operatorname{Var}_\xi(G_{\rm slope})
 =n\delta_k^2(\tau_k/2+\sigma^2/N),\qquad
 n\delta_k^2\tau_k\asymp nk^{-(a+2s_*)},\qquad
 n\delta_k^2/N\to0.
\]
For completeness, the uncorrected estimator $\widehat{\bm\beta}_0$ has population gap ${\bm G}_{\rm u}=N^{-1}\sum_i(\lambda_i{\bm q}_i-{\bm\ell}_i)\Delta_i$ and influence ${\bm\phi}_{{\rm u},i}={\bm q}_ie_i$. The subscript ${\rm u}$ distinguishes this influence from the common nugget vector ${\bm\phi}_{0i}$ in S5. On nonsingular samples, the S2 identities give
\[
 \widehat{\bm\beta}_0-{\bm\beta}_U
 ={\bm G}_{\rm u}+n^{-1}\sum_i(R_i-\pi_i){\bm\phi}_{{\rm u},i}
 -{\mathbf G}_q^T\widehat{\mathbf M}^{-1}{\bm g}_e.
\]
Indeed, ${\mathbf C}_B^T(\widehat{\bm\alpha}-{\bm\alpha})
={\mathbf U}^T{\bm g}_e-{\mathbf U}^T{\mathbf E}\widehat{\mathbf M}^{-1}{\bm g}_e$.
Section S3 bounds the last product by $O_p(\sqrt{d_{\rm eff}\ell_{\max}}/n)=O_p(k^{d_0}/n)$.
In this construction specifically, ${\bm q}_i={\bm q}_{+,i}={\bm\ell}_i$, so ${\bm G}_{\rm u}={\bm G}$ and ${\bm\phi}_{{\rm u},i}-{\bm\phi}_i={\bm q}_ir_i$. Bounded weights and S4.3 give
\[
 \E_\xi\left\{N^{-1}\sum_i\lambda_i
 \|{\bm\phi}_{{\rm u},i}-{\bm\phi}_i\|^2\right\}
 \le C\E_\xi\|{\bm r}\|_\lambda^2
 \le C\{k^{d_0-a}+k^{d_0}/N\}\to0.
\]
Thus both influences approach the same nugget vector, and
${\mathbf V}_{\rm u}=n^{-2}\sum_i\pi_i(1-\pi_i){\bm\phi}_{{\rm u},i}{\bm\phi}_{{\rm u},i}^T$
has the same nondegenerate first-order scale. The bounded weights and localized fourth moments also give the centered conditional CLT for $\widehat{\bm\beta}_0$ by S3's Lyapunov argument.

For its feasible variance, use the uncorrected inclusion derivative, holding ${\bm\Gamma}$ fixed:
\[
 \widehat{\bm\psi}_{{\rm u},i}
 ={\mathbf C}_B^T{\mathbf T}^{-1}{\bm b}_i\widehat e_i
 =n^{-1}\widehat{\bm q}_i\widehat e_i,
 \qquad
 \widehat{\mathbf V}_{\rm u}
 =\sum_{i\in S}\widehat{\bm\psi}_{{\rm u},i}
                  \widehat{\bm\psi}_{{\rm u},i}^T.
\]
This is the covariance for the uncorrected fit, rather than the corrected derivative in the main paper. Set $\widehat{\bm\phi}_{{\rm u},i}=\widehat{\bm q}_i\widehat e_i$. Expanding its error as
$(\widehat{\bm q}_i-{\bm q}_i)e_i+\widehat{\bm q}_i(\widehat e_i-e_i)$,
the S3 weight and fitted-residual bounds yield
$n^{-1}\sum_iR_i\|\widehat{\bm\phi}_{{\rm u},i}-{\bm\phi}_{{\rm u},i}\|^2=o_p(1)$.
The oracle fourth-moment variance bound, the vanishing finite-population correction and empirical Cauchy--Schwarz then give
$n\|\widehat{\mathbf V}_{\rm u}-{\mathbf V}_{\rm u}\|_{\rm op}=o_p(1)$.
These arguments do not assume gap recovery, and the equality of the two population gaps is asserted only for this construction.

For each fixed candidate, using the covariance corresponding to each estimator, Section S3.1 applied to the slope contrast gives nominal conditional coverage if and only if $\sqrt nG_{\rm slope}\to0$ in $P_\xi$ probability. This centered Gaussian gap converges to zero at that scale if and only if $nk^{-(a+2s_*)}\to0$. If this rate condition fails, its conditional slope coverage does not converge to the nominal level in $P_\xi$ probability. The intercept gap has only the vanishing nugget term, so intercept coverage recovers under the stated sampling conditions. Candidates share the same population gap here; changing candidates cannot remove its structural component. The construction establishes the stochastic directional-gap order in the zero-mean subclass, not sharpness of the general mean-approximation term, optimality over all methods, or necessity of the sampling complexity constraint.

A deterministic Fourier check evaluates this construction for $d_0=1,2$ and $k=1,2,3$ on side-65 grids, with $L_k=4k+3$, dictionary radius $2k$, $\delta_k=0.1k^{-(d_0/2+0.75)}$, and coefficient variance $(L_k+1)^{-(d_0+0.75)}$. In the cross-platform replay, the maximum covariance discrepancy is $3.44\times10^{-20}$ and the maximum loading discrepancy is $4.87\times10^{-14}$, both below the fixed tolerance $10^{-11}$. Across these six cases, write ${\bm{\Sigma}}_G$ for the theoretical covariance and $\widehat{\bm{\Sigma}}_G$ for its deterministic matrix evaluation. The maximum relative Frobenius covariance error, $\|\widehat{\bm{\Sigma}}_G-{\bm{\Sigma}}_G\|_F/\|{\bm{\Sigma}}_G\|_F$, is at most $4.2\times10^{-14}$; this descriptive scale comparison does not change the original absolute tolerance. These finite matrix checks do not estimate coverage or prove the asymptotic result.

\subsection{Verification for the periodic Fourier frame}
On the odd grid $N=m^{d_0}$, canonical frequencies lie in $\{-(m-1)/2,\ldots,(m-1)/2\}^{d_0}$. Use the constant and normalized real sine/cosine columns for one representative of each nonzero opposite-frequency pair. Finite geometric sums give orthogonality under $N^{-1}$ summation and completeness; entries are bounded by $\sqrt2$. With degree $|h|=\|\bm h\|_\infty$, the shell at radius $r\ge1$ has $(2r+1)^{d_0}-(2r-1)^{d_0}=O(r^{d_0-1})$ modes. Also $\|\bm h\|_\infty\le\|\bm h\|_2\le\sqrt{d_0}\|\bm h\|_\infty$, so the original Sobolev and variance-decay conditions are equivalent up to fixed constants to those of the general theorem. Equal sine/cosine variances imply stationarity, but are not required by the upper bounds.

Full-cube dictionaries through any larger radius, including the whole frame, satisfy the subset rule. The former penalties are exactly the degree penalties of the theorem. Fourier truncation and the shell argument above verify the mean approximation whenever the discrete mean has a uniform Sobolev-$t$ coefficient bound. More generally its two direct approximation conditions suffice. The rates give $k^{2d_0}=o(n)=o(N)$ and hence $k/m\to0$, so all protected frequencies are available. The general proof now gives Corollary 1. In particular $d_0=2$, $a=3$, $s_*=5/4$, $t\ge3/2$ and $k\asymp n^{1/5}$ recover the periodic comparison in S4.1. The trigonometric product construction of S5.2 remains specific to this Fourier class.

\subsection{Verification for the nonperiodic cosine frame}
We prove Corollary 2 on the exact finite midpoint grid, rather than deducing discrete regularity from a continuous assertion without controlling aliasing.

\paragraph{Finite orthogonality and coordinate coefficients.}
Put $x_i=(i+1/2)/m$ and define $a_h$ as in the main paper. For integer $1\le l<2m$,
$\sum_{i=0}^{m-1}\cos(\pi l x_i)=0$; this follows by taking the real part of the geometric sum. The product-to-sum identity therefore gives $m^{-1}\sum_i a_h(x_i)a_{h'}(x_i)=\mathbf1\{h=h'\}$ for $0\le h,h'<m$. Tensor products are complete, have absolute value at most 2, and have $2r+1=O(1+r)$ modes on shell $r<m$.

For $u_i=x_i-1/2$, denote its discrete coefficient by $\widehat u_h=m^{-1}\sum_i u_i a_h(x_i)$. To evaluate it, differentiate the finite identity
\[
 \sum_{i=0}^{m-1}\sin\{(i+1/2)\theta\}
 =\frac{\sin^2(m\theta/2)}{\sin(\theta/2)}.
\]
At $\theta=\pi h/m$, $1\le h<m$, the derivative of the numerator vanishes. The centered constant contribution also vanishes by the preceding cosine sum, giving
\[
 \widehat u_h=-\frac{\sqrt2}{2m^2}
 \frac{\sin^2(\pi h/2)\cos(\pi h/(2m))}{\sin^2(\pi h/(2m))}.
\]
Thus even coefficients vanish; odd coefficients have the displayed nonzero value. Since $\sin(\pi h/(2m))\ge h/m$ on this range, $|\widehat u_h|\le Ch^{-2}$ uniformly in $m$. The zero coefficient is zero by symmetry.

\paragraph{A uniform discrete coefficient bound.}
For $f\in C^2([0,1])$, let $\gamma_l=\int_0^1f(x)a_l(x)\,dx$. Twice integrating by parts gives, for $l\ge1$,
\[
 \gamma_l=\frac{\sqrt2}{(\pi l)^2}
 \left\{(-1)^lf'(1)-f'(0)-\int_0^1f''(x)\cos(\pi lx)\,dx\right\}.
\]
Consequently $|\gamma_l|\le C_f l^{-2}$, with $C_f$ controlled by the endpoint derivatives and $\|f''\|_{L^1}$. The cosine series is absolutely and uniformly convergent and equals $f$: its sum has the same complete cosine coefficients, hence agrees in $L^2$, and both are continuous. It may therefore be substituted into each finite midpoint sum.

Write $\widehat f_h^{(m)}=m^{-1}\sum_i f(x_i)a_h(x_i)$. For $1\le h<m$, the discrete inner product of $a_l$ and $a_h$ is zero unless $l\equiv h$ or $-h\pmod{2m}$. In the latter cases its absolute value is at most one; the signs account for midpoint aliasing. Hence
\[
 |\widehat f_h^{(m)}|
 \le C_f\sum_{q\in\mathbb Z}(1+|2qm+h|)^{-2}
 \le C_f'(1+h)^{-2}.
\]
Indeed, the $q=0$ term has this order, while for $q\ne0$, $|2qm+h|\ge(2|q|-1)m$; their sum is $O(m^{-2})=O(h^{-2})$. The integers in this sum represent both positive-frequency alias classes by absolute value. The zero coefficient is bounded by $\|f\|_\infty$. No zero endpoint-derivative condition is needed. The bound is uniform over families with uniformly bounded $C^2$ norms.

\paragraph{All target weights and the mean.}
The exact midpoint moments give
\[
 \mathbf J=\diag(1,V_m,V_m),\qquad
 V_m=\frac{m^2-1}{12m^2},\qquad
 \frac{2}{27}\le V_m<\frac1{12}\quad(m\ge3).
\]
The cosine sum gives $N^{-1}\sum_{i,j}\lambda(x_i,y_j)=1$, and $1-\eta\le\lambda\le1+\eta$. For fixed $0<\eta<1$, the functions
$w_0(x)=(1+\eta\cos(\pi x))^{-1}$ and $w_1(x)=(x-1/2)w_0(x)$ have bounded $C^2$ norms. The target weights are $w_0(x)$, $w_1(x)/V_m$, and $w_0(x)(y-1/2)/V_m$. Their tensor coefficients factor exactly into one-dimensional discrete coefficients, so all obey
\[
 |\widehat g_j(u,v)|\le C(1+u)^{-2}(1+v)^{-2},\qquad 0\le u,v<m.
\]
For $s<3/2$, the inequality
$(1+\max(u,v)^2)^s\le C_s\{(1+u)^{2s}+(1+v)^{2s}\}$ yields
\[
 \sum_{u,v<m}(1+\max(u,v)^2)^s|\widehat g_j(u,v)|^2
 \le C_s\left\{\sum_{u\ge0}(1+u)^{2s-4}\right\}
          \left\{\sum_{v\ge0}(1+v)^{-4}\right\}<\infty.
\]
The constants are independent of $m$. Any fixed $1<s_*<3/2$ therefore satisfies the two-dimensional target condition. The same argument, or the exact coordinate coefficients above, gives a uniform spectral Sobolev-$t$ bound for the stated affine mean with bounded coefficients and $1<t<3/2$. The general truncation argument then gives both required mean errors. These claims concern the coordinate targets without modifying the scientific design.

Finally, the stipulated independent Gaussian DCT coefficients give a covariance diagonal in this basis with eigenvalues $N\operatorname{Var}(c_{uv})$; $a>2$ yields the required tail bounds. The positive independent nugget and the displayed Gram and intensity bounds verify nondegeneracy. Dictionary complexity, conditional sampling limits and finite-library selection follow from the general proof. The rate conditions imply $k/m\to0$, so protected degrees are available. For $a=3,s_*=5/4,t=7/5,k\asymp n^{1/5}$, the complexity, stochastic-gap and mean-gap powers are respectively $-1/5,-1/10,-3/50$. With $n/N\to0$, this proves Corollary 2. It does not assert diagonalization of a general nonperiodic covariance by the DCT or validity on irregular or masked frames.

\subsection{Finite periodic spectral calculus}
For these examples set $d_0=d$ and use the main paper's dictionary and penalty rules. The Mat\'ern smoothness is $\upsilon$, distinct from candidate index $\nu$; all model parameters are fixed across grid sizes. Write $\mathbb T^d=[0,1)^d$ with addition modulo $1$ and let $m$ be odd, $N=m^d$, and $\mathcal H_m=\{-(m-1)/2,\ldots,(m-1)/2\}^d$. Work temporarily with the complex Fourier characters $e_h(s)=e^{2\pi\mathrm i h\cdot s}$ and the discrete matrix $\mathbf F_m$ having $\mathbf F_m^*\mathbf F_m=NI$. The real sine--cosine system is obtained by an unitary change from conjugate complex characters to real sine--cosine columns within each pair and has the same high/low projectors.

\begin{lemma}[Uniform aliasing and tails]\label{lem:alias}
If a nonnegative torus spectral sequence obeys $w_h\le C(1+\|h\|_\infty)^{-a}$ for $a>d$, the sampled stationary field has discrete spectral weights
\[w_h^{(m)}=\sum_{q\in\mathbb Z^d}w_{h+mq}\qquad(h\in\mathcal H_m),\]
and a constant $C'$ independent of odd $m$ with $w_h^{(m)}\le C'(1+\|h\|_\infty)^{-a}$. Consequently
\[
\|\mathbf Q_k\mathbf F_m\diag(w^{(m)})\mathbf F_m^*\mathbf Q_k\|_\op\le CNk^{-a},\quad
\frac1N\tr\{\mathbf Q_k\mathbf F_m\diag(w^{(m)})\mathbf F_m^*\mathbf Q_k\}\le Ck^{d-a}.
\]
If $w_h\asymp(1+\|h\|)^{-a}$ on $\mathbb Z^d$, then $w_h^{(m)}\asymp(1+\|h\|)^{-a}$ uniformly in $m$.
\end{lemma}
\begin{proof}
All spectral sequences are absolutely summable. The continuous covariance series is therefore uniformly convergent; grouping its terms by their residue modulo $m$ on the grid gives the stated alias formula, with no missing factor of $N$. The eigenvalues of the frame covariance are $Nw_h^{(m)}$.
For $q\ne0$, $\|h+mq\|_\infty\ge m(\|q\|_\infty-1/2)$, so
\[
\sum_{q\ne0}w_{h+mq}\le Cm^{-a}\sum_{q\ne0}(\|q\|_\infty-1/2)^{-a}\le C'm^{-a}.
\]
The sum converges because $a>d$. As $\|h\|_\infty\le m/2$, this is at most $C''(1+\|h\|_\infty)^{-a}$.  The zero-alias term gives the remaining upper bound and the lower bound when assumed. The Fourier diagonalization supplies the operator bound; the number of degree-$r$ indices is $O(r^{d-1})$, giving the trace bound.
\end{proof}

\subsection{Model A: a nonstationary, non-diagonal Gaussian field}
\begin{corollary}[Amplitude-modulated periodic Mat\'ern-type field]\label{cor:amp}
Let $z$ be a real-valued periodic stationary mean-zero Gaussian field with spectral weights $w_h\asymp (\kappa^2+4\pi^2\|h\|_2^2)^{-(\upsilon+d/2)}$, normalized to sum to a fixed positive constant, where $\upsilon>0$ and $\kappa>0$. Let $v(s)=b(s)z(s)$ with $b(s)=1+\gamma\cos(2\pi s_1)$, $0<|\gamma|<1$, and add $\zeta_i=\sigma\varepsilon_i$, $\sigma>0$, where the $\varepsilon_i$ are iid standard Gaussian and independent of $z$. With $\mu=0$, bounded smooth scientific regressors and acquisition intensities chosen as in the proof, and the main-paper dictionary rules, the localized-moment route of the main theorem applies under the main-paper rate conditions with $a=2\upsilon+d$.
\end{corollary}
\begin{proof}
At finite grid size, write $\mathbf W_m=\diag(w_h^{(m)})$. Multiplication by $b$ gives the coefficient map $\mathbf T_m=\mathbf I+(\gamma/2)(S_++S_-)$ where $S_\pm$ are the circular shifts $h\mapsto h\pm e_1$, so coefficient covariance is $\bm\Omega_m=\mathbf T_m\mathbf W_m\mathbf T_m^*$. Write $R_m(h)=\|h\|_\infty$ for the canonical residue representative. A circular shift by one changes $R_m(h)$ by at most one: at a wrapped endpoint $(m-1)/2$ becomes $-(m-1)/2$, preserving its absolute value. Thus
$1/2\le(1+R_m(h\pm e_1))/(1+R_m(h))\le2$.
The two-sided comparison in Lemma~\ref{lem:alias} then bounds both ratios of adjacent spectral weights by a constant independent of $m$. Therefore each weighted shift obeys $\|\mathbf W_m^{-1/2}S_\pm \mathbf W_m^{1/2}\|_\op\le C$, and
\[
\|\mathbf W_m^{-1/2}\mathbf T_m\mathbf W_m^{1/2}\|_\op\le C_\gamma,
\qquad \bm\Omega_m\preceq C_\gamma^2\mathbf W_m.
\]
The Loewner inequality and Lemma~\ref{lem:alias} yield the structural conditions (G); the independent nugget separately verifies (G). Since $b$ is nonconstant, $\Var\{v(s)\}=b(s)^2\Var\{z(s)\}$ is nonconstant, and $\bm\Omega_m$ is generically non-diagonal.

For explicit target and acquisition choose $\bm x(s)=(1,\sqrt{2}\cos(2\pi s_1))^\top$, $\lambda(s)=1+\rho\cos(2\pi s_1)$ with $0<|\rho|<1$, and $f\equiv0$; then $\mathbf J=\mathbf I_2$ for odd $m\ge5$, $\lambda$ has grid average one, and $\ell_{\cdot j}/\lambda$ is periodic analytic with grid-uniform discrete Sobolev bounds. The mean is zero, hence exactly reproduced by $H_k$. The discrete target Sobolev bound is established in the target-weight paragraph below, which applies to both models.

For the moment route, Loewner domination and bounded Fourier characters give
\[
 \max_i\Var\{(\mathbf Q_k\bm v)_i\}
 \le C\sum_{|h|>k}w_h^{(m)}\le Ck^{d-a}.
\]
The mean is zero and the field plus nugget is jointly Gaussian. Lemma~\ref{lem:gaussian-moments} therefore bounds the pointwise fourth moments of both $e,r$, including the $\mathcal S\mathbf Q_k$ and $\mathcal S^2\mathbf Q_k$ terms. This verifies (T) with $\delta=2$; the response itself is not bounded.

Finally $n\mathbf V_0=\sigma^2N^{-1}\sum_i\{(1-\pi_i)/\lambda_i\}\bm\ell_i\bm\ell_i^\top\varepsilon_i^2$. Independent Gaussian squares and bounded loadings show its entrywise difference from $\sigma^2N^{-1}\sum_i\bm\ell_i\bm\ell_i^\top/\lambda_i$ is $O_{P_\xi}(N^{-1/2}+n/N)$.  The latter matrix lies between positive multiples of $\mathbf J^{-1}$. Hence (ND) holds.  Choose a valid rate sequence as in the final paragraph for Model B (zero mean permits any fixed $t$), and apply the moment-route theorem.
\end{proof}

\paragraph{Grid-uniform target weights for both examples.}
The two target functions are $1/(1+\rho\cos(2\pi s_1))$ and
$\sqrt{2}\cos(2\pi s_1)/(1+\rho\cos(2\pi s_1))$.
Choose a fixed strip width $b_*>0$ such that $|\rho|\cosh(2\pi b_*)<1$.
Their denominators are bounded away from zero in this complex strip; shifting the Fourier integration contour gives coefficient bounds $|a_h|\le C e^{-c|h_1|}$, with coefficients zero off the first coordinate axis.
The sampled coefficients are $a_h^{(m)}=\sum_q a_{h+mq}$.
For canonical $h$, $|h+mq|\ge|h|$ and $|h+mq|\ge m(|q|-1/2)$ for $q\ne0$, so splitting the exponential into two factors gives
$|a_h^{(m)}|\le C'e^{-c|h_1|/2}$ uniformly in odd $m$.
Consequently every fixed-order discrete Sobolev sum is uniformly bounded. Real pair conversion changes at most fixed constants. Discrete orthogonality gives $\mathbf J=\mathbf I_2$ and mean intensity one for all odd $m\ge5$.

\subsection{Model B: a periodic logistic--Mat\'ern binary population}
Let $Z$ be a stationary mean-zero periodic Gaussian field with $\Var\{Z(s)\}=1$, with two-sided spectral weights $w_h\asymp (1+\|h\|)^{-a}$, $a=2\upsilon+d>d$, $\upsilon>0$. Conditional on its whole path, let site responses be independent $Y_i\mid Z\sim\operatorname{Bernoulli}\{p_Y(s_i)\}$, $p_Y(s)=g_{\mathrm{log}}(Z(s))$, $p_{Y,i}:=p_Y(s_i)$, $g_{\mathrm{log}}(z)=1/(1+e^{-\eta_0-z})$.

\begin{lemma}[Hermite spectrum under a logistic link]\label{lem:hermite}
The centered field $v(s)=g_{\mathrm{log}}(Z(s))-E_\xi g_{\mathrm{log}}(Z(s))$ is stationary and its Fourier spectral weights obey $\widetilde w_h\le C(1+\|h\|)^{-a}$.
\end{lemma}
\begin{proof}
For probabilists' Hermite polynomials $H_r$ and $Z_0\sim N(0,1)$, write in $L^2(N(0,1))$
\[
g_{\mathrm{log}}(z)-Eg_{\mathrm{log}}(Z_0)=\sum_{r\ge1}b_rH_r(z),\qquad E\{H_r(Z_0)H_s(Z_0)\}=r!\mathbf1\{r=s\}.
\]
For a stationary unit-variance Gaussian field with spectral probability mass $w$ ($\sum_h w_h=1$), the bivariate Hermite orthogonality identity gives
\[
\widetilde w_h=\sum_{r\ge1}r!b_r^2 (w^{*r})_h.
\]
Indeed, for standard Gaussian $A,B$ with correlation $c$, the generating-function identity
\[\E[e^{uA-u^2/2}e^{vB-v^2/2}]=e^{cuv}\]
gives $\E[H_r(A)H_s(B)]=\mathbf1\{r=s\}r!c^r$.
Since $\sum_r r!b_r^2<\infty$ and $|c|\le1$, the covariance series is uniformly absolutely convergent. Each $c^r$ has the absolutely summable convolution spectrum $w^{*r}$; termwise Fourier integration and then Tonelli yield the formula above. If $h=h_1+\cdots+h_r$, at least one $\|h_j\|\ge\|h\|/r$, and on this part of the sum $w_{h_j}\le Cr^a(1+\|h\|)^{-a}$. For each fixed position $j$, once the other $r-1$ indices are fixed there is exactly one $h_j$ with sum $h$. Summing the other weights gives one; summing over $r$ positions yields
\begin{equation}
(w^{*r})_h\le Cr^{a+1}(1+\|h\|)^{-a}.\label{eq:convolution}
\end{equation}
Let $r_0$ be a fixed positive integer (a derivative order, unrelated to grid size $m$). Every derivative $g_{\mathrm{log}}^{(r_0)}$ is bounded: $g_{\mathrm{log}}'=g_{\mathrm{log}}(1-g_{\mathrm{log}})$ and induction expresses $g_{\mathrm{log}}^{(r_0)}$ as a polynomial in $g_{\mathrm{log}}\in[0,1]$.
With standard Gaussian density $\varphi$, integration by parts $r_0$ times gives
\[
 \E[g_{\mathrm{log}}^{(r_0)}(Z_0)H_j(Z_0)]
 =\E[g_{\mathrm{log}}(Z_0)H_{j+r_0}(Z_0)]=(j+r_0)!b_{j+r_0},
\]
because $(H_j\varphi)^{(r_0)}=(-1)^{r_0} H_{j+r_0}\varphi$ and all boundary terms vanish. Bessel's inequality, which suffices without termwise differentiation of an infinite expansion, yields
\[
 \sum_{r\ge r_0}(r)_{r_0} r!b_r^2
 \le\E[g_{\mathrm{log}}^{(r_0)}(Z_0)^2]<\infty.
\]
For $r\ge2r_0$, $(r)_{r_0}\ge(r/2)^{r_0}$. Choosing integer $r_0>a+1$ proves
$\sum_{r\ge1}r^{a+1}r!b_r^2<\infty$; the finitely many remaining terms are finite by $g_{\mathrm{log}}\in L^2$.
Together with \eqref{eq:convolution} this proves the spectral upper bound with a constant depending on fixed $a$ and the link offset, but not on grid size.
This argument uses the unit-variance normalization and is not a claim about every nonlinear transform or nonstationary latent field.
\end{proof}

\begin{corollary}[A binary dependent periodic population]\label{cor:binary}
Take $\bm x(s)=(1,\sqrt{2}\cos(2\pi s_1))^\top$ and $\lambda(s)=1+\rho\cos(2\pi s_1)$ with $0<|\rho|<1$. Set $c=E_\xi g_{\mathrm{log}}(Z(s))$ (constant by stationarity), and choose fixed prediction $f(s)=c+\gamma_f\cos(2\pi s_1)$ with $|\gamma_f|<\min(c,1-c)$. For the binary responses above, the bounded-response route of the main theorem applies on odd periodic grids as $m,n,k\to\infty$ with the main-paper rate conditions and $k/m\to0$. In particular, the field exponent may be $a=2\upsilon+d$.
\end{corollary}
\begin{proof}
Write
\[
\Delta_i=Y_i-f_i=\underbrace{c-f_i}_{\mu_i}+\underbrace{p_Y(s_i)-c}_{v_i}+\underbrace{Y_i-p_Y(s_i)}_{\zeta_i}.
\]
Then $E(\zeta_i\mid Z)=0$, $\Cov(v,\zeta)=0$, and $\bm\Sigma_\zeta=\diag\{E[p_{Y,i}(1-p_{Y,i})]\}\preceq \mathbf I/4$. Also $0\le f_i\le1$ and $Y_i\in\{0,1\}$, so $|\Delta_i|\le1$. Lemma~\ref{lem:hermite} and uniform aliasing (Lemma~\ref{lem:alias}) verify both structural covariance conditions with $a=2\upsilon+d$. On odd grids $m\ge5$, $\mathbf J=\mathbf I_2$ and the intensity has mean one; the target weights $\ell_{\cdot j}/\lambda$ are analytic periodic functions with uniformly bounded discrete Sobolev norms of each fixed order, by exponential coefficient decay and aliasing. The mean $\mu(s)=-\gamma_f\cos(2\pi s_1)$ lies in $H_1$, so both mean approximation errors vanish for all $k\ge1$.

For (ND) let $\phi_{0i}=\bm\ell_i\zeta_i/\lambda_i$. Conditional on $Z$, $\zeta_i$ are independent bounded centered Bernoulli residuals with $E(\zeta_i^2\mid Z)=p_{Y,i}(1-p_{Y,i})$. Fixed $p=2$ and bounded $\bm\ell_i/\lambda_i$ imply by entrywise conditional variance bounds that
\[
\left\|n\mathbf V_0-\frac1N\sum_i\frac{1-\pi_i}{\lambda_i}\bm\ell_i\bm\ell_i^\top p_{Y,i}(1-p_{Y,i})\right\|_\op=O_{P_\xi}(N^{-1/2}).
\]
The periodic Gaussian field $Z$ has a continuous modification for $\upsilon>0$: for any $0<\beta<\min(2\upsilon,2)$, its increment second moments are bounded by $C_\beta\|s-t\|^\beta$ (summability of $w_h\|h\|^\beta$); Gaussian increments satisfy
$\E|Z(s)-Z(t)|^{2q}\le C_q\|s-t\|^{q\beta}$.
Choose an integer $q$ with $q\beta>d$ and apply the multi-parameter Kolmogorov continuity criterion. Take this continuous version of a single fixed torus field and sample it on every grid; all bounds below are then uniform in $m$.  Thus $\sup_{s\in\mathbb T^d}|Z(s)|<\infty$ almost surely, and for every $\eta'>0$ there is $M_{\eta'}<\infty$ such that $P_\xi(\sup_s|Z(s)|\le M_{\eta'})\ge1-\eta'$. On that event $\delta_{\eta'}\le p_Y(s)\le1-\delta_{\eta'}$ for some $\delta_{\eta'}>0$. Hence
\[
\frac1N\sum_i\frac{p_{Y,i}(1-p_{Y,i})}{\lambda_i}\bm\ell_i\bm\ell_i^\top\succeq\frac{\delta_{\eta'}(1-\delta_{\eta'})}{\lambda_+}\mathbf J^{-1}.
\]
As $\pi_{\max}=O(n/N)\to0$, use the path event with failure allowance $\eta'/2$ and then the covariance error smaller than half its positive eigenvalue bound. Its failure probability tends to zero, establishing (ND) with limiting success probability at least $1-\eta'$.  Finally there is a nonempty rate window: $a+2s_*>2d$, so take $k\asymp n^\theta$ with $1/(a+2s_*)<\theta<1/(2d)$ and $n/N\to0$. Since $\mu\in H_1$, its approximation inequalities hold with zero error for every fixed $t>d/2$. For this chosen $\theta$, choose
$t>\max\{d/2,1/(2\theta)-s_*\}$;
then $nk^{-2(s_*+t)}\to0$ as required by the theorem. We have not dropped that theorem assumption. Also
$k/m\asymp(n/N)^\theta m^{d\theta-1}\to0$
because $d\theta<1/2$, ensuring availability of the protected modes. This proves the corollary in the stated restricted periodic setting.
\end{proof}

The two additional models are restricted to the specified periodic grids and smooth designs. They do not verify the theorem on irregular LUCAS sites, with discontinuous map-based targets, or for an arbitrary nonperiodic Euclidean Mat\'ern field. No additional numerical experiment is reported here.

\section{Additional empirical results}
\subsection{LUCAS methods and numerical conventions}
In this section MCSE is Monte Carlo standard error, RMSE is root mean squared estimation error, and FPC is finite-population correction. The scientific design is always ${\mathbf{X}}=(\bm{1},{\bm{x}}_E,{\bm{x}}_N)$, where ${\bm{x}}_E,{\bm{x}}_N$ are the centered easting and northing coordinates in units of 100 km. To construct auxiliary polynomials, rescale each projected coordinate over the complete frame to $[-1,1]$, obtaining ${\bm{u}},{\bm{v}}$. The ordered map/polynomial dictionary before orthonormalization is
\[
 (\bm{1},{\bm{f}},{\bm{u}},{\bm{v}},{\bm{u}}^2,{\bm{u}}{\bm{v}},{\bm{v}}^2,{\bm{u}}^3,{\bm{u}}^2{\bm{v}},{\bm{u}}{\bm{v}}^2,{\bm{v}}^3,
 {\bm{u}}^4,{\bm{u}}^3{\bm{v}},{\bm{u}}^2{\bm{v}}^2,{\bm{u}}{\bm{v}}^3,{\bm{v}}^4).
\]
Writing the displayed raw dictionary matrix as ${\mathbf{B}}_{\rm raw}={\mathbf{Q}}_{\rm QR}{\mathbf{R}}_{\rm QR}$, with ${\mathbf{Q}}_{\rm QR}^T{\mathbf{Q}}_{\rm QR}={\mathbf{I}}$ and positive diagonal entries of the triangular ${\mathbf{R}}_{\rm QR}$, gives ${\mathbf{B}}=\sqrt N {\mathbf{Q}}_{\rm QR}$, so ${\mathbf{B}}^T{\mathbf{B}}/N={\mathbf{I}}$. Common map uses the first two columns, common map/spatial and target map/spatial the first seven. Common spatial uses the separately orthonormalized ordered dictionary $(\bm{1},{\bm{u}},{\bm{v}},{\bm{u}}^2,{\bm{u}}{\bm{v}},{\bm{v}}^2)$. Common product space uses population QR on
\[
 (\bm{1},{\bm{f}},{\bm{x}}_E,{\bm{x}}_N,{\bm{x}}_E^2,{\bm{x}}_E{\bm{x}}_N,{\bm{x}}_N^2,{\bm{x}}_E{\bm{f}},{\bm{x}}_N{\bm{f}},{\bm{x}}_E^3,{\bm{x}}_E^2{\bm{x}}_N,{\bm{x}}_E{\bm{x}}_N^2,{\bm{x}}_N^3).
\]
These thirteen functions span the products of ${\mathbf{X}}$ with the seven map and quadratic basis functions. Target ${\mathbf{X}}$ uses population QR on ${\mathbf{X}}$. All basis construction uses complete-frame covariates and predictions only.

For an unpenalized common method, $\bar {\bm{b}}_U=N^{-1}\sum_i {\bm{b}}_i$ and
\[
 w_i=\bar {\bm{b}}_U^T({\mathbf{B}}_S^T{\mathbf{B}}_S)^{-1}{\bm{b}}_i,\qquad
 \widehat{\bm{\beta}}_{\rm com}={\bm{\beta}}_f+\sum_{i\in S}w_i{\bm{\ell}}_i\Delta_i.
\]
Unweighted score PPI takes ${\mathbf{B}}=\bm{1}$, giving $w_i=1/n_S$ for $n_S=|S|$. Common map uses the span of $(\bm{1},{\bm{f}})$. Soft common methods apply the same penalized fit and, for BC, the same one-step augmentation as in the main paper to each score response $\ell_{ij}\Delta_i$, with calibration target $\bar {\bm{b}}_U$. Targeted methods instead fit $\Delta_i$ with target ${\mathbf{C}}_B={\mathbf{B}}^T{\mathbf{L}}/N$. Every sixteen-dimensional soft method uses ${\bm{\Gamma}}=n\diag(0_7,{\mathbf{I}}_9)$ in the stated orthonormal coordinates; the nine penalized coordinates arise from degrees three and four. BC denotes the one-step correction. The known-$\pi$ HT comparator uses the true inclusion probabilities and its finite-population variance correction.

For each common-weight method, define its population scalar loading by
\[
 a_i={\bm{b}}_i^T{\mathbf{M}}^{-1}\bar {\bm{b}}_U,\qquad
 a_{+,i}={\bm{b}}_i^T({\mathbf{M}}^{-1}+{\mathbf{M}}^{-1}{\bm{\Lambda}} {\mathbf{M}}^{-1})\bar {\bm{b}}_U,
 \qquad {\mathbf{M}}={\mathbf{A}}+{\bm{\Lambda}}.
\]
Write $a_i^\star=a_i$ without BC and $a_i^\star=a_{+,i}$ with BC. The corresponding population score displacement is
\[
 G_j^{\rm com}=\frac1N\sum_i(\lambda_i a_i^\star-1)\ell_{ij}\Delta_i.
\]
This is the population-normal-equation correction for the score response $\ell_{ij}\Delta_i$ minus its census mean. Thus Pop. gap reports $G_j^{\rm com}$ for common-weight rows, the targeted $G_j$ of the main paper for targeted rows, and zero for the known-probability HT row. In each case it is an oracle diagnostic using the complete residual population and true intensities; it is not the finite-sample bias.

The two prediction variants use response columns ${\bm{Y}}$ and ${\bm{Y}}-{\bm{f}}$, with offsets $0$ and ${\bm{\beta}}_f$, respectively. Both retain the same dictionaries, including their map feature. Under uniform, map and map/east acquisition, $\lambda_i$ is proportional to $1$, $3^{f_i}$ and $3^{f_i}3^{u_{01,i}}$, respectively, normalized to frame mean one; $u_{01,i}$ is easting rescaled to $[0,1]$ and $\pi_i=n\lambda_i/N$.

For the fixed-penalty jackknife, delete each labelled point in turn while retaining the original ${\bm{\Gamma}}$, full-frame basis, calibration target and prediction offset. If $\widehat{\bm{\beta}}_{(-i)}$ is that estimate and $\bar{\bm{\beta}}_{(-)}=n_S^{-1}\sum_{i\in S}\widehat{\bm{\beta}}_{(-i)}$, the covariance is
\[
 \widehat {\mathbf{V}}_{\rm JK}=\frac{n_S-1}{n_S}\sum_{i\in S}
 (\widehat{\bm{\beta}}_{(-i)}-\bar{\bm{\beta}}_{(-)})
 (\widehat{\bm{\beta}}_{(-i)}-\bar{\bm{\beta}}_{(-)})^T.
\]
The implementation uses exact deletion identities. The expected budget $n$, and hence the absolute penalty, are not decreased on deletion. Rank is assessed by a fixed augmented-QR diagonal threshold $10^{-12}$ relative to its largest diagonal magnitude; deletion identities require $\min_i(1-h_i)>10^{-12}$, where $h_i={\bm{b}}_i^T{\mathbf{T}}^{-1}{\bm{b}}_i$. Undefined intervals or regions count as noncoverage in the all-attempt denominator and are also reported separately.

\subsection{Additional tables}

The following tables retain the same populations, label sets, methods and reporting denominators as the main-text experiments. Values are regenerated from the local replay summaries. Coordinate coefficients describe the observed-site frame, not national area-weighted or causal effects. The reference-prediction year difference and image-interpreted reference records remain part of that measurement definition.
\begin{table}[htbp]
\centering\small
\caption{LUCAS north slope, map/east acquisition, expected $n=250$. All 12 predeclared methods; full inclusion-derivative covariance. MCSE is in percentage points. Pop. gap is the method-specific population-normal-equation displacement: $G_j^{\rm com}$ for common-weight methods, $G_j$ for targeted methods, and zero for the known-probability HT oracle (S6.1). It uses the complete response and true acquisition intensities.}\label{tab:north}
\begin{tabular}{lrrrrr}\toprule
Method ($K$) & Coverage (\%) & MCSE & RMSE & Width & Pop. gap\\
\midrule
Unweighted score PPI & 96.25 & 0.42 & 0.04129 & 0.17062 & +0.00641\\
Common map (2) & 95.30 & 0.47 & 0.04612 & 0.18049 & +0.00076\\
Common spatial (6) & 95.45 & 0.47 & 0.03879 & 0.15555 & +0.00467\\
Common map/spatial (7) & 94.55 & 0.51 & 0.04856 & 0.18584 & -0.00036\\
Common product space (13) & 94.45 & 0.51 & 0.04303 & 0.16035 & +0.00013\\
Target ${\mathbf{X}}$ (3) & 96.30 & 0.42 & 0.03634 & 0.14903 & +0.00340\\
Target map/spatial (7) & 95.65 & 0.46 & 0.03729 & 0.14850 & +0.00499\\
Common soft (16) & 93.40 & 0.56 & 0.04834 & 0.18046 & -0.00022\\
Common soft + BC (16) & 93.25 & 0.56 & 0.04865 & 0.18016 & -0.00017\\
Target soft (16) & 95.10 & 0.48 & 0.03727 & 0.14422 & +0.00512\\
Target soft + BC (16) & 95.00 & 0.49 & 0.03755 & 0.14368 & +0.00518\\
Known-$\pi$ HT (oracle) & 95.05 & 0.49 & 0.04401 & 0.16596 & +0.00000\\
\bottomrule\end{tabular}
\end{table}

\begin{table}[htbp]
\centering\small
\caption{Paired east-slope coverage differences under map/east acquisition at $n=250$, in percentage points. Same labels for both methods.}\label{tab:paired}
\begin{tabular}{>{\raggedright\arraybackslash}p{0.66\linewidth}rr}\toprule
Comparison & Difference & MCSE\\
\midrule
Target map/spatial (7) minus Common map/spatial (7) & +5.00 & 0.78\\
Common product space (13) minus Common map/spatial (7) & +3.00 & 0.63\\
Target soft + BC (16) minus Target soft (16) & -0.45 & 0.24\\
Common soft + BC (16) minus Common soft (16) & -0.40 & 0.22\\
Target soft + BC (16) minus Common soft + BC (16) & +5.75 & 0.83\\
\bottomrule\end{tabular}
\end{table}

\begin{table}[htbp]
\centering\small
\caption{East slope under map/east acquisition at $n=100$: small-sample undercoverage and the width cost of the jackknife.}\label{tab:smalln}
\begin{tabular}{lrrrr}\toprule
Method & Linear (\%) & Width & Jackknife (\%) & Width\\
\midrule
Common map (2) & 95.20 & 0.1844 & 95.80 & 0.1897\\
Common map/spatial (7) & 93.10 & 0.2086 & 97.25 & 0.2463\\
Target map/spatial (7) & 91.80 & 0.1772 & 95.35 & 0.2111\\
Target soft + BC (16) & 87.75 & 0.1627 & 95.85 & 0.2232\\
\bottomrule\end{tabular}
\end{table}

\FloatBarrier

The principal common-versus-target comparison keeps the same seven-dimensional basis, label set and absence of regularization. A separate product-basis common comparator shows that targeting is not the only way to enforce score reproduction. Matched sixteen-dimensional soft/BLUP pairs isolate the augmentation step. Baseline ${\mathbf{X}}$-only residual regression is retained, including its poor intercept behavior under some acquisition mechanisms, rather than reporting only slopes for which it looks competitive.

No row is deleted because its covariance is singular. Individual intervals, joint ellipsoids and reporting rates have separate denominators. There are 51 undefined joint-region evaluations, all at $n=100$: 45 for WorldCover predictions and six for the ${\bm{f}}=0$ variant. They are associated with two label configurations, one under uniform acquisition and one under map acquisition; they count as noncoverage in the all-attempt denominator.

\subsection{Fixed-dictionary gap and label-budget sensitivity}\label{sec:gap-budget}
This additional experiment studies the omitted-direction mechanism of the main paper's Proposition 2 at finite budgets. It is separate from the two-dimensional reproduction example and the LUCAS replay. On the one-dimensional periodic frame $s_i=i/N$, $i=0,\ldots,N-1$, take $N=32769$, ${\mathbf{X}}=(\bm{1},{\bm{x}})$ with $x_i=\sqrt2\cos(2\pi s_i)$, and a fixed prediction ${\bm{f}}=0$. For each $k\in\{2,4,8\}$, put
\[
 L_k=4k+3,\quad \delta_k=0.4k^{-3/4},\quad
 \lambda_i=1+\delta_k\sqrt2\cos(2\pi L_ks_i),
\]
\[
 \Delta_i=Z_c\sqrt2\cos\{2\pi(L_k+1)s_i\}
          +Z_s\sqrt2\sin\{2\pi(L_k+1)s_i\}+0.2\varepsilon_i,
 \qquad Z_c,Z_s\stackrel{\rm iid}{\sim}N(0,k^{-3/2}),
\]
where the $\varepsilon_i$ are iid standard Gaussian and independent of $Z_c,Z_s$. Thus $d_0=1$, $s_*=3/4$ and $a=3/2$. The scientific target is always the census coefficient $N^{-1}{\mathbf{X}}^T{\bm{\Delta}}$, since ${\mathbf{J}}={\mathbf{I}}_2$. There are no exclusions or masks. The zero prediction isolates the correction mechanism; it does not demonstrate efficiency gains from an informative prediction model.

For each $k$, generate 16 independent residual populations. In each population draw 64 independent arrays of uniform random variables. Threshold each array at $\pi_i=(n/N)\lambda_i$ for $n\in\{128,512,2048\}$, sharing the resulting labels across methods and coupling the budgets within each population. Different $k$ values use separate draws and different residual and acquisition models. Consequently comparisons across $k$ do not isolate dictionary size on a common population. Randomness uses NumPy PCG64 with master seed 2026100501 and separately indexed field and label streams; the complete specification and per-evaluation outputs are retained with the reproducibility files.

Target hard uses the Fourier cube of radius $k$, with $K=2k+1$. Target soft and target soft with one-step correction use radius $2k$, with $K=4k+1$, and normalized diagonal penalties zero through $k$ and $(|h|/k)^4$ above $k$. Absolute penalties are $n$ times these fixed normalized penalties. There is no tuning or candidate selection. The known-probability HT comparator is an oracle benchmark. Feasible methods use the complete inclusion-derivative covariance without an FPC; HT uses its design variance estimator with the FPC. The same augmented-QR rank tolerance $10^{-12}$ is retained. A joint covariance is required to have positive eigenvalues with minimum-to-maximum ratio greater than $10^{-10}$; no pseudoinverse is used.

The construction gives ${\mathbf{A}}={\mathbf{I}}_K$ and ${\bm{q}}_i={\bm{q}}_{+,i}={\bm{\ell}}_i$ for every unit $i$ in all three targeted methods, so they have the same population displacement even though their fitted estimates and covariances differ. In particular,
\[
 \Cov_\xi({\bm{G}})=\delta_k^2\diag\left\{\frac{0.2^2}{N},
                     \frac{k^{-3/2}}2+\frac{0.2^2}{N}\right\}.
\]
For each realised field and method, compute the oracle linearization covariance ${\mathbf{V}}$ and the Gaussian approximation to marginal slope coverage,
\[
 c_j({\bm{G}},{\mathbf{V}})=\Phi\left(z_{0.975}-\frac{G_j}{\sqrt{V_{jj}}}\right)
          -\Phi\left(-z_{0.975}-\frac{G_j}{\sqrt{V_{jj}}}\right).
\]
Here $j$ is the slope coordinate, $\Phi$ is the standard normal distribution function, and $z_{0.975}=\Phi^{-1}(0.975)$. These oracle diagnostics are not inputs to feasible estimation. The approximation uses ${\mathbf{V}}$ rather than $\widehat {\mathbf{V}}$ and excludes the fitting remainder; it is not exact finite-sample coverage.

Table~\ref{tab:gap-budget} reports all 36 method--budget--radius combinations. Each cell has 1024 field--label evaluations. Overall there are 9216 label sets and 36864 method evaluations, with no fit failures, undefined marginal intervals or undefined joint regions. Coverage uses every attempted evaluation. Parenthesized Monte Carlo standard errors are the standard deviation of the 16 field-level coverage estimates divided by $\sqrt{16}$, in percentage points. This accounts for the nested field/label structure; 1024 evaluations are not treated as independent field draws. Cross-method and cross-budget comparisons are paired within fields.

\small
\begin{longtable}{rrlrrrr}
\caption{Fixed-dictionary sensitivity. Empirical intercept, slope and joint coverage in percent, with field-cluster MCSE in parentheses. Pred. is the field-average oracle Gaussian slope-coverage approximation. Hard, soft and soft + BC are the targeted methods defined in S6.3.}\label{tab:gap-budget}\\
\toprule
$k$ & $n$ & Method & Intercept & Slope & Pred. & Joint\\
\midrule
\endfirsthead
\multicolumn{7}{l}{Table \thetable\ continued}\\
\toprule
$k$ & $n$ & Method & Intercept & Slope & Pred. & Joint\\
\midrule
\endhead
\bottomrule
\endfoot
2 & 128 & Hard & 93.16 (0.65) & 69.04 (4.03) & 70.43 & 73.63 (3.01)\\
2 & 128 & Soft & 92.87 (0.75) & 68.75 (4.09) & 70.43 & 73.73 (2.92)\\
2 & 128 & Soft + BC & 92.87 (0.75) & 68.65 (4.03) & 70.43 & 73.93 (2.94)\\
2 & 128 & HT (oracle) & 94.14 (0.52) & 94.34 (0.68) & 95.00 & 94.92 (0.66)\\
\midrule
2 & 512 & Hard & 93.75 (0.99) & 29.49 (6.96) & 27.73 & 33.89 (6.94)\\
2 & 512 & Soft & 93.85 (0.92) & 29.49 (6.92) & 27.73 & 33.89 (6.94)\\
2 & 512 & Soft + BC & 93.85 (0.92) & 29.49 (6.92) & 27.73 & 33.79 (6.95)\\
2 & 512 & HT (oracle) & 94.04 (0.83) & 95.90 (0.55) & 95.00 & 94.34 (0.55)\\
\midrule
2 & 2048 & Hard & 95.31 (0.62) & 7.23 (4.56) & 6.07 & 8.79 (5.33)\\
2 & 2048 & Soft & 95.31 (0.62) & 7.32 (4.64) & 6.07 & 8.79 (5.33)\\
2 & 2048 & Soft + BC & 95.31 (0.62) & 7.32 (4.64) & 6.07 & 8.79 (5.33)\\
2 & 2048 & HT (oracle) & 95.02 (0.77) & 94.34 (0.74) & 95.00 & 94.63 (0.78)\\
\midrule
4 & 128 & Hard & 93.16 (0.84) & 89.65 (1.04) & 89.83 & 90.14 (1.17)\\
4 & 128 & Soft & 92.87 (0.84) & 89.75 (1.02) & 89.83 & 89.84 (1.06)\\
4 & 128 & Soft + BC & 92.38 (0.89) & 89.26 (1.08) & 89.83 & 88.96 (1.08)\\
4 & 128 & HT (oracle) & 94.82 (0.85) & 95.41 (0.56) & 95.00 & 95.31 (0.55)\\
\midrule
4 & 512 & Hard & 95.61 (0.66) & 74.61 (5.10) & 74.17 & 80.18 (3.86)\\
4 & 512 & Soft & 95.70 (0.61) & 74.61 (5.14) & 74.17 & 80.08 (3.83)\\
4 & 512 & Soft + BC & 95.70 (0.63) & 74.71 (5.11) & 74.17 & 79.98 (3.77)\\
4 & 512 & HT (oracle) & 96.00 (0.55) & 96.58 (0.67) & 95.00 & 95.80 (0.76)\\
\midrule
4 & 2048 & Hard & 95.70 (0.73) & 44.63 (10.35) & 42.22 & 49.51 (10.18)\\
4 & 2048 & Soft & 95.80 (0.74) & 44.73 (10.38) & 42.22 & 49.41 (10.18)\\
4 & 2048 & Soft + BC & 95.80 (0.74) & 44.63 (10.35) & 42.22 & 49.12 (10.14)\\
4 & 2048 & HT (oracle) & 95.41 (0.83) & 96.00 (0.47) & 95.00 & 95.51 (0.81)\\
\midrule
8 & 128 & Hard & 91.89 (1.05) & 92.09 (0.73) & 94.21 & 90.14 (0.58)\\
8 & 128 & Soft & 91.02 (1.26) & 90.92 (0.88) & 94.20 & 88.28 (0.68)\\
8 & 128 & Soft + BC & 90.23 (1.24) & 90.43 (0.85) & 94.20 & 87.70 (0.75)\\
8 & 128 & HT (oracle) & 94.53 (0.75) & 96.19 (0.51) & 95.00 & 96.09 (0.61)\\
\midrule
8 & 512 & Hard & 95.02 (0.72) & 90.23 (1.26) & 91.71 & 91.50 (1.05)\\
8 & 512 & Soft & 94.82 (0.74) & 89.55 (1.33) & 91.71 & 90.92 (1.11)\\
8 & 512 & Soft + BC & 94.73 (0.74) & 89.06 (1.35) & 91.71 & 90.72 (1.14)\\
8 & 512 & HT (oracle) & 94.92 (0.63) & 94.73 (0.73) & 95.00 & 93.95 (0.65)\\
\midrule
8 & 2048 & Hard & 96.39 (0.58) & 82.03 (3.29) & 81.05 & 86.91 (2.46)\\
8 & 2048 & Soft & 96.48 (0.56) & 82.23 (3.27) & 81.05 & 86.82 (2.51)\\
8 & 2048 & Soft + BC & 96.58 (0.52) & 82.13 (3.33) & 81.05 & 86.82 (2.51)\\
8 & 2048 & HT (oracle) & 95.21 (0.65) & 96.00 (0.57) & 95.00 & 95.70 (0.56)\\
\end{longtable}
\normalsize
\clearpage

At fixed $k=2$, targeted soft + BC slope coverage falls from 68.65\% to 29.49\% and 7.32\% as $n$ increases; the oracle approximations are 70.43\%, 27.73\% and 6.07\%. The gap remains while sampling uncertainty decreases. HT slope coverage is 94.34\%, 95.90\% and 94.34\% at the same budgets. Augmentation alone does not remove a direction omitted by every dictionary in this construction. At $k=8,n=128$, targeted soft + BC slope coverage is 90.43\% compared with the oracle approximation of 94.20\%; its mean ratio $\widehat V_{jj}/V_{jj}$ is 0.904, illustrating that finite-sample fitting and variance estimation also matter.

This fixed-$N$, fixed-dictionary budget comparison is not a growing-sequence test of the recovery theorem. Its sparse spectrum does not test the operator-versus-energy dimensional gain either. Some cells have substantial uncertainty across the 16 fields. Sampling fractions of 0.0039--0.0625 also leave room for differences between the non-FPC estimate and oracle design covariance.

\subsection{Equal-energy directional-covariance experiment}\label{sec:directional-covariance}\hypertarget{directional-covariance}{}
We study whether total residual energy or its orientation relative to the coefficient-specific calibration defect better explains conditional coverage in a finite-frame experiment. Take the periodic $129\times129$ grid, with $N=16{,}641$, and fix the scientific design $\mathbf X=(\bm1,\bm x)$, where $x(s)=\sqrt2\cos(2\pi s_1)$. The prediction is fixed at $\bm f=\bm0$, so the simulated response equals $\bm\Delta$. Let the acquisition intensity be $\lambda(s)=1+\delta z(s)$, where $z(s)=\sqrt2\cos(2\pi Ls_1)$, $k=2$, $L=11$ and $\delta=0.4k^{-5/4}$. The protected targeted Fourier dictionary is the cube of radius $k$ (25 columns). The targeted soft dictionary is the cube of radius $2k=4$ (81 columns), with normalized penalties zero through degree $k$ and $(|h|/k)^4$ above $k$; $\bm\Gamma=n\bm\Lambda$ is fixed at the expected label count $n$. We use $n\in\{250,500,1000\}$ and the independent Bernoulli label-acquisition design of the main paper.

Write $\varphi_{h,c}(s)=\sqrt2\cos(2\pi hs_1)$ and $\varphi_{h,s}(s)=\sqrt2\sin(2\pi hs_1)$. The scalar $M_{\mathrm{pair}}\in\{1,4,16\}$ counts structural Fourier pairs and is distinct from the population matrix $\mathbf M=\mathbf A+\bm\Lambda$. With iid standard Gaussian coefficients $A_j,B_j$, define
\begin{align*}
 v_{M_{\mathrm{pair}}}^A(s)&=\sqrt{\tau/M_{\mathrm{pair}}}\sum_{j=0}^{M_{\mathrm{pair}}-1}
 \{A_j\varphi_{L+1+2j,c}(s)+B_j\varphi_{L+1+2j,s}(s)\},\\
 v_{M_{\mathrm{pair}}}^O(s)&=\sqrt{\tau/M_{\mathrm{pair}}}\sum_{j=0}^{M_{\mathrm{pair}}-1}
 \{A_j\varphi_{L+3+2j,c}(s)+B_j\varphi_{L+3+2j,s}(s)\}.
\end{align*}
For each orientation $o\in\{A,O\}$, put $\Delta_{M_{\mathrm{pair}}}^o(s)=v_{M_{\mathrm{pair}}}^o(s)+\sigma\varepsilon(s)$, using the same nugget realization in both orientations. Here $\tau=0.04$, $\sigma=0.4$ and $\varepsilon(s)$ are independent standard Gaussian nugget errors, independent of the structural coefficients. Gaussian coefficients, nugget realizations and uniform random numbers generating Bernoulli labels are coupled between the two covariance scenarios and the two targeted procedures.

\paragraph{Analytic comparison.}
For each $M_{\mathrm{pair}}$, the two structural covariance matrices have the same eigenvalue multiset. For $o\in\{A,O\}$,
\[
 \E_\xi\|\bm v_{M_{\mathrm{pair}}}^o\|_N^2=2\tau=0.08,\qquad
 \E_\xi\|\bm\Delta_{M_{\mathrm{pair}}}^o\|_N^2=2\tau+\sigma^2=0.24.
\]
This is equality of expected energy, not a claim of identical total residual mean squared error for every realization. Fourier orthogonality and $L=11>2(2k)$ give $\mathbf A=\mathbf B^T\mathbf D_\lambda\mathbf B/N=\mathbf I$, $\bm q_i=\bm q_{+,i}=\bm\ell_i$, and defect columns $\bm d_{\rm intercept}=\delta\bm z$, $\bm d_{\rm slope}=\delta\bm z\bm x$. Since $zx=(\varphi_{L-1,c}+\varphi_{L+1,c})/\sqrt2$,
\[
 \Var_\xi(G_{\rm slope})=
 \begin{cases}
  \delta^2\tau/(2M_{\mathrm{pair}})+\sigma^2\delta^2/N,&\text{aligned},\\
  \sigma^2\delta^2/N,&\text{orthogonal}.
 \end{cases}
\]
Thus the slope component of $\mathbf D^T\bm\Sigma_\Delta\mathbf D$ differs even at the same $M_{\mathrm{pair}}$ and covariance spectrum. Changing $M_{\mathrm{pair}}$ also changes spectral concentration and is distinct from comparing orientations at a fixed $M_{\mathrm{pair}}$.

\paragraph{Monte Carlo design and results.}
For each condition, we generated 128 independent residual populations, with 128 Bernoulli label sets per population and expected label counts $n=250,500,1000$. The same latent Gaussian coefficient draws, nugget vector and label-generating uniform random numbers were coupled across covariance orientations and methods; label draws were also coupled across budgets. The master seed was 2026100901. All attempted fits remained in the coverage denominator. Confidence intervals used the complete inclusion-derivative covariance of the implemented estimator, retaining the derivative of the fitted augmentation for soft calibration. The paired Monte Carlo standard error was computed from the 128 field-level conditional coverage differences, not by treating all 16,384 label sets as independent residual populations. The supplied execution log records no failed or undefined fits. A uniform-acquisition negative control used the same residual populations, soft calibration with correction and budgets $n=500,1000$, with a separate label-randomness stream. Table~\ref{tab:directional} retains all 18 primary conditions and Table~\ref{tab:directional-null} all six controls. In particular, the orthogonal scenario still undercovers at $n=250$. This fixed-$N$ experiment illustrates covariance geometry; it is not a growing-sequence test of Theorem 1.

\begin{table}[!htbp]
\centering\small
\caption{Equal-energy directional-covariance experiment under nonuniform acquisition. Nominal 95\% slope intervals; 128 fields and 128 label sets per field. Coverage is in percent; difference (orthogonal minus aligned) and paired field-cluster MCSE are in percentage points.}\label{tab:directional}
\begin{tabular}{lrrrrrr}\toprule
Method & $n$ & $M_{\mathrm{pair}}$ & Aligned & Orthogonal & Difference & MCSE\\
\midrule
Target hard & 250 & 1 & 87.89 & 93.48 & 5.58 & 0.54\\
Target hard & 250 & 4 & 91.72 & 93.33 & 1.61 & 0.27\\
Target hard & 250 & 16 & 92.94 & 93.53 & 0.59 & 0.20\\
Target hard & 500 & 1 & 84.41 & 94.21 & 9.80 & 1.04\\
Target hard & 500 & 4 & 91.52 & 94.33 & 2.81 & 0.41\\
Target hard & 500 & 16 & 93.44 & 94.50 & 1.06 & 0.21\\
Target hard & 1000 & 1 & 75.73 & 95.08 & 19.35 & 1.81\\
Target hard & 1000 & 4 & 89.33 & 95.25 & 5.91 & 0.71\\
Target hard & 1000 & 16 & 93.70 & 95.23 & 1.53 & 0.26\\
\midrule
Target soft + BC & 250 & 1 & 86.04 & 92.14 & 6.10 & 0.56\\
Target soft + BC & 250 & 4 & 90.26 & 91.78 & 1.53 & 0.28\\
Target soft + BC & 250 & 16 & 91.63 & 92.14 & 0.51 & 0.22\\
Target soft + BC & 500 & 1 & 83.34 & 93.74 & 10.41 & 1.06\\
Target soft + BC & 500 & 4 & 90.87 & 93.91 & 3.04 & 0.41\\
Target soft + BC & 500 & 16 & 93.01 & 93.99 & 0.99 & 0.23\\
Target soft + BC & 1000 & 1 & 75.39 & 94.85 & 19.46 & 1.83\\
Target soft + BC & 1000 & 4 & 89.06 & 95.00 & 5.94 & 0.72\\
Target soft + BC & 1000 & 16 & 93.48 & 95.04 & 1.56 & 0.26\\
\bottomrule\end{tabular}
\end{table}

\begin{table}[!htbp]
\centering\small
\caption{Uniform-acquisition negative control, targeted soft calibration with one-step correction. Nominal 95\% slope intervals; coverage is in percent, and orthogonal-minus-aligned differences and paired field-cluster MCSEs are in percentage points.}\label{tab:directional-null}
\begin{tabular}{rrrrrr}\toprule
$n$ & $M_{\mathrm{pair}}$ & Aligned & Orthogonal & Difference & MCSE\\
\midrule
500 & 1 & 94.23 & 94.17 & -0.05 & 0.17\\
500 & 4 & 93.95 & 94.31 & 0.37 & 0.20\\
500 & 16 & 94.42 & 94.24 & -0.17 & 0.18\\
1000 & 1 & 94.96 & 94.90 & -0.07 & 0.15\\
1000 & 4 & 95.18 & 95.50 & 0.32 & 0.16\\
1000 & 16 & 95.22 & 95.19 & -0.03 & 0.15\\
\bottomrule\end{tabular}
\end{table}
\FloatBarrier

\paragraph{Reproducibility record.}
The standalone script is publicly available as \texttt{code/directional.py} at \url{https://github.com/ShinShirota/spatial-ppi}. The saved reference summaries are available in the data bundle attached to \href{https://github.com/ShinShirota/spatial-ppi/releases/tag/v0.1.0-candidate}{release v0.1.0-candidate}; the original and replay records are retained locally. The supplied replay record reports Linux x86\_64, Python 3.13.5, NumPy 2.3.5, SciPy 1.17.0, pandas 2.2.3 and threadpoolctl 3.6.0. It records a full same-script, same-seed replay with no failures: 589,824 primary and 196,608 uniform-control method--field--label--budget--scenario evaluations. These counts are not independent field draws. Normalizing method names and column order gives zero numerical difference between the supplied original and replay CSVs (18 primary summary rows, 4,608 field-level rows and six control rows); the script and replay file hashes match their supplied audit record. Reaggregation of all 18 primary paired contrasts from the stored field-level data agrees within $5.1\times10^{-9}$ percentage points, attributable to CSV rounding.

The archived CSVs also retain noncentral-Gaussian oracle diagnostics; these use the complete population and true acquisition probabilities and are not inputs to the feasible estimators. Individual label-set outputs were not saved, and the control has only aggregate saved output; its field-cluster MCSE was therefore not independently reaggregated from archived field-level data. The seed and script permit regeneration. The same-script replay and saved-output checks do not constitute validation by a separate implementation and do not repeat the earlier finite-difference or exact-expansion checks. The separate release data bundle includes the retained summaries; no data-archive DOI is claimed.
\FloatBarrier

\subsection{Cross-condition LUCAS comparison}\label{sec:lucas-crosscondition}
Table~\ref{tab:lucas-crosscondition} summarizes the same seven-dimensional common and targeted corrections across all nine acquisition--budget settings for the WorldCover prediction variant. Each cell uses 2,000 shared Bernoulli label sets. Differences are targeted minus common coverage; paired MCSEs are computed from label-wise differences in coverage indicators, conditional on the fixed observed-site frame. The table is a post hoc reaggregation of preserved outputs, not new label simulations. Across the 18 slope--condition comparisons, targeted calibration reduces RMSE and mean interval width in every case, but increases coverage in 10 and brings it closer to 95\% in six; these counts are descriptive, not independent replications. Full signed gaps, RMSEs, widths and the separate ${\bm{f}}=0$ results remain in the archived machine-readable summaries.

\begin{table}[!htbp]
\centering\small
\hypertarget{lucas-crosscondition}{}
\caption{LUCAS--WorldCover cross-condition comparison: 95\% marginal slope coverage. C/T denote common and targeted map/spatial corrections with the same seven-dimensional basis. Coverage is in percent; differences (T minus C) and paired-label MCSEs are in percentage points.}\label{tab:lucas-crosscondition}
\begin{tabular}{lrcccc}\toprule
Acquisition & Expected & \multicolumn{2}{c}{East slope} & \multicolumn{2}{c}{North slope}\\
 & $n$ & C/T (\%) & $\Delta$ (MCSE), pp & C/T (\%) & $\Delta$ (MCSE), pp\\
\midrule
uniform & 100 & 94.25 / 93.90 & $-0.35$ (0.47) & 94.95 / 93.55 & $-1.40$ (0.46)\\
uniform & 250 & 96.00 / 95.55 & $-0.45$ (0.36) & 94.70 / 94.80 & $+0.10$ (0.44)\\
uniform & 500 & 96.65 / 97.05 & $+0.40$ (0.33) & 97.15 / 96.90 & $-0.25$ (0.28)\\
map & 100 & 94.55 / 93.50 & $-1.05$ (0.59) & 94.35 / 93.15 & $-1.20$ (0.58)\\
map & 250 & 95.20 / 95.80 & $+0.60$ (0.50) & 94.25 / 95.40 & $+1.15$ (0.55)\\
map & 500 & 95.35 / 97.25 & $+1.90$ (0.47) & 96.35 / 96.95 & $+0.60$ (0.48)\\
map/east & 100 & 93.10 / 91.80 & $-1.30$ (0.70) & 93.75 / 92.95 & $-0.80$ (0.54)\\
map/east & 250 & 89.90 / 94.90 & $+5.00$ (0.78) & 94.55 / 95.65 & $+1.10$ (0.58)\\
map/east & 500 & 86.35 / 96.10 & $+9.75$ (0.82) & 96.05 / 96.70 & $+0.65$ (0.50)\\
\bottomrule\end{tabular}
\end{table}
\FloatBarrier

\subsection{Growing-radius pilot along a sufficient-rate exponent window}\label{sec:growing-pilot}
To complement the fixed-dictionary failure mechanism in S6.3, we conducted a small exploratory growing-radius pilot using the same Gaussian triangular-array construction, varying the expected label count and frame size. We set $(n,k,N)=(128,8,8193),(512,14,65537),(2048,25,524289)$, approximately $k\asymp n^{0.42}$ with decreasing $n/N$. The scientific design, target, $f=0$, structural variance $\tau_k=k^{-3/2}$, intensity modulation $\delta_k=0.4k^{-3/4}$, frequency $L_k=4k+3$ and nugget standard deviation $0.2$ follow S6.3. At each level, 16 independent residual populations each received 64 independent Bernoulli label samples, shared across methods. Hard calibration uses radius $k$; soft calibration with one-step correction uses radius $2k$, protects modes through $k$ and penalizes higher modes by $(|h|/k)^4$, with $\bm\Gamma=n\bm\Lambda$ fixed. Across 48 populations and 3,072 label samples, all 9,216 method evaluations were retained, with zero fit failures and zero undefined marginal or joint regions.

Table~\ref{tab:growing-pilot} reports observed nominal 95\% slope-interval coverage. Field-cluster Monte Carlo standard errors use the 16 population-specific coverage rates; the 1,024 label evaluations per level and method are not treated as independent populations.

\begin{table}[!htbp]\centering
\caption{Exploratory Gaussian growing-radius pilot. Observed nominal 95\% slope coverage (percent) across 16 independent finite populations and 64 Bernoulli label acquisitions per population at each level. Parentheses give population-cluster Monte Carlo standard errors in percentage points. HT uses true inclusion probabilities and an FPC; targeted methods use feasible inclusion-derivative covariance without an FPC. Populations differ across levels.}\label{tab:growing-pilot}
\begin{tabular}{rrrrrr}\toprule
Expected $n$ & Radius $k$ & Frame $N$ & Target hard & Target soft + BC & Known-$\pi$ HT\\\midrule
128 & 8 & 8,193 & 90.72 (0.99) & 88.77 (1.03) & 96.00 (0.68)\\
512 & 14 & 65,537 & 93.16 (0.77) & 92.87 (0.74) & 95.51 (0.42)\\
2,048 & 25 & 524,289 & 93.55 (0.74) & 93.75 (0.74) & 94.82 (0.66)\\
\bottomrule\end{tabular}
\end{table}

Targeted coverage increases across these levels, but neither method reaches nominal 95\% at the largest level. For hard calibration, the across-run mean of $\sqrt n\,|G_{\mathrm{slope}}|$ decreases from $0.10747$ to $0.07849$, that of $\sqrt n\,\|\bm{\mathcal R}\|$ from $0.17453$ to $0.04650$, and that of $n\|\widehat{\mathbf V}-\mathbf V\|_{\mathrm{op}}$ from $0.03768$ to $0.00365$. The corresponding remainder means for soft + BC decrease from $0.18525$ to $0.04839$. At $n=2048$, $k^2/n=0.3052$ and $nk^{-3}=0.1311$ remain nonnegligible. Increasing $k$ also changes acquisition modulation and structural field variance, and levels use independent populations, so the trend cannot be attributed solely to dictionary size. The pilot neither proves asymptotic 95\% conditional coverage nor tests the operator-norm versus total-energy dimensional gain.

The exploratory specification (master seed 2026101302), code, generated populations, label indices, all evaluation records and environment logs are publicly available in the growing-radius pilot archive attached to \href{https://github.com/ShinShirota/spatial-ppi/releases/tag/v0.2.0-candidate}{release v0.2.0-candidate}; private machine-path prefixes in distributed metadata are normalized. Table~\ref{tab:growing-pilot} was reconciled against the saved evaluation records, without new simulation draws or an independent implementation rerun.
\FloatBarrier

\section{Computational verification}
The LUCAS replay evaluates 12 methods on all 18,000 saved label sets, giving 216,000 method--label-set evaluations. Each evaluation returns both prediction variants, for a total of 432,000 three-coefficient estimates. Response columns are processed together, and matrix factorizations can be shared across methods, so these counts are not counts of separate factorizations. The constructed experiment evaluates five methods on 3,000 shared label sets, giving 15,000 method--label-set evaluations, each returning estimates for 32 field realizations. Both replays preserve all attempted evaluations and use the original specifications; neither records a fit exception. Undefined joint regions are recorded separately.

Against the preserved outputs, 199 LUCAS and 161 constructed-experiment array or summary-column comparisons pass absolute tolerance $10^{-8}$ plus relative tolerance $10^{-10}$, with exact discrete values and missingness patterns. Agreement across computing environments is not asserted to be bitwise equality.

For the LUCAS replay, a separate SVD factorization checks 297 coefficient and covariance cases. The maximum differences are $4.78\times10^{-15}$ for estimates and $6.25\times10^{-17}$ for covariance. Thirty-three inclusion finite-difference checks have maximum error $5.18\times10^{-11}$, and 33 explicit one-point deletion checks have maximum error $3.63\times10^{-15}$. Sixty synthetic finite-remainder checks have maximum error $1.38\times10^{-16}$. All satisfy the original tolerances. The 1,242 marginal summary rows are also checked by separate reaggregation.

Additional deterministic checks verify the scaling $\widehat{\bm{\phi}}=n\widehat{\bm{\psi}}$, covariance normalization, protected prediction-shift invariance, scientific-covariate recoding, and invariance of feasible fits to changes in unobserved responses. These are finite implementation checks, not empirical evidence that every assumption of the spectral-basis theorem holds on the LUCAS frame. No same-label candidate selector is evaluated in these experiments. The S6.3 specification was fixed before its run, informed by the existing manuscript. General population normal-equation calculations and direct sample solves were checked for all three targeted methods on the first field and label set; the largest discrepancy across estimates, gaps, linear terms and covariances was $2.64\times10^{-16}$, below the fixed absolute tolerance $10^{-10}$. This is a finite implementation check, not a separate proof or a full independent replay.